\documentclass[11pt]{article}

\usepackage[margin=1in]{geometry}

\usepackage[linesnumbered,ruled,vlined,commentsnumbered]{algorithm2e}

\usepackage[utf8]{inputenc} %
\usepackage[T1]{fontenc}    %
\usepackage{url}            %
\usepackage{booktabs}       %
\usepackage{amsfonts}       %
\usepackage{nicefrac}       %
\usepackage{microtype}      %

\usepackage[usenames,dvipsnames]{xcolor}
\definecolor{Gred}{RGB}{219, 50, 54}
\definecolor{Ggreen}{RGB}{60, 186, 84}
\definecolor{Gblue}{RGB}{72, 133, 237}
\definecolor{Gyellow}{RGB}{247, 178, 16}
\definecolor{ToCgreen}{RGB}{0, 128, 0}
\definecolor{myGold}{RGB}{231,141,20}
\definecolor{myBlue}{rgb}{0.19,0.41,.65}
\definecolor{myPurple}{RGB}{175,0,124}

\usepackage{graphicx}
\usepackage{hyperref}
\hypersetup{
			colorlinks=true,
	citecolor=ToCgreen,
	linkcolor=Sepia,
	filecolor=Gred,
	urlcolor=Gred
	}

\usepackage{apptools}
\usepackage{chngcntr}
\usepackage{commath}
\usepackage{amssymb,amsmath,amsthm}

\title{Learning Mixtures of Gaussians Using the DDPM Objective}
\author{
    Kulin Shah\thanks{Email: \texttt{kulinshah@utexas.edu}. Supported by the NSF AI Institute for Foundations of Machine Learning (IFML). } \\
    UT Austin
        \and
    Sitan Chen\thanks{Email: \texttt{sitan@seas.harvard.edu}. Supported by NSF Award 2103300.} \\
    UC Berkeley
        \and
    Adam Klivans\thanks{Email: \texttt{klivans@cs.utexas.edu}.  Supported by the NSF AI Institute for Foundations of Machine Learning (IFML).} \\
    UT Austin
}

\newcommand{\citep}{\cite}
\renewcommand{\epsilon}{\varepsilon}

\usepackage{mathtools}
\mathtoolsset{showonlyrefs}

\begin{document}

\newcommand{\dtp}[2]{\langle #1, #2 \rangle}
\newcommand{\mcT}{\mathcal{T}}
\newcommand{\bt}{\Bar{t}}
\newcommand{\brdg}[1]{z_{#1}^x}
\newcommand{\rb}[1]{\left( #1 \right)}
\newcommand{\h}[1]{\hat{#1}}

\newcommand{\muhp}{\hat{\mu}^{\perp}}  %
\newcommand{\wnorm}[1]{\| #1 \|} %
\newcommand{\sgnorm}[1]{\wnorm{#1}_{\psi_2}}

\makeatletter
\newcommand*{\rom}[1]{\expandafter\@slowromancap\romannumeral #1@}
\makeatother

\newcommand{\mc}[1]{\mathcal{#1}}
\newcommand{\mbb}[1]{\mathbb{#1}}
\newcommand{\mbf}[1]{\mathbf{#1}}
\newcommand{\PM}[1]{\mbb{#1}}
\newcommand{\mean}[3]{\mu_{#1, #2}^{#3}}

\newtheorem{theorem}{Theorem}
\newtheorem{lemma}[theorem]{Lemma}
\newtheorem{definition}[theorem]{Definition}
\newtheorem{fact}[theorem]{Fact}
\newtheorem{assumption}[theorem]{Assumption}
\newtheorem{remark}[theorem]{Remark}

\newcommand{\kulin}[1]{\textcolor{red}{#1}}
\newcommand{\TODO}[1]{\textcolor{blue}{TODO: #1}}

\newcommand{\calN}{\mathcal{N}}
\newcommand{\R}{\mathbb{R}}
\newcommand{\Id}{\mathrm{Id}}

\maketitle

\begin{abstract}

Recent works have shown that diffusion models can learn essentially any distribution provided one can perform score estimation.
Yet it remains poorly understood under what settings score estimation is possible, let alone when practical gradient-based algorithms for this task can provably succeed. 

In this work, we give the first provably efficient results along these lines for one of the most fundamental distribution families, Gaussian mixture models.
We prove that gradient descent on the denoising diffusion probabilistic model (DDPM) objective can efficiently recover the ground truth parameters of the mixture model in the following two settings:
\begin{enumerate}
    \item We show gradient descent with random initialization learns mixtures of two spherical Gaussians in $d$ dimensions with $1/\text{poly}(d)$-separated centers.
    \item We show gradient descent with a warm start learns mixtures of $K$ spherical Gaussians with $\Omega(\sqrt{\log(\min(K,d))})$-separated centers.
\end{enumerate}
A key ingredient in our proofs is a new connection between score-based methods and two other approaches to distribution learning, the EM algorithm and spectral methods.

\end{abstract}

\tableofcontents

\section{Introduction}

In recent years diffusion models~\cite{song2020score,sohl_thermodynamics,ncsn} have emerged as a powerful framework for generative modeling and now form the backbone of notable image generation systems like DALL$\cdot$E 2~\cite{dalle2}, Imagen~\cite{imagen}, and Stable Diffusion~\cite{latent_diffusion}. At the heart of this framework is a reduction from \emph{distribution learning} to \emph{denoising} or \emph{score estimation}. That is, in order to generate new samples from a data distribution $q$ given a collection of independent samples, it suffices to learn the score function, i.e., the gradient of the log-density of the data distribution when convolved with varying levels of noise (see Section~\ref{subsec:preliminaries}).  A popular and well-studied objective for score matching is the \emph{denoising diffusion probabilistic model (DDPM) objective} due to \cite{ho2020denoising}. Optimizing this objective amounts to solving the following type of problem: given a noisy observation $\widetilde{x}$ of a sample $x$ from $q$, estimate the mean of the posterior distribution over $x$. 

While a number of theoretical works \cite{debetal2021scorebased,BloMroRak22genmodel,chen2022improved, DeB22diffusion, leelutan22sgmpoly, liu2022let, Pid22sgm, WibYan22sgm, chen2023sampling, chen2023restoration, lee2023convergence,chen2023probability,li2023towards,benton2023error} have established rigorous convergence guarantees for diffusion models under mild assumptions on the data distribution, these works assume the existence of an oracle for score estimation and leave open whether one can actually provably implement such an oracle for interesting families of data distributions. In practice, the algorithm of choice for score estimation is simply to train a student network via gradient descent (GD) to fit a set of examples $(x,\widetilde{x})$. We thus ask:

\begin{center}
    \emph{Are there natural data distributions under which GD provably achieves accurate score estimation?}
\end{center}

In this work, we consider the setting where $q$ is given by a \emph{mixture of Gaussians}. Concretely, we assume that there exist centers $\mu_1^*,\ldots,\mu_K^*\in\R^d$ such that
\begin{equation}
    q = \frac{1}{K}\sum^K_{i=1} \calN(\mu_i^*, \Id)\,.
\end{equation}
We answer the above question in the affirmative for this class of distributions:

\begin{theorem}[Informal, see Theorems~\ref{thm:mo2g-const-sep} and~\ref{thm:2-mog-small-sep}]\label{thm:informal_2} 
    Gradient descent on the DDPM objective with random initialization efficiently learns the parameters of an unknown mixture of two spherical Gaussians with $1/\text{poly}(d)$-separated centers.
\end{theorem} 

\begin{theorem}[Informal, see Theorem \ref{thm:mog-k-main}]\label{thm:informal_K}
    When there is a warm start of the centers, gradient descent on the DDPM objective efficiently learns the parameters an unknown mixture of $K$ spherical Gaussians with $\Omega(\sqrt{\log(\min(K,d))})$-separated centers.  
\end{theorem}

\noindent The DDPM objective is described in Algorithm \ref{alg:denoise}.  The term ``efficiently'' above means that both the running time and sample complexity of our algorithm is polynomial in the dimension $d$, the inverse accuracy $1/\epsilon$, and the number of components $K$. In the informal discussion, we often work with population gradients for simplicity, but in our proofs we show that empirical estimates of the gradient suffice (full details can be found in the Appendix). 

\begin{algorithm}
\DontPrintSemicolon
\caption{\textsc{GMMDenoiser}($t, \{ \mu_i^{(0)} \}_{i=1}^K, H$)}
\label{alg:denoise}
	\KwIn{Noise scale $t$, initialization $\{ \mu_i^{(0)} \}_{i=1}^K$, number of gradient descent steps $H$}
    Initialize the parameters for the score estimate at $\theta_t^{(0)} = \{ \mu_{i,t}^{(0)} \}_{i=1}^K$ (see Eq.~\eqref{eq:score-function-definition} for how the estimate $s_{\theta}$ depends on the parameters $\theta$, and Eq.~\eqref{eq:mog-k-pdf-time-t} for the definition of $\mu^{(0)}_{i,t}$)\;
    Run gradient descent on the DDPM objective $L_t(s_{\theta_t})$ for $H$ steps where
	\begin{equation}
	    L_t(s_{\theta_t}) = \mbb{E} \Big[ \Big\| s_{\theta_t}(X_t) + \frac{Z_t}{ \sqrt{1 - \exp(-2t)} } \Big\|^2 \Big]\,,
	\end{equation}\;\vspace{-\baselineskip}
    \Return  $\theta_t^{(H)} = \{ \mu_{i, t}^{(H)} \}_{i=1}^K$ where $\theta_t^{(H)}$ denotes the parameters after $H$ steps of GD. 
\end{algorithm}

\noindent We refer to Section~\ref{subsec:preliminaries} for a formal description of the quantities used in Algorithm \ref{alg:denoise}. Note that there are by now a host of different algorithms for provably learning mixtures of Gaussians (see Section~\ref{sec:related}).
For instance, it is already known that expectation-maximization (EM) achieves the quantitative guarantees of Theorems~\ref{thm:informal_2} and \ref{thm:informal_K} \cite{daskalakis2017ten,xu2016global,kwon2020algorithm,segol2021improved}, and in fact even stronger guarantees are known via the method of moments. Unlike works based on the method of moments however, our algorithm is practical. And unlike works based on EM, it is based on an approach which is empirically successful for a wide range of realistic data distributions. Furthermore, as we discuss in Section~\ref{sec:overview}, the analysis of Algorithm~\ref{alg:denoise} leverages an intriguing and, to our knowledge, novel connection from score estimation to EM, as well as to another notable approach for learning mixture models, namely spectral methods. Roughly speaking, at large noise levels, the gradient updates in Algorithm~\ref{alg:denoise} are essentially performing a type of power iteration, while at small noise levels, the gradient updates are performing the ``M'' step in the EM algorithm. %

\subsection{Related work}
\label{sec:related}

\paragraph{Theory for diffusion models.} A number of works have given convergence guarantees for DDPMs and variants~\cite{debetal2021scorebased,BloMroRak22genmodel,chen2022improved, DeB22diffusion, leelutan22sgmpoly, liu2022let, Pid22sgm, WibYan22sgm, chen2023sampling, chen2023restoration, lee2023convergence,li2023towards,benton2023error,chen2023probability}. These results show that, given an oracle for accurate score estimation, diffusion models can learn essentially any distribution over $\R^d$ (e.g. \cite{chen2023sampling,lee2023convergence,chen2022improved} show this for arbitrary compactly supported distributions). Additionally, two recent works \cite{el2022sampling,montanari2023posterior} have used Eldan's stochastic localization~\cite{eldan2013thin,eldan2020taming}, which is a reparametrization in time and space of the reverse SDE for DDPMs, to give sampling algorithms for certain distributions arising in statistical physics. 
As we discuss next, these works are end-to-end in that they also give provable algorithms for score estimation via approximate message passing, though the statistical task they address is not distribution learning. 

\paragraph{Provable score estimation.} There is a rich literature giving Bayes-optimal algorithms for various natural denoising problems via methods inspired by statistical physics, like approximate message passing (AMP) (e.g.~\cite{montanari2021estimation,celentano2021local,bayati2011dynamics,kabashima2003cdma,donoho2009message,donoho2010message}) and natural gradient descent (NGD) on the TAP free energy~\cite{celentano2021local,el2022sampling,celentano2022sudakov}. The abovementioned works \cite{el2022sampling,montanari2023posterior} (see also \cite{celentano2022sudakov}) build on these techniques to give algorithms for the denoising problems that arise in their implementation of stochastic localization. These works on denoising via AMP or NGD are themselves part of a broader literature on variational inference, a suitable literature review would be beyond the scope of this work, see e.g.~\cite{blei2017variational,wainwright2008graphical,mezard2009information}. 

We are not aware of any provable algorithms for score estimation explicitly in the context of \emph{distribution learning}. That said, it may be possible to extract a distribution learning result from~\cite{el2022sampling}. While their algorithm was for sampling from the Sherrington-Kirkpatrick (SK) model given the Hamiltonian rather than training examples as input, if one is instead given training examples drawn from the SK measure, then at sufficiently high temperature one can approximately recover the Hamiltonian~\cite{alaoui2022bounds}. In this case, a suitable modification \cite{el2022sampling} should be able to yield an algorithm for approximately generating fresh samples from the SK model given training examples.

\paragraph{Learning mixtures of Gaussians.} The literature on provable algorithms for learning Gaussian mixture models is vast, dating back to the pioneering work of Pearson~\cite{pearson1894contributions}, and we cannot do justice to it here. We mention only works whose quantitative guarantees are closest in spirit to ours and refer to the introduction of \cite{liu2022clustering} for a comprehensive overview of recent works in this direction. For mixtures of identity-covariance Gaussians in high dimensions, the strongest existing guarantee is a polynomial-time algorithm~\cite{liu2022clustering} for learning the centers as long as their pairwise separation slightly exceeds $\Omega(\sqrt{\log K})$ based on a sophisticated instantiation of method of moments inspired by the quasipolynomial-time algorithms of~\cite{diakonikolas2018list,hopkins2018mixture,kothari2018robust}. By the lower bound in \cite{regev2017learning}, this is essentially optimal. In contrast, our Theorem~\ref{thm:informal_K} only applies given one initializes in a neighborhood of the true parameters of the mixture. We also note the exponential-time spectral algorithm of \cite{suresh2014near} and quasipolynomial-time tensor-based algorithm of \cite{diakonikolas2020small}, which achieve \emph{density estimation} even in the regime where the centers are arbitrarily closely spaced and learning the centers is information-theoretically impossible.

A separate line of work has investigated the ``textbook'' algorithm for learning Gaussian mixtures, namely the EM algorithm \cite{balakrishnan2017statistical,dasgupta2007probabilistic,daskalakis2017ten,xu2016global,yan2017convergence,zhao2020statistical,kwon2020algorithm,segol2021improved}. Notably, for balanced mixtures of two Gaussians with the same covariance, \cite{daskalakis2017ten} showed that finite-sample EM with random initialization converges exponentially quickly to the true centers. For mixtures of $K$ Gaussians with identity covariance, \cite{kwon2020algorithm,segol2021improved} showed that from an initialization sufficiently close to the true centers, finite-sample EM converges exponentially quickly to the true centers as long as their pairwise separation is $\Omega(\sqrt{\log K})$. In particular, \cite{segol2021improved} establish this local convergence as long as every center estimate is initialized at distance at most $\Delta/2$ away from the corresponding true center, where $\Delta$ is the minimum separation between any pair of true centers; this radius of convergence is provably best possible for EM.

Lastly, we note that there are many works giving parameter recovery algorithms mixtures of Gaussians with general mixing weights and covariances, all of which are based on method of moments~\cite{kalai2010efficiently,hardt2015tight,kane2021robust,belkin2015polynomial,moitra2010settling,liu2023robustly,bakshi2022robustly,diakonikolas2020robustly}. Unfortunately, for general mixtures of $K$ Gaussians, these algorithms run in time at least $d^{O(K)}$, and there is strong evidence~\cite{diakonikolas2017statistical,bruna2021continuous} that this is unavoidable for computationally efficient algorithms.

\subsection{Technical overview}
\label{sec:overview}

We begin by describing in greater detail the algorithm we analyze in this work. For the sake of intuition, in this overview we will focus on the case of mixtures of two Gaussians $(K=2)$ where the centers are well-separated and symmetric about the origin, that is, the data distribution is given by
\begin{equation}
    q = \frac{1}{2}\mathcal{N}(\mu^*,\Id) + \frac{1}{2}\mathcal{N}(-\mu^*,\Id)\,. \label{eq:mixture_simple}
\end{equation}
At the end of the overview, we briefly discuss the key challenges for handling smaller separation and general $K$. 

\paragraph{Loss function, architecture of the score function and student network.} The algorithmic task at the heart of score estimation is that of denoising. Formally, for some noise level $t > 0$, we are given a noisy sample 
\begin{equation}
    X_t = \exp(-t) X_0 + \sqrt{1 - \exp(-2t)} Z_t\,, \label{eq:denoise_overview}
\end{equation}
where $X_0$ is a clean sample drawn from the data distribution $q$, and $Z_t\sim\calN(0,\Id)$. Conditioning on $X_t$ induces some posterior distribution over the noise $Z_t$, and our goal is to form an estimate $s$ for the mean of this posterior which achieves small error on average over the randomness of $X_0$ and $Z_t$. That is, we would like to minimize the \emph{DDPM objective}, which up to rescaling is
given by\footnote{The real DDPM objective is slightly different, see~\eqref{eq:diffusion-loss}. The latter is what we actually consider in this paper, but this distinction is unimportant for the intuition in this overview.}
\begin{equation}
    L_t(s) = \mathbb{E}_{X_0, Z_t} \|s(X_t) - Z_t\|^2\,.
\end{equation}
As discussed in the introduction, the algorithm of choice for minimizing this objective in practice is gradient descent on some student network. To motivate our choice of architecture, note that when the data distribution is given by~\eqref{eq:mixture_simple}, the true minimizer of $L_t$ is, up to scaling,
\begin{equation}
    \tanh(\langle \mu^*_t, x\rangle)\mu^*_t - x\,, \;\;\; \text{where} \ \mu^*_t \triangleq \mu^* \exp(-t)\,. \label{eq:simple_score}
\end{equation}
See Appendix \ref{sec:proof-in-preliminaries} for the derivation. Notably, Eq.~\eqref{eq:simple_score} is exactly a two-layer neural network with $\tanh$ activation. As a result, we use the same architecture for our student network when running gradient descent. 
That is, given weights $\mu\in\R^d$, our student network is given by $s_{\mu}(x) \triangleq \tanh(\mu^\top x)\mu - x$. The exact gradient updates on $\mu$ are given in Lemma~\ref{lemma:update}. 

As we discuss next, depending on whether the noise level $t$ is large or small, this update closely approximates the update in one of two well-studied algorithms for learning mixtures of Gaussians: power method and EM respectively.

\paragraph{Learning mixtures of two Gaussians.} We first provide a brief overview of the analysis and then go into the details of the analysis. We start with mixtures of two Gaussians of the form~\eqref{eq:mixture_simple} where $\|\mu^*\|$ is $\Omega(1)$. In this case, we analyze the following two-stage algorithm. We first use gradient descent on the DDPM objective with large $t$ starting from random initialization. We show that gradient descent in this ``high noise'' regime resembles a type of power iteration and gives $\mu$ that has a nontrivial correlation with $\mu^*_t$. Starting from this $\mu$, we then run gradient descent with small $t$. We show that the gradient descent in this ``small noise'' regime corresponds to the EM algorithm and converges exponentially quickly to the ground truth.

\paragraph{Large noise level: connection to power iteration.} When $t$ is large, we show that gradient descent on the DDPM objective is closely approximated by power iteration. More precisely, in this regime, the negative gradient of $L_t(s_{\mu})$ is well-approximated by
\begin{equation}
\label{eq:large-noise-grad-equivalence-intro}
    -\nabla_{\mu} L_t(s_{\mu}) \approx (2\mu_t^* \mu_t^{*\top} - r \Id)\, \mu\,,
\end{equation}
where $r$ is a scalar that depends on $\mu$ (See Lemma \ref{lemma:grad-equivalence}). So the result of a single gradient update with step size $\eta$ starting from $\mu$ is given by
\begin{equation}
    \mu' \triangleq \mu - \eta \nabla_{\mu} L_t(s_{\mu}) \approx ((1 - \eta r )\,\Id + 2 \eta \mu_t^* \mu_t^{*\top} ) \mu\,. \label{eq:overview_power}
\end{equation}
This shows us that each gradient step can be approximated by one step of power iteration (without normalization) on the matrix $(1 - \eta r )\,\Id + 2 \eta \mu_t^* \mu_t^{*\top}$. It is know that running enough iterations of the latter from a random initialization will converge in angular distance to the top eigenvector, which in this case is given by $\mu^*_t$. This suggests that if we can keep the approximation error in \eqref{eq:overview_power} under control, then gradient descent on $\mu$ will also allow us to converge to a neighborhood of the ground truth. We implement this strategy in Lemma \ref{lemma:large-inner-product-main-paper}. Next, we argue that once we are in a neighborhood of the ground truth, we can run GD on the DDPM objective at \emph{low noise level} to refine our estimate.

\paragraph{Low noise level: connection to the EM algorithm.} When $t$ is small, we show that gradient descent on the DDPM objective is closely approximated by EM. Here, our analysis uses the fact that $\mu^*$ is sufficiently large and requires that we initialize $\mu$ to have sufficiently large correlation with the true direction $\mu^*_t$. We can achieve the latter using the large-$t$ analysis in the previous section.

Provided we have this, when $t$ is small it turns out that the negative gradient is well-approximated by
\begin{align}
    -\nabla_{\mu} L_t(s_\mu) \approx \mbb{E}_{X \sim \mc{N}(\mu_t^*, \Id) }[ \tanh ( \langle \mu, X\rangle ) X ] - \mu\,.
\end{align}
Note that the expectation is precisely the ``M''-step in the EM algorithm for learning mixtures of two Gaussians (see e.g. Eq. (2.2) of \cite{daskalakis2017ten}).
We conclude that a single gradient update with step size $\eta$ starting from $\mu$ is given by mixing the old weights $\mu$ with the result of the ``M''-step in EM:
\begin{align*}
    \mu' \triangleq \mu - \eta \nabla_{\mu} L_t(s_\mu) \approx (1 - \eta) \mu + \eta \underbrace{\mbb{E}_{X \sim \mc{N}(\mu_t^*, \Id) }[ \tanh ( \langle \mu, X \rangle ) X ]}_{\text{``M'' step in the EM algorithm}}\,.
\end{align*}
\cite{xu2016global} and \cite{daskalakis2017ten} showed that EM converges exponentially quickly to the ground truth $\mu^*_t$ from a warm start, and we leverage ingredients from their analysis to prove the same guarantee for gradient descent on the DDPM objective at small noise level $t$ (see Lemma \ref{lemma:l2-contraction-main-paper}).

\paragraph{Extending to small separation.} Next, suppose we instead only assume that $\|\mu^*\|$ is $\Omega({1}/{\mathrm{poly}(d)})$, i.e. the two components in the mixture may have small separation. The above analysis breaks down for the following reason: while it is always possible to show that gradient descent at large noise level converges in \emph{angular distance} to the ground truth, if $\|\mu^*\|$ is small, then we cannot translate this to convergence in Euclidean distance. 

We circumvent this as follows. Extending the connection between gradient descent at large $t$ and power iteration, we show that a similar analysis where we instead run \emph{projected} gradient descent over the ball of radius $\|\mu^*\|$ yields a solution arbitrarily close to the ground truth, even without the EM step.\footnote{Note that although $\mu^*$ is unknown, we can estimate its norm from samples.} The projection step can be thought of as mimicking the normalization step in power iteration. %

It might appear to the reader that this projected gradient-based approach is strictly superior to the two-stage algorithm described at the outset. However, in addition to obviating the need for a projection step when separation is large, our analysis for the two-stage algorithm has the advantage of giving much more favorable statistical rates. Indeed, we can show that the sample complexity of the two-stage algorithm has optimal dependence on the target error ($1/\epsilon^2$), whereas we can only show a suboptimal dependence ($1/\epsilon^8$) for the single-stage algorithm.

\paragraph{Extending to general $K$.} The connection between gradient descent on the DDPM objective at small $t$ and the EM algorithm is sufficiently robust that for general $K$, our analysis for $K = 2$ can generalize once we replace the ingredients from~\cite{xu2016global} and \cite{daskalakis2017ten} with the analogous ingredients in existing analyses for EM with $K$ Gaussians. For the latter, it is known that if the centers of the Gaussians have separation $\Omega(\sqrt{\log \min(K,d)})$, then EM will converge from a warm start \cite{kwon2020algorithm, segol2021improved}. By carefully tracking the error in approximating the negative gradient with the ``M''-step in EM, we are able to show that gradient descent on the DDPM objective at small $t$ achieves the same guarantee. %

\subsection{Preliminaries}
\label{subsec:preliminaries}

\paragraph{Diffusion models.}
Throughout the paper, we use either $q$ or $q_0$ to denote the data distribution and $X$ or $X_0$ to denote the corresponding random variable on $\mbb{R}^d$. The two main components in diffusion models are the \emph{forward process} and the \emph{reverse process}. The forward process transforms samples from the data distribution into noise, for instance via the \emph{Ornstein-Uhlenbeck (OU) process}: 
\begin{align*}
    \mathrm{d} X_t = - X_t \, \mathrm{d} t + \sqrt{2} \,\mathrm{d} W_t \;\;\;\text{with} \;\;\; X_0 \sim q_0\,,
\end{align*}
where $(W_t)_{t\ge 0}$ is a standard Brownian motion in $\R^d$.
We use $q_t$ to denote the law of the OU process at time $t$. Note that for $X_t \sim q_t$,
\begin{equation}
\label{eq:Xt-density}
    X_t = \exp(-t) X_0 + \sqrt{1 - \exp(-2t)} Z_t \;\;\; \text{with} \;\; X_0 \sim q_0, \;\; Z_t \sim \mc{N}(0, \Id)\,.
\end{equation}

The reverse process then transforms noise into samples, thus performing generative modeling. Ideally, this could be achieved by running the following stochastic differential equation for some choice of terminal time $T$:
\begin{equation}
    \mathrm{d} X^\leftarrow_t = \{X^\leftarrow_t + 2\nabla_x \ln q_{T-t}(X^\leftarrow_t)\}\, \mathrm{d}t + \sqrt{2}\,\mathrm{d}W_t \;\;\;\text{with} \;\;\; X^\leftarrow_0 \sim q_T\,,
\end{equation}
where now $W_t$ is the reversed Brownian motion. In this reverse process, the iterate $X^\leftarrow_t$ is distributed acccording to $q_{T - t}$ for every $t\in[0,T]$, so that the final iterate $X^\leftarrow_T$ is distributed according to the data distribution $q_0$. The function $\nabla_x \ln q_t$ is called the \emph{score function}, and because it depends on $q$ which is unknown, in practice one estimates it by minimizing the \emph{score matching loss}
\begin{align}
\label{eq:score-matching-obj}
    \min_{s_t} \;\; \mbb{E}_{X_t \sim q_t}[ \| s_t(X_t) - \nabla_x \ln q_t(X_t) \|^2 ]\,.
\end{align}
A standard calculation (see e.g. Appendix A of~\cite{chen2023sampling}) shows that this is equivalent to minimizing the \emph{DDPM objective} in which one wants to predict the noise $Z_t$ from the noisy observation $X_t$, i.e.
\begin{align}
\label{eq:diffusion-loss}
    \min_{s_t} \;\; L_t(s_t) = \mbb{E}_{X_0, Z_t} \Big[ \Big\| s_t(X_t) + \frac{Z_t}{ \sqrt{1 - \exp(-2t)} } \Big\|^2 \Big]\,.
\end{align}
While we have provided background on diffusion models for context, in this work we focus specifically on the optimization problem~\eqref{eq:diffusion-loss}.

\paragraph{Mixtures of Gaussians.}
We consider the case of learning mixtures of $K$ equally weighted Gaussians:
\begin{align}
\label{eq:mog-k-pdf}
    q = q_0 = \frac{1}{K} \sum_{i=1}^K \mc{N}(\mu_i^*, \Id),
\end{align}
where $\mu_i^*$ denotes the mean of the $i^{\text{th}}$ Gaussian component. We define $\theta^* = \{ \mu_1^*, \mu^*_2 \ldots , \mu_K^* \}$. 
For the mixtures of two Gaussians, we can simplify the data distribution as
\begin{align}
\label{eq:mog-2-pdf}
    q = q_0 = \frac{1}{2}\mc{N}(\mu^*, \Id) + \frac{1}{2}\mc{N}(-\mu^*, \Id).
\end{align}
Note that distribution in Eq.~\eqref{eq:mog-2-pdf} is equivalent to the distribution Eq.~\eqref{eq:mog-k-pdf} with $K=2$ because shifting the latter by its mean will give the former distribution, and furthermore the necessary shift can be estimated from samples. The following is immediate:

\begin{lemma}
\label{lemma:Xt-mog-pdf}
If $q_0$ is a mixture of $K$ Gaussians as in Eq.~\eqref{eq:mog-k-pdf}, then for any $t > 0$, $q_t$ is the mixture of $K$ Gaussians given by
\begin{equation}
\label{eq:mog-k-pdf-time-t}
    q_t = \frac{1}{K} \sum_{i=1}^K \mc{N}(\mu_{i, t}^*, \Id) \;\;\; \text{where} \;\; \mu_{i, t}^* \triangleq \mu_i^* \exp(-t)\,.
\end{equation}
\end{lemma} 
\noindent See Appendix~\ref{sec:proof-in-preliminaries} for a proof of this fact. We can see that the means of $q_t$ get rescaled according to the noise level $t$. We also define $\theta_t^* = \{ \mu_{1,t}^*, \mu_{2,t}^*, \ldots, \mu_{K,t}^* \}$. 

\begin{lemma} 
\label{lemma:score-fn-calculation}
The score function for distribution $q_t$, for any $t > 0$, is given by
\begin{align*}
    \nabla_x \ln q_t(x) = \sum_{i=1}^K w^*_{i, t}(x) \mu_{i, t}^* - x\,, \hspace{5mm}  \text{ where } \hspace{5mm} w_{i, t}^*(x) = \frac{ \exp(-\| x-\mu_{i, t}^* \|^2/2 ) }{ \sum_{j=1}^K \exp(-\| x-\mu_{j, t}^* \|^2/2 ) }.
\end{align*}
For a mixture of two Gaussians, the score function simplifies to 
\begin{align*}
    \nabla_x \log q_t(x) = \tanh( \mu^{*\top}_t x ) \mu_t^* - x\,, \hspace{5mm} \text{where} \hspace{5mm} \mu_t^* \triangleq \mu^* \exp(-t)
\end{align*}
\end{lemma}
\noindent See Appendix \ref{sec:proof-in-preliminaries} for the calculation.

Recall that $\nabla_x \log q_t(x)$ is the minimizer for the score-matching objective given in Eq.~\eqref{eq:score-matching-obj}. Therefore, we parametrize our student network architecture similarly to the optimal score function. Our student architecture for mixtures of $K$ Gaussians is 
\begin{align}
\label{eq:score-function-definition}
    s_{\theta_t}(x) = \sum_{i=1}^K w_{i, t}(x) \mu_{i, t} - x\,, \hspace{5mm}
    \text{ where } \hspace{5mm} w_{i, t}(x) &\triangleq \frac{ \exp(- \| x-\mu_{i, t} \|^2 / 2 )} { \sum_{j=1}^K \exp(- \| x-\mu_{j, t} \|^2 / 2 ) } \\  \mu_{i, t} &\triangleq \mu_i \exp(-t).
\end{align}
where $\theta_t = \{ \mu_{1,t}, \mu_{2,t}, \ldots, \mu_{K, t} \}$ denotes the set of parameters at the noise scale $t$. 
For mixtures of two Gaussians, we simplify the student architecture as follows: 
\begin{equation}
\begin{aligned}
    s_{\theta_t}(x) = \tanh(\mu_t^\top x)\mu_t - x\,, 
    \;\; \text{ where } \;\; \mu_{t} \triangleq \mu \exp(-t).
\end{aligned}
\end{equation}
As $\theta_t$ only depends on $\mu_t$ in the case of mixtures of two Gaussians, we simplify the notation of the score function from $s_{\theta_t}(x)$ to $s_{\mu_t}(x)$ in that case. We use $\hat{\mu}_t$ and $\hat{\mu}_t^*$ to denote the unit vector along the direction of $\mu_t$ and $\mu_t^*$ respectively.  Note that we often use $\mu_t$ (or $\theta_t$) to denote the current iterate of gradient descent on the DDPM objective and $\mu'_t$ to denote the iterate after taking a gradient descent step from $\mu_t$. 

\paragraph{Expectation-Maximization (EM) algorithm.} The EM algorithm is composed of two steps: the E-step and the M-step. For mixtures of Gaussians, the E-step computes the expected log-likelihood based on the current mean parameters and the M-step maximizes this expectation to find a new estimate of the parameters. 

\begin{fact}[See e.g., \cite{daskalakis2017ten, yan2017convergence, kwon2020algorithm} for more details]
\label{fact:em-update}
When $X$ is the mixture of $K$ Gaussian and $\{ \mu_1, \mu_2, \ldots, \mu_K \}$ are current estimates of the means,  the population EM update for all $i \in \{1,2,\ldots,K\}$ is given by
\begin{align*}
    \mu_i' = \frac{\mbb{E}_X[w_i(X) X]}{ \mbb{E}_X[w_i(X)] }, \;\;\; \text{where} \;\; w_i(X) = \frac{ \exp(- \| X-\mu_{i} \|^2 / 2 ) }{ \sum_{j=1}^K \exp(- \| X-\mu_{j} \|^2 / 2 ) }.
\end{align*}
The EM update for mixtures of two Gaussians given in Eq.~\eqref{eq:mog-2-pdf} simplifies to 
\begin{align}
\label{eq:mo2g-EM}
    \mu' = \mbb{E}_{X \sim \mc{N}(\mu^*, \Id)}[ \tanh(\mu^\top X) X].
\end{align}
\end{fact} 
An analogous version of the EM algorithm, called the gradient EM algorithm, takes a gradient step in the direction of the M-step instead of optimizing the objective in the M-step fully. 
\begin{fact}[See e.g., \cite{yan2017convergence, segol2021improved} for more details] 
\label{fact:gradient-EM}
For all $i \in \{1,2,\ldots,K\}$,
 the gradient EM-update for mixtures of $K$ Gaussian is given by
\begin{align*}
    \mu_i' = \mu_i + \eta \,\mbb{E}_X[ w_i(X)(X - \mu_i) ],
\end{align*}
where $\eta$ is the learning rate.
\end{fact}

\section{Warmup: mixtures of two Gaussians with constant separation}
\label{sec:mo2g-const-sep}

In this section, we formally state our result for learning mixtures of two Gaussians with constant separation. This case highlights the main proof techniques, namely viewing gradient descent on the DDPM objective as power iteration and as the EM algorithm.

\subsection{Result and algorithm}

\begin{theorem}
\label{thm:mo2g-const-sep}
    There is an absolute constant $c > 0$ such that the following holds. Suppose a mixture of two Gaussians with the mean parameter $\mu^*$ satisfies $\norm{\mu^*} > c$. Then, for any $\epsilon > 0$, there is a procedure that calls Algorithm~\ref{alg:denoise} at two different noise scales $t$ and outputs $\Tilde{\mu}$ such that $\norm{ \Tilde{\mu} - \mu^* } \leq \epsilon$ with high probability. Moreover, the algorithm has time and sample complexity $\mathrm{poly}(d) / \epsilon^2$ (see Theorem~\ref{thm:mo2g-const-sep-appendix} for more precise quantitative bounds).
\end{theorem}

\paragraph{Algorithm.}
The algorithm has two stages. In the first stage we run gradient descent on the DDPM objective described in Algorithm~\ref{alg:denoise} from a random Gaussian initialization and noise scale $t_1$ for a fixed number of iterations $H$ where $t_1 = O(\log d)$ (``high noise'') and $H = \mathrm{poly}(d, 1/\epsilon)$. In the second stage, the procedure uses the output of the first step as initialization and runs Algorithm~\ref{alg:denoise} at a ``low noise'' scale of $t_2 = O(1)$.

\subsection{Proof outline of Theorem~\ref{thm:mo2g-const-sep}}

We provide a proof sketch of correctness of the above algorithm and summarize the main technical lemmas here.  All proofs of the following lemmas can be found in Appendix~\ref{appsec:proof-2-mog-const-sep}.

\paragraph{Part I: Analysis of high noise regime and connection to power iteration.}  We show that in the large noise regime, the negative gradient $-\nabla L_t(s_t)$ is well-approximated by $2\mu_t^* \mu_t^{*\top} \mu_t - 3\norm{\mu_t}^2 \mu_t$. Recall that this result is the key to showing the resemblance between gradient descent and power iteration. Concretely, we show the following lemma: 

\begin{lemma}[See Lemma \ref{lemma:2-mog-high-noise-grad-equivalence} for more details]
\label{lemma:grad-equivalence}
For $t=O(\log d)$, the gradient descent update on the DDPM objective $L_t(s_t)$ can be approximated with $2\mu_t^* \mu_t^{*\top} \mu_t - 3\norm{\mu_t}^2 \mu_t$:
    \begin{align*}
        \Big\| \rb{-\nabla L_t(s_t) } - \rb{ 2\mu_t^* \mu_t^{*\top} \mu_t - 3\norm{\mu_t}^2 \mu_t } \Big\| \leq \mathrm{poly}(1/d).
    \end{align*}
\end{lemma}

\noindent From Lemma~\ref{lemma:grad-equivalence}, it immediately follows that $\mu't$, the result of taking a single gradient step starting from $\mu_t$, is well-approximated by the result of taking a single step of power iteration for a matrix whose leading eigenvector is $\mu^*_t$:
\begin{align*}
    \mu'_t = \mu_t - \eta \nabla L_t(s_{\mu}) \approx (\Id(1 - 3\eta \| \mu_t \|^2 ) + 2\mu_t^* \mu_t^{*\top} ) \mu_t\,.
\end{align*}

The second key element is to show that as a consequence of the above power iteration update, the gradient descent converges in \emph{angular distance} to the leading eigenvector.  Concretely, we show the following lemma:

\begin{lemma}[Informal, see Lemma~\ref{lemma:projection-angle-decrease} for more details]
\label{lemma:angle-decrease-main-paper}
    Suppose $\mu_t'$ is the iterate after one step of gradient descent on the DDPM objective from $\mu_t$. Denote the angle between $\mu_t$ and $\mu_t^*$ to be $\theta$ and between $\mu'_t$ and $\mu_t^*$ to be $\theta'$. In this case, we show that
    \begin{align*}
        \tan \theta'= \max \rb{ \kappa_1 \tan \theta, \kappa_2 },
    \end{align*}
    where $\kappa_1 < 1$ and $\kappa_2 \le 1 / \mathrm{poly}(d)$. 
\end{lemma}

\noindent Note $\tan \theta' < \tan \theta$ implies that $\theta' < \theta$ or equivalently $\dtp{\hat{\mu}_t'}{ \hat{\mu}_t^* } > \dtp{ \hat{\mu}_t }{ \hat{\mu}_t^* }$.
Thus, the above lemma shows that by taking a gradient step in the DDPM objective, the angle between $\mu_t$ and $\mu_t^*$ decreases. By iterating this, we obtain the following lemma:
\begin{lemma}[Informal, see Lemma~\ref{lemma:2-mog-const-sep-inner-product-const} for more details]
\label{lemma:large-inner-product-main-paper}
    Running gradient descent from a random initialization on the DDPM objective $L_t(s_\mu)$ for $t = O(\log d)$ gives $\mu_t$ for which $\dtp{\hat{\mu}_t}{\hat{\mu}_t^*}$ is $\Omega(1)$. %
\end{lemma}

Note that we cannot keep running gradient descent at this high noise scale and hope to achieve $\mu$ such that $\norm{ \mu - \mu^* }$ is $O(\epsilon)$. This is because Lemma~\ref{lemma:angle-decrease-main-paper} can only guarantee that the angle between $\mu_t$ and $\mu_t^*$ is $O(\epsilon)$, but this does not imply $\norm{ \mu - \mu^* }$ is $O(\epsilon)$.  Instead, as described in Part II, we will proceed with a smaller noise scale.

\paragraph{Part II: Analysis of low noise regime and connection to EM.} In the low noise regime, we run Algorithm~\ref{alg:denoise} using the output from Part I as our initialization.  Our analysis here shows that whenever the initialization $\mu_t$ satisfies the condition of $\dtp{\hat{\mu}_t}{\hat{\mu}_t^*}$ being $\Omega(1)$, $\norm{\mu_t - \mu_t^*}$ contracts after every gradient step. To start with, we show that the result of a \emph{population} gradient step on the DDPM objective $L_t(s_\mu)$ results in the following: 
\begin{align*}
    \mu'_t &= (1 - \eta) \mu_t + \eta \,\mbb{E}_{x \sim \mc{N}(\mu_t^*, \Id)}[ \tanh (\mu_t^\top x) x ] + \eta G(\mu_t, \mu_t^*),
\end{align*}
where $\mu'_t$ is the parameter after a gradient step, $\eta$ is the learning rate, and function $G$ is given by 
\begin{align*}
    G(\mu, \mu^*) = \mbb{E}_{x \sim \mc{N}(\mu^*, \Id)} [ - \frac{1}{2} \tanh''( \mu^\top x ) \norm{ \mu }^2 x + \tanh'( \mu^\top x ) \mu^\top x x  -  \tanh'( \mu^\top x ) \mu ].
\end{align*}
Note we use the population gradient here only for simplicity; in the Appendix we show that empirical estimates of the gradient suffice. 
After some calculation, we can show that
\begin{align}
\label{eq:large-noise-grad-contract}
    \norm{ \mu'_t - \mu^*_t } \leq (1 - \eta) \norm{ \mu_t - \mu^*_t } + \eta \| \mbb{E}_{x \sim \mc{N}(\mu_t^*, \Id)}[ \tanh (\mu_t^\top x) x ] - \mu_t^* \| + \eta \norm{ G(\mu_t, \mu^*_t) }\,.
\end{align}
Using Fact~\ref{fact:em-update}, we know that $\mbb{E}_{x \sim \mc{N}(\mu_t^*, \Id)}[ \tanh (\mu_t^\top x) x ]$ is precisely the result of one step of EM starting from $\mu_t$, and it is known~\cite{daskalakis2017ten} that the EM update contracts the distance between $\mu_t$ and $\mu_t^*$ as follows:
\begin{align}
\label{eq:em-contraction}
    \| \mbb{E}_{x \sim \mc{N}(\mu_t^*, \Id)}[ \tanh (\mu_t^\top x) x ] - \mu_t^* \| \leq \lambda_1 \norm{ \mu_t - \mu_t^* } \;\;\; \text{for some } \; \lambda_1 < 1
\end{align} 
It remains to control the second term in Eq.~\eqref{eq:large-noise-grad-contract}, for which we prove the following: %
\begin{lemma}[Informal, see Lemma~\ref{lemma:multi-dimension-G-contraction} for more details]
\label{lemma:G-contraction-main}
When $\norm{\mu^*} = \Omega(1)$ and the noise scale $t = O(1)$, then for every $\mu$ with $\dtp{\hat{\mu}}{\hat{\mu}^*}$ being $\Omega(1)$, the following inequality holds:
    \begin{align*}
        \| G(\mu_t, \mu_t*) \| \leq \lambda_2 \norm{ \mu_t - \mu_t^* } \;\;\; \text{for some} \;\; \lambda_2 < 1\,.
    \end{align*}
\end{lemma}

Combining Eq.~\eqref{eq:em-contraction} and Lemma~\ref{lemma:G-contraction-main} with Eq.~\eqref{eq:large-noise-grad-contract}, we have
\begin{align}
\label{eq:one-step-contraction}
    \norm{ \mu'_t - \mu^*_t } \leq (1 - \eta(1 - \lambda_1 - \lambda_2) ) \norm{ \mu_t - \mu^*_t }.
\end{align}
We can set parameters to ensure that $\lambda_1 + \lambda_2 < 1$ and therefore that $\norm{ \mu_t - \mu^*_t }$ contracts with each gradient step. Applying Lemma~\ref{lemma:G-contraction-main} and Eq.~\eqref{eq:one-step-contraction}, we obtain the following lemma summarizing the behavior of gradient descent on the DDPM objective in the low noise regime. 
\begin{lemma}[Informal]
\label{lemma:l2-contraction-main-paper}
    For any $\epsilon > 0$ and for the noise scale $t = O(1)$, starting from an initialization $\mu_t$ for which $\dtp{\hat{\mu}_t}{\hat{\mu}_t^*}=\Omega(1)$, running gradient descent on the DDPM objective $L_t(s_{\mu})$ will give us mean parameter $\Tilde{\mu}$ such that $\| \Tilde{\mu} - \mu^* \| \leq O(\epsilon)$.
\end{lemma}
Combining Lemma~\ref{lemma:large-inner-product-main-paper} and Lemma~\ref{lemma:l2-contraction-main-paper}, we obtain our first main result, Theorem~\ref{thm:mo2g-const-sep}, for learning mixtures of two Gaussians with constant separation. For the full technical details, see Appendix~\ref{appsec:proof-2-mog-const-sep}.  

\section{Extensions: small separation and more Components}

\subsection{Mixtures of two Gaussians with small separation}

In this section, we briefly sketch how the ideas from Section~\ref{sec:mo2g-const-sep} can be extended to give our second main result, namely on learning mixtures of two Gaussians even with \emph{small separation}. We defer the full technical details to Appendix~\ref{appsec:proof-2-mog-small-sep}.

\begin{theorem}
\label{thm:2-mog-small-sep}
    Suppose a mixture of two Gaussians has mean parameter $\mu^*$ that satisfies $\| \mu^* \| = \Omega( \frac{1}{ \mathrm{poly} (d)} )$. Then, for any $\epsilon > 0$, there exists a modification of Algorithm~\ref{alg:denoise} that provides an estimate $\mu$ such that $\| \mu - \mu^* \| \leq O(\epsilon)$ with high probability. Moreover, the algorithm has time and sample complexity $\mathrm{poly}(d)/\epsilon^8$ (see Theorem~\ref{thm:mo2g-small-sep-appendix} for more precise quantitative bounds).
\end{theorem}

\paragraph{Algorithm modification.} The algorithm that we analyze runs \emph{projected} gradient descent on the DDPM objective but only in the high noise scale regime where $t = O(\log d)$. At each step, we project the iterate $\mu$ to the ball of radius $R$, where $R$ is an empirical estimate for $\norm{\mu^*}$ obtained by drawing samples $x_1,\ldots,x_n$ from the data distribution and forming $R\triangleq (\frac{1}{n}\sum_{i=1}^n \norm{x_i}^2 - d)^{1/2}$.

\paragraph{Proof sketch.}
Lemma~\ref{lemma:angle-decrease-main-paper} and Lemma~\ref{lemma:large-inner-product-main-paper} apply even when the components of the mixture have small separation, and they show that running gradient descent on the DDPM objective results in $\mu_t$ and $\mu_t^*$ being $O(1)$ close in angular distance.  Although our analysis can be extended to show that gradient descent can achieve $O(\epsilon)$ angular distance, this does not guarantee that $\norm{\mu_t - \mu_t^*}$ is $O(\epsilon)$.  If in addition to being $O(\epsilon)$ close in angular distance, we also have that $\norm{\mu_t} \approx \norm{\mu_t^*}$, then it is easy to see that $\norm{\mu_t - \mu_t^*}$ is indeed $O(\epsilon)$.

Observe that if $R$ is approximately equal to $\norm{\mu_t^*}$, then the projection step in our algorithm ensures that our final estimate $\mu_t$ satisfies this additional condition of $\norm{\mu_t} \approx \norm{\mu^*_t}$. It is not hard to show that $R^2$ is an unbiased estimate of $\norm{\mu^*_t}^2$, so standard concentration shows that taking $n = \mathrm{poly}(d, \frac{1}{\epsilon})$ suffices to ensure that $R$ is sufficiently close to $\norm{\mu_t^*}$.

\subsection{Mixtures of $K$ Gaussians, from a warm start}

In this section, we state our third main result, namely for learning mixtures of $K$ Gaussians given by Eq.~\eqref{eq:mog-k-pdf} from a warm start, and provide an overview of how the ideas from Section~\ref{sec:mo2g-const-sep} can be extended to obtain this result.

\begin{assumption}
\label{asm:mog-k-seperation}
(Separation)
    For a mixture of $K$ Gaussians given by Eq.~\eqref{eq:mog-k-pdf}, for every pair of components $i, j \in \{1,2, \ldots, K\}$ with $i \neq j$, we assume that the separation between their means $\| \mu_{i}^* - \mu_{j}^* \|\ge C\sqrt{\log (\min(K, d)})$ for sufficiently large absolute constant $C > 0$.
\end{assumption}

\begin{assumption}
\label{asm:mog-k-initialization}
(Initialization)
    For each component $i \in \{1,2,\ldots,K\}$, we have an initialization $\mu_i^{(0)}$ with the property that $\|  \mu_i^{(0)} - \mu_i^* \| \leq C'\sqrt{\log (\min(K, d) ) }$ for sufficiently small absolute constant $C' > 0$.
\end{assumption}

\begin{theorem}
\label{thm:mog-k-main}
    Suppose a mixture of $K$ Gaussians satisfies Assumption~\ref{asm:mog-k-seperation}. Then, for any $\epsilon = \Theta(1/\mathrm{poly}(d))$, running gradient descent on the DDPM objective (Algorithm~\ref{alg:denoise}) at low noise scale $t=O(1)$ and with initialization satisfying Assumption~\ref{asm:mog-k-initialization} results in mean parameters $\{ \mu_i \}_{i=1}^K$ such that with high probability, the mean parameters satisfy $\| \mu_i - \mu_i^* \| \leq O(\epsilon)$ for each $i \in \{1, 2, \ldots, K\}$. Additionally, the runtime and sample complexity of the algorithm is $\mathrm{poly}(d, 1/\epsilon)$ (see Theorem~\ref{thm:k-mog-appendix} for more precise quantitative bounds).
\end{theorem}

\noindent We provide a brief overview of the proof here. The full proof can be found in Appendix~\ref{appsec:proof-k-mog}.

\paragraph{Proof sketch.} For learning mixtures of two Gaussians, we have already established the connection between gradient descent on the DDPM objective and the EM algorithm. For mixtures of $K$ Gaussians, however, in a local neighborhood around the ground truth parameters $\theta^*$, we show an equivalence between {\em gradient} EM (recall gradient EM performs one-step of gradient descent on the ``M'' step objective) and gradient descent on the DDPM objective. In particular, our main technical lemma (Lemma \ref{lemma:population-GD-gradient-EM}) shows that for noise scale $t=O(1)$ and for any $\mu_i$ that satisfies $\| \mu_{i} - \mu_{i}^* \| \leq O(\sqrt{\log (\min(K, d) ) } )$, we have 
\begin{align*}
    -\nabla_{\mu_{i,t}} L_t( s_{\theta_t} ) \approx \mbb{E}_{X_t}[ w_{i, t}(X_t)(X_t - \mu_{i, t}) ]. 
\end{align*}
Therefore, the iterate $\mu_{i, t}'$ resulting from a single gradient step on the DDPM objective $L_t( s_{\theta_t} )$ with learning rate $\eta$ is given by
\begin{align}
\label{eq:GD-mog-k}
    \mu_{1, t}' = \mu_{1,t} - \eta \nabla_{\mu_{1,t}} L_t( s_{\theta_t} ) \approx \mu_{1,t} + \eta \,\mbb{E}_{X_t}[ w_{1, t}(X_t)(X_t - \mu_{1, t}) ].
\end{align}
Comparing Fact~\ref{fact:gradient-EM} with Eq.~\eqref{eq:GD-mog-k}, we see the correspondence in this regime between gradient descent on the DDPM objective to gradient EM. Using this connection and an existing local convergence guarantee from the gradient EM literature \citep{segol2021improved, kwon2020algorithm}, we obtain our main theorem for mixtures of $K$ Gaussians. Full details can be found in Appendix~\ref{appsec:proof-k-mog}. 

\section*{Acknowledgments}

SC would like to thank Sinho Chewi, Khashayar Gatmiry, Frederic Koehler, and Holden Lee for enlightening discussions on sampling and score estimation.

\bibliographystyle{alpha}
\bibliography{references}

\newpage
\appendix

\numberwithin{equation}{section}
\counterwithin{theorem}{section}

\paragraph{Roadmap.} In Appendix~\ref{sec:proof-in-preliminaries}, we provide proofs of some simple lemmas from Section~\ref{subsec:preliminaries} and some basic inequalities. In Appendix~\ref{app:prelims} we give additional notation and preliminaries. In Appendix~\ref{appsec:proof-2-mog-const-sep}, we provide the proof details for Theorem~\ref{thm:mo2g-const-sep}, our result on learning mixtures of two Gaussians with constant separation. In Appendix~\ref{appsec:proof-2-mog-small-sep}, we extend this analysis to give a proof of Theorem~\ref{thm:2-mog-small-sep}, our result on learning mixtures of two Gaussians with small separation. In Appendix~\ref{appsec:proof-k-mog}, we provide the proof details for Theorem~\ref{thm:mog-k-main}, our result on learning mixtures of $K$ Gaussians. Finally, in Appendix~\ref{sec:additional} we give further deferred proofs. %

\section{Proofs from Section~\ref{subsec:preliminaries}}
\label{sec:proof-in-preliminaries}

\subsection{$X_t$ is a mixture of Gaussians}
\label{subsec:xt-mog}

\begin{proof}[Proof of Lemma \ref{lemma:Xt-mog-pdf}]
Suppose $X_0$ is mixture of $K$ Gaussians with density function given by
\begin{align*}
    q_0 &= \frac{1}{K} \sum_{i=1}^K \mc{N}(\mu_{i, 0}^*, \Id)
\end{align*}
We know that $X_t = \exp(-t) X_0 + \sqrt{1 - \exp(-2t)} Z_t$ where $Z_t \sim \mc{N}(0, \Id).$ Then, by change of variable of probability density, we have
\begin{align*}
    \text{pdf of } \exp(-t) X_0 &= \frac{1}{K} \sum_{i=1}^K \mc{N}( \mu_{i, 0}^* \exp(-t) , \exp(-2t)\cdot\Id ) \\
     \text{pdf of } \sqrt{1 - \exp(-2t)} Z_t &= \mc{N}( 0,  (1 - \exp(-2t))\cdot \Id)\,. 
\end{align*}
Combining these, we have
\begin{align*}
    q_t(X_t) = \frac{1}{K} \sum_{i=1}^K \mc{N}( \mu_{i, t}^* , I) \hspace{1cm} \text{where} \hspace{1cm} \mu_{i, t}^* = \mu_{i, 0}^* \exp(-t)\,,
\end{align*}
as claimed.
\end{proof}

\subsection{Derivation of score function}
\label{subsec:score-function-derivation}

\begin{proof}[Proof of Lemma \ref{lemma:score-fn-calculation}]
    For mixtures of $K$ Gaussians in the form of Eq.~\eqref{eq:mog-k-pdf}, the score function at time $t$ is given by
    \begin{align*}
        \nabla \log q_t(x) &= -\frac{ \sum_{i=1}^K  e^{ -\frac{\| x - \mu_{i, t}^* \|^2 }{2} } (x - \mu_{i, t}^* ) }{ \sum_{j=1}^K e^{ -\frac{\|x - \mu_{j, t}^* \|^2 }{2} } } \\ 
        &= \sum_{i=1}^K w_{i, t}^*(x) \mu_{i, t}^* - x \;\; \text{ where } \;\; w_{i, t}^*(x) = \frac{ e^{ -\frac{\| x - \mu_{i, t}^* \|^2 }{2} } }{ \sum_{j=1}^K e^{ -\frac{\|x - \mu_{j, t}^* \|^2 }{2} } }. \\
    \end{align*}
    For mixtures of two Gaussians in the form of Eq.~\eqref{eq:mog-2-pdf}, the score function is given by
    \begin{align}
        \label{eq:original-score-2-mog}
            \nonumber \nabla \log q_t(x) &= w_{1, t}^*(x) \mu_{1,t}^* + w_{2, t}^*(x) \mu_{2, t}^* - x \\ 
            \nonumber &= w_{1, t}^*(x) \mu^* - (1 - w_{1, t}^*(x)) \mu^* - x \\ 
            &= (2w_{1, t}^*(x) - 1) \mu^* - x
    \end{align}
By simplifying $w_{1, t}^*(x)$, we obtain
\begin{align}
    w_{1, t}^*(x) = & \frac{1}{1 + \exp( \frac{\norm{x - \mu^* }^2}{2} -\frac{\norm{x + \mu^* }^2}{2} ) } \\ 
    = & \; \frac{1}{1 + \exp( -2 \mu^{*\top} x )} \\ 
    = & \; \sigma( 2 \mu^{*\top} x ) \label{eq:w1-sigmoid}
\end{align}
where $\sigma( \cdot )$ denotes the sigmoid function. Using Eq.~\eqref{eq:w1-sigmoid} in Eq.~\eqref{eq:original-score-2-mog}, we obtain
\begin{align*}
    \nabla \log q_t(x) = \tanh( \mu^{*\top} x ) \mu^* - x.
\end{align*}
\end{proof} 

\section{Additional notations and preliminaries}
\label{app:prelims}

In this section, we provide additional notations and preliminaries for the proofs to follow. Recall that we use $L_t(s_{\theta_t})$ to denote the population denoising loss at noise scale $t$.
\begin{align*}
    L_t(s_{\theta_t}) = \mbb{E} \Big[ \Big\| s_{\theta_t}(X_t) + \frac{Z_t}{ \sqrt{1 - \exp(-2t)} } \Big\|^2 \Big].
\end{align*}
We use $L_t(s_{\theta_t}(x_0, z_t))$ to denote the denoising loss at noise scale $t$ on a sample $x_0$ from the data distribution and $z_t$ from the standard Gaussian distribution: 
\begin{align*}
    L_t(s_{\theta_t}(x_0, z_t)) =  \Big\| s_{\theta_t}(x_t) + \frac{z_t}{ \sqrt{1 - \exp(-2t)} } \Big\|^2,
\end{align*}
where $x_t = \exp(-t) x_0 + \sqrt{1 - \exp(-2t)} z_t$. We use $\alpha_t$ as shorthand notation for $\exp(-t)$ and $\beta_t$ as shorthand notation for $\sqrt{1 - \exp(-2t)}$. 

For mixtures of two Gaussians, we use $B$ to denote the upper bound on $\| \mu^* \|^2$, that is,
\begin{equation}
    \| \mu^* \|^2 \leq B\,.
\end{equation}
Throughout, we assume that $B = \mathrm{poly}(d)$.

For any vector $v$, we use $\hat{v}$ to denote the unit vector along the direction of $v$. For a vector $v$, we use $[v]_i$ to denote the $i^{th}$ coordinate of $v$. Similarly, for a matrix $X$, we use $[X]_i$ to denote the $i^{th}$ row of the matrix. For any positive integer $n$, we use $[n]$ to denote the set $\{1, 2, \ldots, n\}$. We use $\mc{N}(\mu, \sigma^2 \cdot \Id)$ to denote the standard Gaussian with mean $\mu$ and covariance $\sigma^2 \cdot \Id$. Sometimes, we use a shorter notation $\mc{N}_{\mu}$ to denote $\mc{N}(\mu, \Id)$. For any two quantities $X$ and $Y$ that are both implicitly functions of some parameter $a$ over $\R_{\ge 0}$, we use the shorthand $X \lesssim Y$ and $X = O(Y)$ interchangeably to denote that there exists absolute constant $C > 0$ such that for all $a$ sufficiently large,  $X(a) \leq C Y(a)$. We also use the shorthand $X \gtrsim Y$ and $X = \Omega(Y)$, defined in the obvious way.

Finally, we will use the following standard bounds.

\begin{lemma}[Sub-Gaussian norm, see e.g. \cite{vershynin-HDP-book}]
\label{def:subgaussian}
The sub-Gaussian norm of a random variable $X \in \mbb{R}$, denoted by $\sgnorm{X}$ is defined as 
\begin{align*}
    \sgnorm{X} = \inf \{ t > 0 \; : \; \mbb{E}[ \exp(X^2 / t^2) ] \leq 2 \}.
\end{align*} 
The sub-Gaussian norm has the following properties: 
\begin{enumerate}
    \item (Bounded): Any bounded random variable $X$ (i.e., there is a finite $A$ for which $|X| \leq A$ with probability 1) is sub-Gaussian: $$ \sgnorm{X} \leq \frac{ A }{ \sqrt{\ln 2} }$$
    \item (Centering): If $X$ is a sub-Gaussian random variable, then $X - \mbb{E}[X]$ is also a sub-Gaussian random variable. Specifically, the following holds for some absolute constant $C$.
    $$ \sgnorm{ X - \mbb{E}[X] } \leq C \sgnorm{X}$$
    \item (Moment generating function bound): If $X$ is a sub-Gaussian random variable with $E[X] = 0$, then
    $$ \mbb{E}[ \exp( \lambda X ) ] \leq \exp( C \lambda^2 \sgnorm{X}^2 ) \hspace{5mm} \text{for all $\lambda \in \mbb{R}$},$$ 
    where $C$ is some absolute constant. 
    \item (Sum of sub-Gaussian random variables): If $X_1$ and $X_2$ are mean zero sub-Gaussian random variables, then 
    $$ \sgnorm{X_1 + X_2} \leq  \; \sgnorm{X_1} + \sgnorm{X_2} \,. $$
    \item (Product with a bounded random variable): If $X$ is a sub-Gaussian random variable and $Y$ is a bounded random variable $Y \in [0, 1]$, then 
    $$ \sgnorm{X Y} \leq \sgnorm{X} \,.$$
\end{enumerate}
    
\end{lemma}

\begin{lemma}[Sub-exponential norm, see e.g. \cite{vershynin-HDP-book}]
    The sub-exponential norm of a random variable $X \in \mbb{R}$, denoted by $\wnorm{X}_{\psi_1}$ is defined as 
\begin{align*}
    \wnorm{X}_{\psi_1} = \inf \{ t > 0 \; : \; \mbb{E}[ \exp( |X| / t ) ] \leq 2 \}.
\end{align*} 
    The sub-exponential norm has the following properties: 
    \begin{enumerate}
        \item (Sum of sub-exponential distributions): If $X_1$ and $X_2$ are mean-zero sub-exponential random variables, then $X_1 + X_2$ is also a mean-zero sub-exponential variable. Specifically, 
        $$\wnorm{ X_1 + X_2 }_{\psi_1} \leq \sqrt{2} ( \wnorm{ X_1 }_{\psi_1} + \wnorm{ X_2 }_{\psi_1} ) \,.$$ 
        \item (Centering) If $X$ is a sub-exponential random variable, then $X - \mbb{E}[X]$ is sub-exponential with 
        \begin{align*}
            \wnorm{X - \mbb{E}[X]}_{\psi_1} \leq C \wnorm{X}_{\psi_1},
        \end{align*}
        where $C$ is some absolute constant. 
    \end{enumerate}
\end{lemma}
\begin{proof}
    The proof follows from following the equivalent definition of a sub-exponential random variable: If any random variable $X$ satisfies
    \begin{align*}
        \mbb{E}[ \exp(\lambda X) ] \leq \exp( C \wnorm{ X }_{\psi_1}^2 \lambda^2 ) \;\; \text{for all $\lambda$ such that $\abs{\lambda} \leq \frac{1}{ C \wnorm{ X }_{\psi_1}^2 }$}, 
    \end{align*}
    for some constant $C$, then $X$ is sub-exponential random variable with sub-exponential norm $\wnorm{X}_{\psi_1}$.  
    Then, for any $\abs{\lambda} \leq \frac{1}{2C \max( \wnorm{ X_1 }_{\psi_1}^2, \wnorm{ X_2 }_{\psi_1}^2 ) }$, the MGF of $X_1 + X_2$ is given by
    \begin{align*}
        \mbb{E}[ \exp( \lambda(X_1 + X_2) ) ] &\leq \mbb{E}[ \exp(2 \lambda X_1) ]^{1/2} \mbb{E}[ \exp(2 \lambda X_2) ]^{1/2} \\ 
        &\leq \exp( C \wnorm{ X_1 }_{\psi_1}^2 2\lambda^2   ) \exp( C \wnorm{ X_2 }_{\psi_1}^2 2\lambda^2   ) \\
        &\leq \exp( C \lambda^2 (2 \wnorm{ X_1 }_{\psi_1}^2 + 2 \wnorm{ X_2 }_{\psi_1}^2 ) )\,. \qedhere
    \end{align*}
    Using $\wnorm{ X_1 }_{\psi_1} + \wnorm{ X_2 }_{\psi_1} \geq \max( \wnorm{ X_1 }_{\psi_1}, \wnorm{ X_2 }_{\psi_1} )$, we know that above inequality is true for any $\lambda$ with $| \lambda | \leq \frac{1}{2C ( \wnorm{ X_1 }_{\psi_1} + \wnorm{ X_2 }_{\psi_1} )^2 } \leq \frac{1}{2C \max( \wnorm{ X_1 }_{\psi_1}^2, \wnorm{ X_2 }_{\psi_1}^2 ) }$. This completes the proof. 
\end{proof}

\begin{lemma}[Corollary 2.8.4 in \cite{vershynin-HDP-book}]
\label{lemma:bernstein-inequality}
 (Bernstein's inequality for sub-exponential random variable) Let $X_1, X_2, \ldots, X_N$ be independent, mean zero, sub-exponential random variables. Then, for every $\epsilon \geq 0$, we have
\begin{align*}
    \Pr  \Bigg[ \bigg| \frac{1}{N} \sum_{i=1}^N X_i \bigg| \geq \epsilon \Bigg] \leq 2 \exp \bigg[ -c N \min \Big( \frac{\epsilon}{ \max_i \wnorm{X_i}_{\psi_1} }, \frac{\epsilon^2}{ (\max_i \wnorm{X_i}_{\psi_1} )^2 } \Big) \bigg]
\end{align*}
where $c > 0$ is some absolute constant. 
    
\end{lemma}

\section{Learning mixtures of two Gaussians with constant separation}
\label{appsec:proof-2-mog-const-sep}

In this section, we provide the details and proofs for learning mixtures of two Gaussians with constant separation. Our results in this section can be summarized in the following theorem statement. 

\begin{theorem}[Formal version of Theorem \ref{thm:mo2g-const-sep}]
\label{thm:mo2g-const-sep-appendix}
Let $q$ be a mixture of two Gaussians (in the form of Eq.~\eqref{eq:mog-2-pdf}) with mean parameter $\mu^*$ satisfying $\norm{\mu^*} > c$ for some absolute constant $c > 0$. Recalling that $B$ denotes an \emph{a priori} upper bound on $\norm{\mu^*}$, we have that for any $\epsilon \leq \epsilon'$ where $\epsilon' \lesssim \frac{1}{d^2 B^9}$, there exists a procedure satisfying the following. If the procedure is run for at least $\Omega(B^6\log(d/\epsilon))$ iterations with at least $\mathrm{poly}(d,B)/\epsilon^2$ samples from $q$, then it outputs $\Tilde{\mu}$ such that $\wnorm{ \Tilde{\mu} - \mu^* } \leq \epsilon$ with high probability.
\end{theorem}

\noindent As described earlier, the procedure first runs gradient descent on the DDPM objective described in Algorithm \ref{alg:denoise} from a random Gaussian initialization in a high noise scale regime with noise scale $t_1 = O(\log d)$. It then uses the output of the first step as initialization and runs the Algorithm \ref{alg:denoise} in a low noise scale regime with noise scale $t_2 = O(1)$.

We begin by calculating the form of the gradient updates:

\begin{lemma}
\label{lemma:update}
    For any noise scale $t > 0$, the gradient update for the mixture of two Gaussians on the DDPM objective is given by
    \begin{align*}
        -\nabla_{\mu_t} L_t(s_{\mu_t}) = & \; \mbb{E}_{x \sim \mc{N}(\mu_t^*, \Id)} \Big[ \big( \tanh (\mu_t^\top x) - \frac{1}{2} \tanh''( \mu_t^\top x ) \norm{ \mu_t }^2 + \tanh'( \mu_t^\top x ) \mu_t^\top x \big) x \Big] \\
    &- \mu_t - \mbb{E}_{x \sim \mc{N}(\mu_t^*, \Id)} \sbr{\tanh'( \mu_t^\top x ) \mu_t }\,.
    \end{align*}
\end{lemma}
\noindent The proof of Lemma \ref{lemma:update} is given in Appendix \ref{subsec:update-proof}. %

\subsection{High noise regime--connection to power iteration}

Here we show that running population gradient descent on the DDPM objective at \emph{high} noise scale behaves like power iteration on the covariance matrix of the data and thus reaches an iterate $\mu$ with constant correlation with $\mu^*$.

\begin{lemma} 
\label{lemma:2-mog-high-noise-grad-equivalence}
For any noise scale $t > t'$ and number of samples $n > n'$ where $t' \lesssim \log d$ and $n'=\Theta \big( \frac{d^4 B^3}{\epsilon^2} \big)$, with high probability, the negative gradient of the diffusion model objective $L_t(s_t)$ can be approximated by $2\mu_t^* \mu_t^{*\top} \mu_t - 3\norm{\mu_t}^2 \mu_t$. More precisely, given independent samples $\{x_{i,t}\}_{i=1,\ldots,n}$ from $q_t$ generated using noise vectors $\{z_{i,t}\}_{i=1,\ldots,n}$ sampled from $\mc{N}(0,\Id)$, we have
    \begin{align*}
        \bigg\| -\nabla \Big( \frac{1}{n} \sum_{i=1}^n L_t(s_{\mu_t}(x_{i, t}, z_{i, t})) \Big) - \rb{ 2\mu_t^* \mu_t^{*\top} \mu_t - 3\norm{\mu_t}^2 \mu_t } \bigg\| \leq 250 \sqrt{d} \norm{ \mu_t }^5 + 10  \norm{\mu_t}^3 \norm{\mu_t^*}^2 + \epsilon\,.
    \end{align*} %
\end{lemma}
\begin{proof}
    Recall that the population gradient update on the DDPM objective is given by 
    \begin{align*}
        -\nabla L_t(s_{\mu_t}) = & \; \mbb{E}_{x \sim \mc{N}(\mu_t^*, \Id) } \big[ \tanh (\mu_t^\top x) x - \frac{1}{2} \tanh''( \mu_t^\top x ) \norm{ \mu_t }^2 x + \tanh'( \mu_t^\top x ) \mu_t^\top x x \big] \\
    & \quad - \mu_t  - \mbb{E}_{x \sim \mc{N}(\mu_t^*, \Id) } [\tanh'( \mu_t^\top x ) \mu_t] \\
    = &  \; \mbb{E}_{x \sim \mc{N}(\mu_t^*, \Id) } \big[  \tanh (\mu_t^\top x) x - \frac{1}{2} \tanh''( \mu_t^\top x ) \norm{ \mu_t }^2 x + \tanh'( \mu_t^\top x ) \mu_t^\top x \mu_t^* \\ 
    & \quad + \tanh''( \mu_t^\top x ) \mu_t^\top x \mu_t \big] - \mu_t\,,
    \end{align*}
    where the last equality follows from the Stein's lemma on $\mbb{E}_{x \sim \mc{N}(\mu_t^*, \Id) } [ \tanh'( \mu_t^\top x ) \mu_t^\top x x ]$, as
    \begin{equation}
        \mbb{E}_{x \sim \mc{N}(\mu_t^*, \Id) } [ \tanh'( \mu_t^\top x ) \mu_t^\top x x ] =   \mbb{E}_{x \sim \mc{N}(\mu_t^*, \Id) } [ \tanh'( \mu_t^\top x ) \mu_t^\top x \mu_t^* + \tanh'( \mu_t^\top x ) \mu_t + \tanh''( \mu_t^\top x ) \mu_t^\top x \mu_t ]\,.    
    \end{equation}
    Using Taylor's theorem, we know that
    \begin{align*}
    & \tanh(\mu_t^\top x) = \mu_t^\top x - \frac{2}{3} ( \mu_t^\top x )^3 + O( \xi(x)^5 ) \hspace{1cm} \text{ where $\xi (x) \in [0, \mu_t^\top x]$ } \\
    \implies & \tanh(\mu^\top x) x = \mu^\top x x - \frac{2}{3} ( \mu_t^\top x )^3 x + O( \xi(x)^5 x ) \\
    \implies & \Big\| \mbb{E}_{x \sim \mc{N}(\mu_t^*, \Id) }[ \tanh (\mu_t^\top x) x ]  - \mbb{E}_{x \sim \mc{N}(\mu_t^*, \Id) } \big[ \mu_t^\top x x - \frac{2}{3} ( \mu_t^\top x )^3 x \big] \Big\| \leq \| \mbb{E}[ \xi(x)^5 x ] \| \lesssim  \sqrt{d} \norm{\mu_t}^5 
    \end{align*}
    where the last inequality follows from $\norm{ \mbb{E}[ \xi(x)^5 x ] } \leq \mbb{E}[ | \mu_t^\top x |^5 \norm{x} ] \leq \rb{ \mbb{E}[ | \mu_t^\top x|^{10} ] }^{1/2} \rb{ \mbb{E}[ \; \wnorm{x}^2 ] }^{1/2 } \lesssim \wnorm{\mu_t}^5 \sqrt{d + \norm{\mu_t^*}^2} \lesssim \sqrt{d} \norm{\mu_t}^5$. Similarly, using Taylor's theorem, we get
    \begin{align*}
        & \tanh''(\mu_t^\top x) = -2 \mu_t^\top x + O( \xi(x)^3 ) \hspace{1cm} \text{ where $\xi (x) \in [0, \mu_t^\top x]$ } \\
        \implies & \tanh''(\mu_t^\top x) \rb{ - \frac{1}{2}\wnorm{\mu_t}^2 x + \mu_t^\top x \mu_t } = \rb{ -2 \mu_t^\top x + O( \xi(x)^3 ) } \rb{ -\frac{1}{2} \norm{\mu_t}^2 x + \mu_t^\top x \mu_t } \\ 
        \implies & \Big\| \mbb{E}[ \tanh'' (\mu_t^\top x) \big( -\frac{1}{2} \norm{\mu_t}^2 x + \mu_t^\top x \mu_t \big) ]  - \mbb{E} \Big[ -2 \mu_t^\top x  \rb{ -\frac{1}{2} \norm{\mu_t}^2 x + \mu_t^\top x \mu_t } \Big] \Big\| \\ 
        & \leq \big\|  -\frac{1}{2} \norm{\mu_t}^2 \mbb{E}_{x \sim \mc{N}( \mu_t^*, I )} [  O( \xi(x)^3 )  x ] + \mbb{E}_{x \sim \mc{N}( \mu_t^*, I )} [ O( \xi(x)^3 ) \mu_t^\top x \mu_t ] \big\| \\ 
        & \leq \frac{1}{2} \norm{\mu_t}^2 \mbb{ E } [ | \mu_t^\top x |^3 \norm{x} ] + \norm{\mu_t} \mbb{E}[ | \mu_t^\top x |^4 ] \\
        &\leq \frac{1}{2} \norm{\mu_t}^2 \sqrt{\mbb{ E } [ | \mu_t^\top x |^6  ] \mbb{E}[ \norm{x}^2 ] } + \norm{\mu_t}  \mbb{E}[ | \mu_t^\top x |^4 ] \\ 
        &\leq 10 \norm{\mu_t}^5 \sqrt{d} + 6 \norm{\mu_t}^5
    \end{align*}
    Using Taylor's theorem for $\tanh'$, we get
    \begin{align*}
        & \tanh'(\mu_t^\top x) = 1 - (\mu_t^\top x)^2 + O( \xi(x)^4 ) \hspace{1cm} \text{where $\xi(x) \in [0, \mu_t^\top x]$} \\
        \implies & \tanh'(\mu_t^\top x) \mu_t^\top x \mu_t^* = \mu_t^\top x \mu_t^* - (\mu_t^\top x)^3 \mu_t^* + O( \xi(x)^4 \mu_t^\top x \mu_t^* ) \hspace{1cm} \text{where $\xi(x) \in [0, \mu_t^\top x]$} \\
        \implies & \norm{ \mbb{E}[ \tanh'(\mu_t^\top x) \mu_t^\top x \mu_t^* ] - \mbb{E}[ \mu_t^\top x \mu_t^* - (\mu_t^\top x)^3 \mu_t^* ] } \leq \norm{\mbb{E}[ \xi(x)^4 (\mu_t^\top x) \mu_t^* ]} \\ 
        & \hspace{18em} \leq \mbb{E}[ | \mu_t^\top x |^5 ] \norm{\mu_t^*} \lesssim \norm{\mu_t^*} \norm{ \mu_t }^5
    \end{align*}
    Additionally, we have
    \begin{align*}
        & \mbb{E}_{x \sim \mc{N}(\mu_t^*, \Id)}[ x x^\top \mu_t (1 + \norm{\mu_t}^2) - \frac{2}{3} (\mu_t^\top x)^3 x - 2 \mu_t ( \mu_t^\top x )^2 + \mu_t^\top x \mu_t^* - (\mu_t^\top x)^3 \mu_t^* ] \\
        = & \; (I + \mu_t^* \mu_t^{*\top}) \mu_t (1 + \norm{\mu_t}^2) - \frac{5}{3} \mbb{E}[ (\mu_t^\top x)^3 \mu_t^* ] + \mu_t^* \mu_t^{*\top} \mu_t - 4 \mbb{E}[ \mu_t ( \mu_t^\top x )^2  ] \\
        = & \; (I + \mu_t^* \mu_t^{*\top}) \mu_t (1 + \norm{\mu_t}^2) - \frac{5 \mu_t^* }{3} ( ( \mu_t^\top \mu_t^* )^3 + 3 ( \mu_t^\top \mu_t^* ) \norm{\mu_t}^2 ) \\ 
        & \quad\quad + \mu_t^* \mu_t^{*\top} \mu_t - 4 \mu_t ( \; \norm{\mu_t}^2  + ( \mu_t^\top \mu_t^* )^2 )  \\
        = & \; \mu_t^* \mu_t^{*\top} \mu_t (2 - 4 \norm{\mu_t}^2) + \mu_t (1 - 3 \norm{\mu_t}^2) - \frac{5 \mu_t^* ( \mu_t^\top \mu_t^* )^3 }{3}  - 4 \mu_t ( \mu_t^\top \mu_t^* )^2 
    \end{align*}
    where the second equality uses Stein's lemma on $\mbb{E}[ (\mu_t^\top x)^3 x ]$ and $\mbb{E}[xx^\top] = \Id + \mu_t^* \mu_t^{*\top}$ and the third equality uses Gaussian moments for $\mbb{E}[ (\mu_t^\top x)^2 ]$ and $\mbb{E}[ (\mu_t^\top x)^3 ]$.
    Putting it all together and using triangle inequality, we obtain the desired bound on $\| -\nabla L_t(s_{\mu_t})  - (2\mu_t^* \mu_t^{*\top} \mu_t - 3\norm{\mu_t}^2 \mu_t ) \|$. 
    \begin{align*}
        & \| -\nabla L_t(s_{\mu_t})  - (2\mu_t^* \mu_t^{*\top} \mu_t - 3\norm{\mu_t}^2 \mu_t ) \| \\
        \leq & \; \Big\| -\nabla L_t(s_{\mu_t})  - \mbb{E}[ x x^\top \mu_t (1 + \norm{\mu_t}^2) - \frac{2}{3} (\mu_t^\top x)^3 x - 2 \mu_t ( \mu_t^\top x )^2 + \mu_t^\top x \mu_t^* - (\mu_t^\top x)^3 \mu_t^* - \mu_t ] \Big\| \\
        & + \Big\| \mbb{E}[ x x^\top \mu_t (1 + \norm{\mu_t}^2) - \frac{2}{3} (\mu_t^\top x)^3 x - 2 \mu_t ( \mu_t^\top x )^2 + \mu_t^\top x \mu_t^* - (\mu_t^\top x)^3 \mu_t^* - \mu_t ] \\ 
        & \quad\quad - \rb{ 2\mu_t^* \mu_t^{*\top} \mu_t - 3\norm{\mu_t}^2 \mu_t } \Big\| \\
        \leq & \; \rb{ 200 \sqrt{d} \norm{\mu_t}^5 + 10 \norm{\mu_t}^5 \sqrt{d} + 6 \norm{\mu_t}^5 + 20 \norm{\mu_t^*} \norm{ \mu_t }^5 } + 10 \norm{\mu_t}^3 \norm{\mu_t^*}^2 \\
        \leq & \; 250 \sqrt{d} \norm{ \mu_t }^5 + 10 \norm{\mu_t}^3 \norm{\mu_t^*}^2
    \end{align*}
    Using Lemma \ref{lemma:sample-complexity-k-mog} and triangle inequality, we obtain the result. 
\end{proof}

\noindent We will use the following simple bound on the correlation between the ground truth and a random initialization:

\begin{lemma}
\label{lemma:random-init-similarity}
    A randomly initialized $\mu_0 \sim \mc{N}(0, \Id)$ satisfies that $\abs{ \dtp{ \hat{\mu}_0 }{ \hat{\mu}^* } } \geq \frac{1}{2d}$ with probability at least $1 - O(d^{-1/2})$.
\end{lemma}
\begin{proof}
For $\mu_0 \sim \mc{N}(0, I)$, we know that $\dtp{\mu_0}{ \hat{\mu}^* } \sim \mc{N}(0, I)$. Using Gaussian anti-concentration, with probability at least $1 - 1/\sqrt{d}$ , we have $\abs{ \dtp{ \mu_0 }{ \hat{\mu}^* } } \geq 1/\sqrt{d}$. Because the $L_2$ norm of a Gaussian vector is sub-exponential, with probability at least $1 - \exp(-\Omega(d))$, we have $\norm{\mu_0} \leq 2 \sqrt{d}$. Using the norm bound, with probability at least $1 - 1/\sqrt{d} - \exp( - O(d)) = 1 - O(d^{-1/2})$, we obtain the claimed bound on $\abs{ \dtp{ \hat{\mu_0} }{ \hat{\mu}^* } }$.
\end{proof}

\noindent We can now track the correlation between the iterates of gradient descent and the ground truth:

\begin{lemma}
\label{lemma:projection-angle-decrease}
    Suppose that the vector $\mu_t$ satisfies $|\langle \hat{\mu}_t, \hat{\mu}^*_t\rangle| \ge \frac{1}{2d}$, and let $\mu'_t$ denote the iterate resulting from a single empirical gradient step with learning rate $\eta$ starting from $\mu_t$. Suppose that the empirical gradient and the population gradient differ by at most $\epsilon$. Denote the angle between $\mu_t$ (resp. $\mu'_t$) and $\mu_t^*$ by $\theta$ (resp. $\theta'$). Then
    \begin{align*}
        \tan \theta' = \max \rb{ \kappa_1 \tan \theta, \kappa_2 }
    \end{align*}
    for 
    \begin{align*}
        \kappa_1 &= \frac{ 1 - 3 \eta \|\mu_t\|^2 }{ 1 -3\eta \|\mu_t \|^2  + \eta( \|\mu_t^*\|^2   - 500 \sqrt{d^3} \| \mu_t \|^4 - 20 d  \|\mu_t\|^2 \wnorm{\mu_t^*}^2 - \eta \Tilde{\epsilon} ) } \;\; , \\ 
        \kappa_2 &= \frac{ 500 \eta \sqrt{d^3} \| \mu_t \|^4 + 20 \eta d \| \mu_t \|^2 \wnorm{\mu_t^*}^2 + \eta \Tilde{\epsilon} }{ \wnorm{ \mu_t^* }^2 } \;\; \text{and} \;\; \Tilde{\epsilon} \lesssim \frac{d \epsilon}{ \wnorm{\mu_t} }\,.
    \end{align*}
\end{lemma}

\begin{proof} 
Define $\hat{\mu}^{*\perp}_t$ as the orthogonal vector to $\mu_t^*$ in the plane of $\mu_t$ and $\mu_t^*$. Note that $\mu'_t$ still lies in this plane, so the orthogonal vector to $\mu_t^*$ in the plane of $\mu'_t$ and $\mu_t^*$ is also given by $\hat{\mu}^{*\perp}_t$.

We have
    \begin{align}
        \tan \theta' &= \frac{ \dtp{ \hat{\mu}^{*\perp} }{ \hat{\mu}'_t } }{ \dtp{ \hat{\mu}^*_t }{ \hat{\mu}'_t } } = \frac{ \dtp{ \hat{\mu}^{*\perp}_t }{ \mu'_t } }{ \dtp{ \hat{\mu}^*_t }{ \mu'_t } } \\
        = & \frac{ \dtp{ \hat{\mu}^{*\perp}_t }{ \mu_t + \eta F(\mu_t, \mu_t^*) } +  \dtp{ \hat{\mu}^{*\perp}_t }{ - \eta \nabla L_t(s_t) - \eta F(\mu_t, 
        \mu_t^*) } + \eta \epsilon }{ \dtp{ \hat{\mu}^*_t }{ \mu_t + \eta F(\mu_t, \mu_t^*) }  +  \dtp{ \hat{\mu}^{*\perp}_t }{ - \eta \nabla L_t(s_t) - \eta F(\mu_t, 
        \mu_t^*) } - \eta \epsilon } \\ 
        & \hspace{16em} \text{where} \hspace{5mm} F(\mu, \mu^*) = \rb{ 2\mu^*_t \mu^{*\top}_t \mu_t- 3\norm{\mu_t}^2 \mu_t }  \\
        \leq & \frac{ \sigma_2 \dtp{ \hat{\mu}^{*\perp}_t }{ \mu_t } + \eta \big\|\nabla L_t(s_t) + F(\mu_t, \mu_t^*) \big\| + \eta \epsilon }{ \sigma_1 \dtp{ \hat{\mu}^*_t }{ \mu_t } - \eta \big\| \nabla L_t(s_t) + F(\mu_t, \mu_t^*) \big\| - \eta \epsilon } \label{eq:lastineq}
    \end{align}
    where $\sigma_1$ and $\sigma_2$ are the first and second eigenvalues of $\Id + F(\mu_t,\mu^*_t) = (1-3 \eta \norm{\mu_t}^2)\Id + 2 \eta \mu_t^* \mu_t^{*\top}$,
    given by
    \begin{align*}
        \sigma_1 &= 1 + \eta(2 \wnorm{\mu_t^*}^2 -3\norm{\mu_t}^2) \\
        \sigma_2 &= 1 - 3 \eta \norm{\mu_t}^2 \,.
    \end{align*}
    The last inequality~\eqref{eq:lastineq} follows from the fact that 
    \begin{align}
        \dtp{ \hat{\mu}^{*}_t }{ \mu_t + \eta F(\mu_t, \mu_t^*) } &= \hat{\mu}^{*\top}_t ( (1- 3 \eta \norm{\mu_t}^2) \Id + 2 \eta \mu_t^* \mu_t^{*\top} )\mu_t \\
        &= \mu^\top_t ( (1- 3 \eta \norm{\mu_t}^2) \Id + 2 \eta \mu_t^* \mu_t^{*\top} )\hat{\mu}^{*}_t = \sigma_1 \mu^\top_t \hat{\mu}^{*}_t
    \end{align}
    because $\hat{\mu}^*$ is the first eigenvector of $(1- 3 \eta \norm{\mu_t}^2) \Id + 2 \eta \mu_t^* \mu_t^{*\top}$.
    Recall from Lemma~\ref{lemma:2-mog-high-noise-grad-equivalence} that the deviation between the negative population gradient and the power iteration update $F(\mu_t,\mu^*_t)$ is bounded by
    \begin{align*}
        \frac{\norm{ \nabla L_t(s_t) + F(\mu_t, \mu_t^*) }}{ \dtp{ \mu_t }{ \hat{\mu}^*_t } } \leq \frac{ 250 \eta \sqrt{d} \norm{ \mu_t }^4 + 10 \eta \norm{\mu_t}^2 \norm{\mu_t^*}^2 }{  \dtp{ \hat{\mu}_t }{ \hat{\mu}^*_t } } \leq 500 \eta \sqrt{d^3} \norm{ \mu_t }^4 + 20 d \eta \norm{\mu_t}^2 \norm{\mu_t^*}^2\,.
    \end{align*}
    Substituting this into Eq.~\eqref{eq:lastineq}, we get 
    \begin{align*}
        \tan \theta' %
        &\leq \frac{ \sigma_2 \dtp{ \hat{\mu}^{*\perp}_t }{ \mu_t } + \eta \wnorm{ \nabla L_t(s_t) + F(\mu_t, \mu_t^*) } + \eta \epsilon }{ \dtp{ \hat{\mu}^*_t }{ \mu_t } (\sigma_1  - 500 \eta \sqrt{d^3} \norm{ \mu_t }^4 - 20 d \eta \norm{\mu_t}^2 \norm{\mu^*_t}^2 - \eta \Tilde{\epsilon} ) } \ \ \ \text{where} \ \ \  \Tilde{\epsilon} \lesssim \frac{d \epsilon}{\norm{\mu}}  \\
        &\leq \frac{ \sigma_2 }{ \Tilde{\sigma}_1 } \tan \theta + \frac{1}{ \Tilde{\sigma}_1 } \rb{ 500 \eta \sqrt{d^3} \norm{ \mu }^4 + 20 d \eta \norm{\mu}^2 \wnorm{\mu^*_t}^2 + \eta \Tilde{\epsilon} } \\
        &\qquad \qquad \qquad \qquad \text{where} \ \ \ \Tilde{\sigma}_1 \triangleq \sigma_1  - 500 \eta \sqrt{d^3} \norm{ \mu }^4 - 20 d \eta \norm{\mu}^2 \wnorm{\mu^*_t}^2 - \eta \Tilde{\epsilon}\\
        &\leq \Big( 1 - \frac{\eta \wnorm{\mu^*_t}^2 }{ \Tilde{\sigma}_1 } \Big) \frac{ \sigma_2 }{ \Tilde{\sigma}_1 - \eta \wnorm{ \mu^*_t }^2 } \tan \theta + \Bigl( \frac{ \eta \wnorm{\mu^*_t}^2 }{ \Tilde{\sigma}_1 } \Bigr) \frac{ 500 \eta \sqrt{d^3} \norm{ \mu_t }^4 + 20 d \eta \norm{\mu_t}^2 \wnorm{\mu^*_t}^2 + \eta \Tilde{\epsilon} }{ \eta \wnorm{ \mu^*_t }^2 } \\
        &\leq \max \Big( \frac{ \sigma_2 }{ \Tilde{\sigma}_1 - \eta \wnorm{ \mu^*_t }^2 } \tan \theta, \frac{ 500 \eta \sqrt{d^3} \norm{ \mu_t }^4 + 20 \eta d \,\norm{\mu_t}^2 \wnorm{\mu^*_t}^2 + \eta \Tilde{\epsilon} }{ \wnorm{ \mu^*_t}^2 } \Big)
    \end{align*}
    where the last inequality uses the fact that convex combinations of two values is less than the maximum of two values. 
\end{proof}

\noindent Finally, we obtain the following bound on the correlation between the ground truth and the final iterate of gradient descent: 

\begin{lemma}
\label{lemma:2-mog-const-sep-inner-product-const}
    For any $h \in\mathbb{N}$, let $\mu^{(h)}_t$ denote the iterate after $h$ empirical gradient steps with learning rate $\eta = 1/20$ starting from random initialization, where the empirical gradients are estimated from at least $\Theta(\frac{d^4 B^3}{\epsilon^2})$ samples. Let $\theta^{(h)}$ denote the angle between $\mu^{(h)}_t$ and $\mu^*_t$. For any $\epsilon \lesssim \frac{1}{d^2 B^9}$, there exists $H' \lesssim B^6 \log d$ such that for any $H \ge H'$, if 
    $\frac{1}{B^3} \leq  \norm{\mu_t^{*}} \leq \frac{1}{B^2}$, we have
    \begin{align*}
        \tan \theta^{(H)} \lesssim 1\,.
    \end{align*}
\end{lemma}

\begin{proof}
    Denote the $h$-th iterate of gradient descent by $\mu^{(h)}_t$. In Lemma~\ref{lemma:gradient-update-bound} we show that $\norm{\mu^{(h)}_t} \leq \frac{1}{B^2}$ for all $h$. We would like to apply the bound in Lemma~\ref{lemma:projection-angle-decrease} to argue that the angle with $\mu^*_t$ decreases when going from $\mu^{(h)}_t$ to $\mu^{(h+1)}_t$. Using that $\frac{1}{B^3} \le \wnorm{\mu^*_t} \le \frac{1}{B^2}$ and $\norm{\mu_t} \le \frac{1}{B^2}$, we can bound the quantity $\kappa_1$ that appears in Lemma~\ref{lemma:projection-angle-decrease} by
    \begin{align}
        \kappa_1 &\le \frac{1 - 3\eta\norm{\mu_t}^2}{1 - 3\eta\norm{\mu_t}^2 + \frac{\eta}{B^6}(1 - \frac{500\sqrt{d^3}}{B^2} - \frac{20d}{B^2} - \epsilon d B^9)} \\
        &\le \frac{1}{1 + \frac{\eta}{B^6}(1 - \frac{500\sqrt{d^3}}{B^2} - \frac{20d}{B^2} - \epsilon d B^9)} \le \frac{1}{1 + \frac{\eta}{2B^6}}\,.
    \end{align}
    On the other hand, for $B$ a sufficiently large polynomial in $d$, we can again use that $\frac{1}{B^3} \le \wnorm{\mu^*_t} \le \frac{1}{B^2}$ and $\norm{\mu_t} \le \frac{1}{B^2}$ to bound the quantity $\kappa_2$ that appears in Lemma~\ref{lemma:projection-angle-decrease} by
    \begin{align}
        \kappa_2 \le \frac{500\eta\sqrt{d^3}}{B^2} + \frac{20\eta d}{B^4} + B^9\eta d\epsilon \lesssim \frac{\eta}{d}\,.
    \end{align}
    As $\abs{ \dtp{\hat{\mu}}{ \hat{\mu}^* } } \geq \frac{1}{2d}$, this implies $|\tan \theta^{(h)}| \leq 2d$. Without loss of generality assume that $\tan \theta^{(h)} \leq 2d$. 

    By Lemma~\ref{lemma:projection-angle-decrease}, for any $h$ we either have $\tan\theta^{(h)} \lesssim \eta/d \ll 1$, in which case we are done as this bound will also hold for subsequent iterates, or $\tan\theta^{(h)} \lesssim (1 + \frac{\eta}{2B^6})^{-1} \tan\theta^{(h-1)}$. If the latter happens consecutively for $H \ge \frac{\log d}{\log(1 + \frac{\eta}{2B^6})}$ steps, then because $(1 + \frac{\eta}{2B^6})^{-H} = \frac{1}{d}$, the angle $\theta$ will satisfy $\tan\theta \le 2d\cdot (1/d) \lesssim 1$. The proof is complete because, by hypothesis, $H \ge \frac{4B^6\log d}{\eta} \ge \frac{\log d}{\log(1 + \frac{\eta}{2B^6})}$ (the last inequality follows from $\log (1 + x) \geq \frac{x}{2}$ for any $0 < x < 1$).
\end{proof}

\begin{lemma}
\label{lemma:gradient-update-bound}
    When parameter $\mu_t$ satisfies $\wnorm{\mu_t} \leq \frac{1}{B^2}$ for the noise scale $t = O(\log d)$ and $\mu'_t$ is the new parameter after performing a gradient descent update on the DDPM objective at noise scale $t = O(\log d)$, then parameter $\mu'_t$ satisfies $ \wnorm{\mu'_t} \leq \frac{1}{B^2}$.
\end{lemma}
\begin{proof}
    When $\wnorm{\mu_t} \leq 0.9 \wnorm{\mu_t^*} \leq \frac{0.9}{B^2}$, we have
    \begin{align*}
        \wnorm{ \mu_t' } & \leq \wnorm{ \mu_t + \eta F(\mu_t, \mu_t^* ) } + \eta \wnorm{ ( - \nabla L_t(s_{\mu_t}) - F(\mu, \mu^* ) ) } + \eta \epsilon \leq  (1 + 2 \eta \wnorm{\mu_t^*}^2  )  \norm{\mu_t} + \frac{1}{d B^9} \\ 
        & \quad \leq 1.05 \norm{\mu_t} + \frac{1}{d B^9} \leq \frac{1}{B^2}.
    \end{align*}
    When $\wnorm{\mu_t} \geq 0.9 \wnorm{\mu_t^*}$, then maximum eigenvalue of $F(\mu_t, \mu_t^* )$ is negative. Therefore, $\wnorm{\mu'_t}$ is less than $\frac{1}{B^2}$. Specifically, we have
    \begin{align*}
        \wnorm{ \mu_t' } & \leq \wnorm{ \mu_t + \eta F(\mu_t, \mu_t^* ) } + \eta \wnorm{ ( - \nabla L_t(s_{\mu_t}) - F(\mu, \mu^* ) ) } + \eta \epsilon \\ 
        & \quad \leq (1 + \eta (2 \wnorm{\mu_t^*}^2 - 3 \wnorm{\mu_t}^2 ) ) \norm{\mu_t} + \frac{1}{d B^9} \leq (1 - 0.01 \norm{\mu_t^*}^2 ) \norm{ \mu_t } + \frac{1}{d B^9} \leq \frac{1}{B^2}. \qedhere
    \end{align*}
\end{proof}

\subsection{Low noise regime - connection to EM algorithm}

In the previous section we showed how to obtain a warm start by running gradient descent on the DDPM objective at high noise. We now focus on proving the contraction of $\wnorm{ \mu_t - \mu_t^* }$ starting from this warm start, by running gradient descent at \emph{low} noise. We first prove the contraction for population gradient descent and then, we argue that the empirical gradient descent concentrates well around the population gradient descent. 

As before, we denote $\mu_t$ as the current iterate and $\mu'_t$ as the next iterate obtained by performing (population) gradient descent on the DDPM objective with step size $\eta$. We upper bound $\wnorm{ \mu_t' - \mu_t^* }$ as follows:
\begin{align*}
    \wnorm{ \mu_t' - \mu_t^* } &= \wnorm{ \mu_t - \eta \nabla_{\mu_t} L_t(s_{\mu_t}) - \mu_t^* } \\
    = & \Big\| (1 - \eta) (\mu_t - \mu_t^*) +\eta \,\mbb{E}_{x \sim \mc{N}(\mu_t^*, 1)} \big[ \big( \tanh (\mu_t^\top x) - \frac{1}{2} \tanh''( \mu_t^\top x ) \norm{ \mu_t }^2 \\
    & + \tanh'( \mu_t^\top x ) \mu_t^\top x \big) x \big]  - \eta\, \mbb{E}_{x \sim \mc{N}(\mu_t^*, 1)} [\tanh'( \mu_t^\top x ) \mu_t]  - \eta \mu_t^* \Big\| \\
    \leq &  (1 - \eta)\, \| \mu_t - \mu_t^* \| +\eta \big\| \mbb{E}_{x \sim \mc{N}(\mu_t^*, 1)}  [ \tanh (\mu_t^\top x)  x ] - \mu_t^* \big\| + \eta\, \| G(\mu_t, \mu_t^*) \|\,,
\end{align*}
where 
\begin{align*}
    G(\mu_t, \mu_t^*) \triangleq \mbb{E}_{x \sim \mc{N}(\mu_t^*, \Id)} \Big[ - \frac{1}{2} \tanh''( \mu_t^\top x ) \norm{ \mu_t }^2 x + (\tanh'( \mu_t^\top x ) \mu_t^\top x) x  -  \tanh'( \mu_t^\top x ) \mu_t \Big]\,.
\end{align*}
Recall that $\mbb{E}_{x \sim \mc{N}(\mu_t^*, 1)}  [\tanh (\mu_t^\top x)  x ]$ is the EM update for mixtures of two Gaussians (See Fact \ref{fact:em-update}). If we can show that the $G(\mu_t,\mu^*_t)$ term above is ``contractive'' in the sense that it is decreasing in $\wnorm{\mu_t - \mu^*_t}$, then we can invoke existing results on convergence of EM to show that the distance between the current iterate and $\mu^*_t$ contracts in a single gradient step~\citep{daskalakis2017ten, xu2016global}. Our goal is thus to control $G(\mu_t, \mu_t^*)$.

For this, we start with the 1D case in Lemma~\ref{lemma:G-contraction}. We then extend to the multi-dimensional case in Lemma~\ref{lemma:multi-dimension-G-contraction}.

\begin{lemma}[One-dimensional version]
\label{lemma:G-contraction}
    Let $\mu, \mu^* > 0$, and consider $\mu \in [c, \frac{4\mu^*}{3}]$ for some constant $c$. In this one-dimensional case, the function $G$ specializes to 
    \begin{align}
    \label{eq:one-d-g-def}
        G(\mu, \mu^*) = \mbb{E}_{x \sim \mc{N}(\mu^*, 1)} \Bigl[ - \frac{1}{2} \tanh''( \mu x ) \mu^2 x + \tanh'( \mu x ) \mu x^2  -  \tanh'( \mu x ) \mu \Bigr]\,,
    \end{align}
    and we have
    \begin{align*}
        G(\mu, \mu^*) \leq 0.01 \abs{\mu - \mu^*}
    \end{align*}
\end{lemma}
The proof uses the fact that the function $G$ only contains first or higher-order derivatives of the $\tanh$ function and all the derivatives of $\tanh$ decay exponential quickly as $\mu$ increases. Therefore, when $\mu$ is at least a constant, we obtain the result. The complete proof of lemma \ref{lemma:G-contraction} is given in Appendix \ref{section:proof-G-contraction-lemma}. 

\begin{lemma}[Multi-dimensional version]
\label{lemma:multi-dimension-G-contraction}
     For any noise scale $t$, when the current parameter at noise scale $t$, $\mu_t$, satisfies $\| \mu_t \| \in [c, \frac{4 \dtp{ \hat{\mu}_t }{ \mu_t^* } }{3}]$ for some sufficiently large constant $c$, then the following inequality holds: 
    \begin{align*}
        \norm{G(\mu_t, \mu_t^*)} \leq 0.01 \norm{ \mu_t - \mu_t^* }
    \end{align*}
\end{lemma}
\begin{proof}
    Suppose $\{v_1, v_2, \ldots, v_d\}$ are $d$ orthonormal directions such that $v_1 = \hat{\mu}_t$ and $v_2$ is either of the two unit vectors $\muhp_t$ which are orthogonal to $\hat{\mu}_t$ in the plane of $\mu_t$ and $\mu_t^*$. Recall that
    \begin{align*}
        G(\mu_t, \mu_t^*) &= \mbb{E}_{x \sim \mc{N}(\mu_t^*, \Id)} \Big[ - \frac{1}{2} \tanh''( \mu_t^\top x ) \norm{ \mu_t }^2 x + (\tanh'( \mu_t^\top x ) \mu_t^\top x) x  -  \tanh'( \mu_t^\top x ) \mu_t \Big] \\ 
        &= \mbb{E}_{x \sim \mc{N}(0, I)} \Big[ - \frac{1}{2} \tanh''( \mu_t^\top (x + \mu_t^*) ) \norm{ \mu_t }^2 (x+ \mu_t^*) \\
        & \quad\quad{}+ \tanh'( \mu_t^\top (x+ \mu_t^*) ) (\mu_t^\top (x + \mu_t^*)) (x + \mu_t^*)  -  \tanh'( \mu_t^\top (x + \mu_t^*) ) \mu_t \Big] \\
        &= \mbb{E}_{\alpha_1, \alpha_2, \ldots, \alpha_d \sim \mc{N}(0, 1)} \Big[ - \frac{1}{2} \tanh''( \; \norm{ \mu_t } ( \alpha_1 + \hat{\mu}_t^\top \mu_t^*) ) \norm{ \mu_t }^2 \big( \sum\nolimits_i \alpha_i v_i + \mu_t^*\big) \\ 
        & \quad\quad{}+ \tanh'( \; \norm{ \mu_t } ( \alpha_1 + \hat{\mu}_t^\top \mu_t^*) ) \norm{ \mu_t } ( \alpha_1 + \hat{\mu}_t^\top \mu_t^*) \big( \sum\nolimits_i \alpha_i v_i + \mu_t^*\big) \\
        &\quad \quad{}-  \tanh'( \;\norm{ \mu_t } ( \alpha_1 + \hat{\mu}_t^\top \mu_t^*) ) \mu_t \Big]\,,
    \end{align*}
    where in the last equality we rewrote $x \sim \mc{N}(0, I)$ as $\sum_{i=1}^d \alpha_i v_i$ for $\alpha_i \sim \mc{N}(0, 1)$. Therefore, we have
    \begin{align*}
        & \dtp{\hat{\mu_t}}{ G(\mu_t, \mu_t^*) } \\ 
        \; & \quad = \mbb{E}_{\alpha_1, \alpha_2, \ldots, \alpha_d \sim \mc{N}(0, I)}\Big[ - \frac{1}{2} \tanh''( \; \norm{\mu_t} (\alpha_1 + \hat{\mu}_t^\top \mu_t^*) ) \norm{ \mu_t }^2 (\alpha_1 + \hat{\mu_t}^\top \mu_t^*) \\
        & \quad\quad\quad + \tanh'( \; \norm{\mu_t} (\alpha_1 + \hat{\mu_t}^\top \mu_t^*) ) \norm{\mu_t} (\alpha_1 + \hat{\mu}_t^\top \mu_t^*)^2  -  \tanh'( \; \norm{\mu_t} (\alpha_1 + \hat{\mu_t}^\top \mu_t^*) ) \norm{ \mu_t } \Big] \\
         & \quad = \mbb{E}_{\alpha_1 \sim \mc{N}(\hat{\mu}^\top_t \mu_t^*, 1)}\Big[ -\frac{1}{2} \tanh''(\hspace{0.005em} \wnorm{\mu_t} \alpha_1) \norm{\mu_t}^2 \alpha_1 + \tanh'( \;\norm{\mu_t} \alpha_1 ) \norm{\mu_t} \alpha_1^2 - \tanh'( \;\norm{\mu_t} \alpha_1 ) \norm{ \mu_t } \Big]\,.
    \end{align*}
    
    \noindent By taking $\norm{\mu_t}$ to be $\mu$ and $\langle \hat{\mu}_t, \mu^*_t\rangle$ to be $\mu^*$, we observe the similarity between the right side of the above equation and the one-dimensional definition of $G$ defined in Eq.~\eqref{eq:one-d-g-def}. Using Lemma \ref{lemma:G-contraction} and if $\norm{\mu_t} \in [c, \frac{4 \dtp{ \hat{\mu}_t }{ \mu_t^* } }{3}]$, we have
    \begin{align*}
        \dtp{\hat{\mu_t}}{ G(\mu_t, \mu_t^*) } \leq 0.01 \abs{\dtp{ \hat{\mu_t} }{ \mu_t } - \dtp{ \hat{\mu_t} }{ \mu_t^* }}
    \end{align*}
    Taking the dot product of $G(\mu_t, \mu_t^*)$ with $v_2 = \muhp_t$, we have
    \begin{align*}
        \dtp{ \muhp_t }{ G(\mu_t, \mu_t^*) } = \; & \mbb{E}_{\alpha_1, \alpha_2, \ldots, \alpha_d \sim \mc{N}(0, 1)} \Big[ - \frac{1}{2} \tanh''( \; \norm{ \mu_t } ( \alpha_1 + \hat{\mu}_t^\top \mu_t^*) ) \norm{ \mu_t }^2 ( \alpha_2 + \dtp{ \muhp_t }{\mu_t^*} ) \\ 
        &\quad + \tanh'( \; \norm{ \mu_t } ( \alpha_1 + \hat{\mu}_t^\top \mu_t^*) ) \norm{ \mu_t } ( \alpha_1 + \hat{\mu}_t^\top \mu_t^*) ( \alpha_2 + \dtp{ \muhp_t }{\mu_t^*} ) \Big] \\
         = \; & \mbb{E}_{\alpha_1 \sim \mc{N}( \hat{\mu}_t^\top \mu_t^* , 1)} \Big[ - \frac{1}{2} \tanh''( \; \norm{ \mu_t }  \alpha_1 ) \norm{ \mu_t }^2 \dtp{ \muhp_t }{\mu_t^*} \\
         & \quad + \tanh'( \; \norm{ \mu_t } \alpha_1  ) \norm{ \mu_t } \alpha_1 \dtp{ \muhp_t }{\mu_t^*} \Big] \\
         = \; & \dtp{ \muhp_t }{\mu_t^*} \,\mbb{E}_{\alpha_1 \sim \mc{N}( \hat{\mu}_t^\top \mu_t^* , 1)} \Big[ - \frac{1}{2} \tanh''( \; \norm{ \mu_t }  \alpha_1 ) \norm{ \mu_t }^2 + \tanh'( \; \norm{ \mu_t } \alpha_1  ) \norm{ \mu_t } \alpha_1 \Big] \,.
    \end{align*}
    In Lemma \ref{lemma:small-expectation} below, we show that when $\| \mu_t \| \in [c, \frac{4 \dtp{ \hat{\mu}_t }{ \mu_t^* } }{3}]$, the expectation in the last expression is upper bounded by 0.01. Therefore, we have
    \begin{align*}
         \big| \dtp{ \muhp_t }{ G(\mu_t, \mu_t^*) } \big| \leq 0.01  | \dtp{ \muhp_t }{\mu_t^*} | \implies \big| \dtp{ \muhp_t }{ G(\mu_t, \mu_t^*) } \big| \leq 0.01 \big| \dtp{ \muhp_t }{\mu_t - \mu_t^*} \big|
    \end{align*}
    Observe that for $i=3,\ldots,d$, $\dtp{G(\mu_t, \mu_t^*) }{v_i} = 0$. Therefore, we have
    \begin{align*}
        \norm{G(\mu_t, \mu_t^*)}^2 &= \sum_{i=1}^d \dtp{ v_i }{ G(\mu_t, \mu_t^*) }^2 \leq 0.01^2 \norm{ \mu_t - \mu_t^* }^2\,.\qedhere
    \end{align*}
\end{proof}

\noindent The next Lemma ensures that the parameter $\mu_t$ after a few steps of gradient descent on the DDPM objective stays in the region where the function $G$ satisfies $\norm{G(\mu_t, \mu_t^*)} \leq 0.01 \norm{ \mu_t - \mu_t^* }$. Recall that the condition of the Lemma is satisfied because we initialize at the warm start obtained by gradient descent in the high noise regime. 

\begin{lemma}
\label{lemma:mu-in-interval}
    Suppose the angle between initialization $\hat{\mu}_t^{(0)}$ and optimal parameter $\mu_t^*$ is $\Theta(1)$, then for any $h$, we have $\| \mu_t^{(h)} \| \in [c, \frac{4 \dtp{ \hat{\mu}_t^{(h)} }{ \mu_t^* } }{3}]$.
\end{lemma}

\noindent The proof of Lemma \ref{lemma:mu-in-interval} is given in Appendix \ref{subsec:mu-in-interval-proof}. Finally, we are ready to prove the main result of this section:

\begin{proof}[Proof of Theorem~\ref{thm:mo2g-const-sep-appendix}]
    To obtain the contraction of $\wnorm{ \mu_t^{(h)} - \mu_t^* }$ after a gradient descent step on the DDPM objective, we write $\wnorm{ \mu_t^{(h+1)} - \mu_t^* }$ in terms of $\wnorm{ \mu_t^{(h)} - \mu_t^* }$ as follows:
    \begin{align*}
        & \big\| \mu_t^{(h+1)} - \mu_t^* \big\| = \big\| \mu_t^{(h)} - \eta \nabla L_t(s_{\mu_t^{(h)}}) - \mu_t^* \big\| + \eta \bigg\| \Big( \frac{1}{n} \sum_{i=1}^n \nabla L_t(s_{\mu_{t}^{(h)}} (x_{i}, z_{i}) ) \Big) - \nabla L_t(s_{\mu_{t}^{(h)}} ) \bigg\| \\
     & \leq   (1 - \eta) \| \mu_t^{(h)} - \mu_t^* \| +\eta \,\big\| \mbb{E}_{x \sim \mc{N}(\mu_t^*, 1)}  [ (\tanh (\mu_t^{(h)^\top} x)  ) x ] - \mu_t^* \big\| + \eta \| G(\mu_t^{(h)}, \mu_t^*) \| + \eta \epsilon\,,
    \end{align*}
    where in the last step we used Lemma~\ref{lemma:sample-complexity-k-mog} below to bound the distance between the population and empirical gradient.
    
    Recall that gradient descent in the low noise regime was initialized using the output of the gradient descent in the high noise regime. Therefore, $\dtp{\hat{\mu}_t^{(0)}}{\hat{\mu}_t^*} \gtrsim 1$. Using Lemma \ref{lemma:mu-in-interval}, we know that the condition on Lemma \ref{lemma:G-contraction} is always satisfied. Using the contractivity of $G$ established in Lemma \ref{lemma:G-contraction} combined with \cite[Theorem 2]{daskalakis2017ten}, and choosing $\eta = 0.05$, we conclude that the distance to the ground truth contracts:
    \begin{align*}
        \big\| \mu_t^{(h+1)} - \mu_t^* \big\| & \leq (1 - 0.05) \big\| \mu_t^{(h)} - \mu_t^* \big\| + 0.01 \big\| \mu_t^{(h)} - \mu_t^* \big\| + 0.01 \big\| \mu_t^{(h)} - \mu_t^* \big\| + \eta \epsilon \\
        & \leq 0.97 \big\| \mu_t^{(h)} - \mu_t^* \big\| + \eta \epsilon. \\
    \end{align*}
    Applying the above for all $h \in [H]$, we obtain
    \begin{align*}
        \wnorm{ \mu_t^{(H)} - \mu_t^* } \leq 0.97^H \wnorm{ \mu_t^{(0)} - \mu_t^* } + 50\epsilon.
    \end{align*}
    The choice of $H$ given in the Theorem statement proves the result. 
\end{proof}

\section{Learning mixtures of two Gaussians with small separation}
\label{appsec:proof-2-mog-small-sep}

In this section, we extend the analysis for learning mixtures of two Gaussians with constant separation, provided in Section \ref{appsec:proof-2-mog-const-sep}, to the low-separation regime and prove the following:

\begin{theorem}[Formal version of Theorem \ref{thm:2-mog-small-sep}]
\label{thm:mo2g-small-sep-appendix}
For any $\mc{L} > 0$, let $q$ be a mixture of two Gaussians (in the form of Eq.~\eqref{eq:mog-2-pdf}) with mean parameter $\mu^*$ satisfying $\norm{\mu^*} > \mc{L}$. Recalling that $B$ denotes an \emph{a priori} upper bound on $\norm{\mu^*}$, we have that for any $\epsilon \le \epsilon'$, where $\epsilon' \lesssim \frac{1}{d^2 B^9}$, there exists a procedure satisfying the following. If the procedure is run for at least $\mathrm{poly}(d, B, \frac{1}{\mc{L}})\frac{1}{\varepsilon^3}$ iterations with at least $\mathrm{poly}(d, B, \frac{1}{\mc{L}})*\frac{1}{\varepsilon^{8}}$ samples from $q$, then it outputs $\Tilde{\mu}$ such that $\wnorm{ \Tilde{\mu} - \mu^* } \leq \epsilon$ with high probability.
\end{theorem}

\noindent As described in Section~\ref{sec:overview}, the algorithm is a simple modification of Algorithm~\ref{alg:denoise} in which gradient descent is replaced by projected gradient descent. We start in Lemma~\ref{lemma:correct-norm} by showing that the projection step in the algorithm ensures that the norm of the current iterate $\mu_t$ is approximately that of $\mu^*_t$. Then in Lemma~\ref{lemma:convergence-small-sep-2-mog}, we extend the analysis of Lemma \ref{lemma:projection-angle-decrease} to show that every projected gradient step contracts the distance to the ground truth. Combined with Lemma \ref{lemma:correct-norm}, this allows us to conclude the proof of Theorem \ref{thm:2-mog-small-sep}.

\begin{lemma}
\label{lemma:correct-norm}
    Let $x_1,\ldots,x_n$ be independent samples from $q$, and define radius parameter $R$ by $R^2 \triangleq \frac{1}{n} \sum^n_{i=1} \wnorm{x_i}^2 - d$. For any $\epsilon > 0$, provided that $n\gtrsim \frac{B^4 + d^2}{\epsilon^2 \mc{L}^2}$, 
    we have $|R - \wnorm{\mu^*}| \leq \varepsilon$ with high probability.
\end{lemma}
\begin{proof}
    Observe that we can write the random variable corresponding to the mixture of two Gaussians $X_0 = X = Z + p \mu^*$ where $Z\sim \mc{N}(0, I)$ and $p$ is a Rademacher random variable. Using Theorem 3.1.1 (concentration of norms) from \cite{vershynin-HDP-book}, we know that $\wnorm{ \wnorm{Z} - \sqrt{d} }_{\psi_2} \lesssim 1$. Therefore, sub-Gaussian norm $\norm{ \wnorm{X_0} }_{\psi_2} \lesssim  \norm{ \wnorm{Z} }_{\psi_2} + \norm{ \wnorm{p \mu^*} }_{\psi_2} \lesssim B + \sqrt{d}.$ Using Lemma 2.7.4 from \cite{vershynin-HDP-book}, we have $\norm{ \wnorm{X_0}^2 }_{\psi_1} \lesssim \norm{ \wnorm{X_0} }_{\psi_2}^2 \lesssim B^2 + d$. Therefore, using number of samples $n$ specified in the Lemma statement, with high probability, we have
    \begin{align*}
        \Big| \frac{1}{n} \sum_{i=1}^n \norm{x_i}^2 - \mbb{E}[\wnorm{X_0}^2 ] \Big| \leq \varepsilon \mc{L} \implies \Big| \wnorm{ \mu }^2 - \wnorm{ \mu^* }^2 \Big| \leq  \varepsilon \mc{L} \implies \big| \wnorm{ \mu } - \wnorm{ \mu^* } \big| \leq \varepsilon
    \end{align*}
    where the penultimate implication uses the fact that $\mbb{E}_{X_0}[ \wnorm{X_0}^2 ] = \mbb{E}[ \wnorm{Z}^2 + \wnorm{\mu^*}^2 ] =  d + \norm{\mu^*}^2$.
\end{proof}

\begin{lemma}
\label{lemma:convergence-small-sep-2-mog}
    Assume that $\mc{L} \leq \wnorm{\mu^*} \leq B$. Then, for any small $\varepsilon > 0$, running projected GD on diffusion models with step size $\eta = \frac{1}{20}$ at noise scale $t = \log \frac{d}{\varepsilon}$ for number of steps $H > H'$ and number of samples $n > n'$ steps will achieve 
    \begin{align*}
        \big\| \mu^{(H)} - \mu^* \big\| \lesssim d^2 B^4 \varepsilon,
    \end{align*}
     where $H' = \frac{d^2}{ \mc{L}^2 \varepsilon^3 }$ and $n' = \frac{d^{10} B^{3}}{\epsilon^8 \mc{L}^6}$. 
\end{lemma}

\begin{proof}
    Recalling that $\mu^*_t = \mu^*_0 \exp(-t)$, note that for $t = \log \frac{d}{ \varepsilon }$, $\frac{\varepsilon \mc{L}}{d} \leq \norm{\mu_t^*} \leq \frac{\varepsilon B}{d}$.
    We would like to apply Lemma~\ref{lemma:projection-angle-decrease}. Note that we may apply this even though it is only stated for gradient descent (without projection). The reason is that it bounds the change in angle between the iterate and the ground truth after a single gradient step, and this angle is unaffected by projection.
    
    Suppose we take one projected gradient step with learning rate $\eta$ starting from an iterate $\mu_t$. As $\mu_t$ was the result of a projection, by Lemma \ref{lemma:correct-norm} we have $\frac{\varepsilon \mc{L} }{d} \lesssim \big\| \mu_t^{(h)} \big\| \lesssim \frac{\varepsilon B}{d}$. 
    
    We now bound $\kappa_2$ in Lemma \ref{lemma:projection-angle-decrease}:
\begin{align*}
    \kappa_2 &= \frac{ 500 \eta \sqrt{d^3} \| \mu_t \|^4 + 20 \eta d \| \mu_t \|^2 \norm{\mu_t^*}^2 + \eta \Tilde{\epsilon} }{ \norm{ \mu_t^* }^2 } \\ 
    &\lesssim 500 \eta \sqrt{d^7} \| \mu_t \|^2 + 20 \eta d \| \mu_t \|^2 + \frac{d^2 \epsilon}{\norm{\mu_t^{*}}^3 } \\
    &\leq 550  d^{7/2} B^2 \exp(-2t) + \frac{d^5 \epsilon}{\varepsilon^3 \mc{L}^3}\\
    &\lesssim d^2 B^2 \varepsilon,
\end{align*}
where the last inequality follows by choosing population gradient estimation error parameter $\epsilon = \frac{ \varepsilon^4 \mc{L}^3 }{ d^3 }$ with the number of samples $n' = \frac{d^{11} B^{6}}{\epsilon^8 \mc{L}^6}$. Additionally, $\kappa_1$ in Lemma \ref{lemma:projection-angle-decrease} is given by
\begin{align*}
    \kappa_1 &= \frac{ 1 - 3 \eta \|\mu_t\|^2 }{ (1 -3\eta \|\mu_t \|^2)  + \eta( \|\mu_t^*\|^2   - 500 \sqrt{d^3} \| \mu_t \|^4 - 20 d  \|\mu_t\|^2 \norm{\mu_t^*}^2 - \Tilde{\epsilon} ) } \\
    & = \frac{ 1 - 3 \eta \|\mu_t\|^2 }{ (1 -3\eta \|\mu_t \|^2)  + \eta \|\mu_t^*\|^2 ( 1 - \kappa_2 ) } \\
    & \lesssim \frac{ 1 - 3 \eta \|\mu_t^{(h)}\|^2 }{ (1 -3\eta \|\mu_t^{(h)} \|^2)  + \eta \|\mu_t^*\|^2 ( 1 - d^2 B^2\varepsilon ) } \\
    & \leq \frac{ 1 }{ 1  + \frac{ \mc{L}^2 \varepsilon^2}{20 d^2} ( 1 - d^2 B^2\varepsilon ) }\,.
\end{align*}
Using bounds on $\kappa_1$ and $\kappa_2$ and Lemma \ref{lemma:projection-angle-decrease}, we conclude that if $\theta$ (resp. $\theta'$) is the angle between $\mu_t$ (resp. the next iterate of projected gradient descent after $\mu_t$) and $\mu^*_t$
\begin{align*}
    \tan \theta' \leq \max \Big( \; \frac{ 1  }{ 1 + \frac{\mc{L}^2 \varepsilon^2}{20 d^2} ( 1 - B^2\varepsilon ) } \tan \theta, d^2 B^2 \varepsilon \Big)\,.
\end{align*}
Doing projected gradient descent for $H = \frac{20 d^2}{\mc{L}^2 \epsilon^3}$ steps, if $\theta^{(h)}$ denotes the angle between the $h$-th iterate and $\mu^*_t$, we obtain
\begin{align*}
    \tan \theta^{(H)} &\leq \tan \theta^{(h+1)} \leq \max \bigg( \Big( \frac{ 1  }{ 1 + \frac{\mc{L}^2 \varepsilon^2}{20 d^2} ( 1 - d^2 B^2\varepsilon ) } \Big)^H \tan \theta^{(0)}, d^2 B^2 \varepsilon \bigg) \\
    &\leq \max \bigg(  \frac{ \tan\theta^{(0)}  }{ 1 + \frac{H \mc{L}^2 \varepsilon^2}{20 d^2} ( 1 - B^2\varepsilon ) } , d^2 B^2 \varepsilon \bigg) \leq d^2 B^2 \varepsilon\,,
\end{align*}
where the last inequality uses $1 + \frac{H \mc{L}^2 \varepsilon^2}{20 d^2} ( 1 - B^2\varepsilon ) \geq \frac{1}{\varepsilon}$ for $\varepsilon \lesssim \frac{1}{B^3}$. Additionally, for a random initialization, Lemma~\ref{lemma:random-init-similarity} shows that $\cos \theta^{(0)} \geq \frac{1}{2d}$ which implies $\tan \theta^{(0)} \leq \sqrt{\sec^2\theta^{(0)} - 1 } \lesssim d$. Using Lemma \ref{lemma:correct-norm}, we have $\wnorm{ \mu^{(H)} } \geq \wnorm{ \mu^* } - \epsilon$ which implies $-2\wnorm{ \mu^{(H)} } \wnorm{ \mu^* } \cos \theta^{(H)} \leq -2\wnorm{ \mu^* }^2 \cos \theta^{(H)} + 2B \epsilon$ and $\wnorm{ \mu^{(H)} }^2 \leq \wnorm{ \mu^* }^2 + 3 B \epsilon$. Using this result, we obtain
\begin{align*}
    \wnorm{ \mu^{(H)} - \mu^* }^2 &= \wnorm{ \mu^{(H)} }^2 + \wnorm{ \mu^* }^2 - 2 \wnorm{ \mu^{(H)} } \wnorm{ \mu^* } \cos \theta^{(H)} \\
    &\lesssim 2 \wnorm{ \mu^* }^2 - 2 \wnorm{ \mu^* }^2 \cos \theta^{(H)} + 5 B \varepsilon  \lesssim 2B^2 \Big(1 - \frac{1}{\sqrt{1 + d^4 B^4 \varepsilon^2}}\Big) + 5 B \varepsilon \lesssim d^2 B^4 \varepsilon,
\end{align*}
where the last inequality follows from the fact that $\sqrt{1+x} \leq 1+\sqrt{x}$ for any $x > 0$. 
\end{proof}

\section{Learning mixtures of $K$ Gaussians from a warm start}
\label{appsec:proof-k-mog}

In this section, we provide details about our main result on learning mixtures of $K$ Gaussians. We start by describing our main theorem in this case.

\begin{theorem}[Formal version of Theorem \ref{thm:mog-k-main}]
\label{thm:k-mog-appendix}
Let $q$ be a mixture of Gaussians (in the form of Eq.~\eqref{eq:mog-k-pdf}) with center parameters $\theta^* = \{\mu_1^*, \mu_2^*, \ldots, \mu_K^*\}\in \mbb{R}^d$ satisfying the separation Assumption~\ref{asm:mog-k-seperation}, and suppose we have estimates $\theta$ for the centers such that the warm initialization Assumption~\ref{asm:mog-k-initialization} is satisfied. For any $\epsilon > \epsilon_0$ and noise scale $t$ where 
\begin{align*}
 \epsilon_0 =  1 / \mathrm{poly}(d) \;\; \text{and} \;\; t = \Theta(\epsilon)\,,
\end{align*}
gradient descent on the DDPM objective at noise scale $t'$ (Algorithm \ref{alg:denoise}) outputs $\Tilde{\theta} = \{ \Tilde{\mu}_1, \Tilde{\mu}_2, \ldots, \Tilde{\mu}_K \}$ such that $\min_i \wnorm{ \Tilde{\mu}_i - \mu_i^* } \leq \epsilon$ with high probability. The algorithm runs for $H \geq H'$ iterations and uses $n \geq n'$ number of samples where
\begin{align*}
    H' =  \Theta(\log( \epsilon^{-1}\log d))  \;\; \text{ and } \;\; n'= \Theta(K^4 d^5 B^6 / \epsilon^2)\,.
\end{align*}
\end{theorem}

\noindent We first give an overview of the proof for population gradient descent, and then show that the empirical gradients concentrate well around the population gradients. We start by simplifying the population gradient update for mixtures of $K$ Gaussians using Stein's lemma in Lemma~\ref{lemma:GD-update-k-mog}, which yields
\begin{align*}
    - \nabla_{\mu_{1,t}} L_t( s_{\theta_t} )  = \mbb{E} [ w_{1, t}(X_t) (X_t - \mu_{1, t}) ] + [\text{extra terms}]\,,
\end{align*}
recalling the notation of Eq.~\eqref{eq:score-function-definition}.
As discussed in the body of the paper, $\mbb{E} [ w_{1, t}(X_t) (X_t - \mu_{1, t}) ]$ is precisely the update for the gradient EM algorithm (see Fact~\ref{fact:gradient-EM}) and known results for the latter~\cite{kwon2020algorithm, segol2021improved} can be used to show that the distance $\| \mu_{1, t} - \mu_{1, t}^* \|$ contracts in each step when the separation Assumption~\ref{asm:mog-k-seperation} and the warm initialization Assumption~\ref{asm:mog-k-initialization} are satisfied. Therefore, showing that the ``extra terms'' do not disturb the progress coming from the gradient EM update is sufficient. We prove that the ``extra terms'' are $1/\mathrm{poly}(d)$ in Lemma~\ref{lemma:population-GD-gradient-EM} when the separation Assumption~\ref{asm:mog-k-seperation} and warm initialization Assumption~\ref{asm:mog-k-initialization} hold.

The intuition behind Lemma~\ref{lemma:population-GD-gradient-EM} is as follows: We start with a key observation that each of the ``extra terms'' either contains $w_{1, t}(X_t)(1 - w_{1, t}(X_t))$ or $w_{1, t}(X_t) w_{j, t}(X_t)$ where $j \neq 1$. Note that the $w_{1, t}(X_t)$ can be interpreted as the conditional probability of the underlying component being $\mc{N}(\mu_{1, t}, I)$ given $X_t$. When Assumption~\ref{asm:mog-k-seperation} and Assumption~\ref{asm:mog-k-initialization} are satisfied, Proposition 4.1 of \cite{segol2021improved} shows that 
\begin{align*}
    \mbb{E}_{X_t \sim \mc{N}(\mu_{1, t}^*, I)}[ w_{j, t}(X_t) ] \lesssim 1/\mathrm{poly}(d) \quad \text{for any $j \neq 1$}\,.
\end{align*}
This result can be extended to show both $\mbb{E}_{X_t} [w_{1, t}(X_t)(1 - w_{1, t}(X_t)) ] \lesssim 1/\mathrm{poly}(d)$ as well as $\mbb{E}_{X_t}[ w_{1, t}(X_t) w_{j, t}(X_t)]\lesssim  1/\mathrm{poly}(d)$ for any $j \neq 1$ (see Lemma~\ref{lemma:EM-initialization-properties} for the proof). Using these bounds, we conclude that $[``\text{extra terms}''] \lesssim 1/\mathrm{poly}(d)$ in Lemma~\ref{lemma:population-GD-gradient-EM}.

\subsection{EM and population gradient descent on DDPM objective}

We begin by writing out the gradient update explicitly:

\begin{lemma} 
\label{lemma:GD-update-k-mog}
For any noise scale $t > 0$, the gradient of the population DDPM objective $\mbb{E} [ L_t( s_{\theta_t}(X_t) ) ]$ with respect to parameter $\mu_{1, t}$ is given by
\begin{align*}
    \nabla_{\mu_{1, t}}  L_t( s_{\theta_t} )  &= \mbb{E} \Big[ - w_{1, t}(X_t) (X_t - \mu_{1, t}) + w_{1, t}(X_t) (X_t - \mu_{1, t}) \sum^K_{i=1} w_{i, t}(X_t) \mu_{i, t}^\top (X_t - \mu_{1, t}) \\
    & \hspace{13mm} + w_{1, t}(X_t) \mu_{1, t} - w_{1, t}(X_t) (X_t  - \mu_{1, t})^\top \mu_{1, t} (X_t - \mu_{1, t}) - w_{1, t}(X_t) \sum^K_{i=1} w_{i, t}(X_t)  \mu_{i, t} \\
    & \hspace{13mm} - w_{1, t}(X_t) \sum^K_{i=1} \nabla_{x} w_{i, t} (X_t)^\top \mu_{i, t} (X_t - \mu_{1, t}) \Big]
\end{align*}
where $w_{1, t}(x)$ and $\mu_{1, t}$ are defined in Eq.~\eqref{eq:score-function-definition}.
\end{lemma}
\begin{proof}
Recall that the score function of mixture of Gaussians is given by
\begin{align}
    s_{\theta_t}(X_t) = \sum_i w_{i, t}( X_t ) \mu_{i, t} - X_t  \hspace{5mm}
\end{align}
Finding the gradient $\nabla_{\mu_{1, t}} w_{i,t}(X_t)$, we have 
\begin{align*}
    \nabla_{\mu_{1, t} } w_{i, t}(X_t) = \begin{cases} w_{1, t}(X_t) (1 - w_{1,t }(X_t)) (X_t - \mu_{1, t})  & \text{if } i=1 \\
    - w_{1, t}(X_t) w_{i, t}(X_t)  (X_t-\mu_{1, t})  & \text{otherwise}.
    \end{cases}
\end{align*}
The gradient of the score function is given by
\begin{align*}
    & \nabla_{\mu_{1, t}} s_{\theta_t}(X_t) =  \nabla_{\mu_{1, t}} \rb{ w_{1, t}(X_t) \mu_{1,t} } + \sum_{i=2}^K \nabla_{\mu_{1, t}} \rb{ w_{i, t}(X_t) \mu_{i, t} }  \\
    &= w_{1, t}(X_t)(1 - w_{1, t}(X_t)) \mu_{1, t} (X_t - \mu_{1, t})^\top + w_{1, t}(X_t) I - w_{1, t}(X_t) \sum_{i=2}^K w_{i, t}(X_t) \mu_{i, t} (X_t - \mu_{1, t})^\top  \\
    &=  w_{1, t}(X_t) \mu_{1, t} (X_t - \mu_{1, t})^\top  + w_{1, t}(X_t) I - w_{1, t}(X_t) \sum_{i=1}^K w_{i, t}(X_t) \mu_{i, t} (X_t - \mu_{1, t})^\top\,.
\end{align*}
The gradient of $ \frac{1}{2} \wnorm{s_{\theta_t}}^2$ is given by
\begin{align*}
    & \frac{1}{2} \nabla  \norm{s_{\theta_t}(X_t)}^2 = \sum_{j=1}^d [ s_{\theta_t}(X_t) ]_j [\nabla_{\mu_{1, t}} s_{\theta_t}(X_t) ]_j = \nabla_{\mu_{1, t}} s_{\theta_t}(X_t)^\top s_{\theta_t}(X_t) \\ 
    & \hspace{3cm} \text{ where } [\nabla_{\mu_{1, t}} s_{\theta_t}(X_t) ]_j \text{ is $j^{th}$ row of } \nabla_{\mu_{1, t}} s_{\theta_t}(X_t)\,.
\end{align*}
The gradient of this is given by
\begin{align}
    \frac{\nabla_{\mu_{1,t}} s_{\theta_t}(X_t)^\top Z_t}{ \beta_t } &= \frac{1}{ \beta_t } \Big( w_{1, t}(X_t) (X_t - \mu_{1, t}) \mu_{1, t}^\top Z_t + w_{1, t}(X_t) Z_t \\
    &\qquad\qquad\qquad - w_{1, t}(X_t) \sum_{i=1}^K w_{i, t}(X_t) (X_t - \mu_{1, t}) \mu_{i, t}^\top Z_t  \Big) \label{eq:loss-grad-second-term}
\end{align}
Applying Stein's lemma to the expectation of the first term in Eq.~\eqref{eq:loss-grad-second-term}, we have
\begin{equation}
    \label{eq:loss-grad-second-term-first}
    \begin{aligned}
        \mbb{E}_{X_0, Z_t} [ w_{1, t}(X_t) (X_t - \mu_{1, t}) \mu_{1, t}^\top Z_t ] &= \sum_{j=1}^d \mbb{E}_{X_0, Z_t} [ w_{1, t}(X_t) (X_t - \mu_{1, t}) \mu_{1, t, j} Z_{t, j} ] \\
        &= \sum_{j=1}^d \mbb{E}_{X_0, Z_t} [ w_{1, t}(X_t) \beta_t e_j \mu_{1, t, j} + \beta_t \nabla_{x} w_{1, t}(X_t)^\top e_j (X_t - \mu_{1, t}) \mu_{1, t, j} ] \\
        &= \mbb{E}_{X_0, Z_t} [ w_{1, t}(X_t) \beta_t \mu_{1, t} + \beta_t \nabla_{x} w_{1, t}(X_t)^\top \mu_{1, t} (X_t - \mu_{1, t}) ]
    \end{aligned}
\end{equation}
The expectation of the second term in Eq.~\eqref{eq:loss-grad-second-term} simplifies to $\beta_t \mbb{E}_{X_t}[ \nabla_x w_{1, t}(X_t) ]$ by Stein's Lemma. Each summand in the third term in Eq.~\eqref{eq:loss-grad-second-term} simplifies as following: 
\begin{align}
        \MoveEqLeft\mbb{E}_{X_0, Z_t} \sbr{ w_{1, t}(X_t) w_{i, t}(X_t) (X_t - \mu_{1, t}) \mu_{i, t}^\top Z_t  } \\ 
        &= \sum_{j=1}^d \mbb{E}_{X_0, Z_t} \sbr{ w_{1, t}(X_t) w_{i, t}(X_t) (X_t - \mu_{1, t}) \mu_{i, t, j} Z_{t, j} } \\
        &= \sum_j  \mu_{i, t, j} \mbb{E}_{X_0, Z_t} \Big[ w_{1, t}(X_t) w_{i, t}(X_t)  \beta_t e_{j}   + \beta_t w_{1, t}(X_t) \nabla_{x} w_{i, t} (X_t)^\top e_j (X_t - \mu_{1, t}) \\ 
        &\qquad\qquad\qquad\qquad\qquad+ \beta_t \nabla_x w_{1, t}(X_t)^\top e_j w_{i, t}(X_t) (X_t - \mu_{1, t})\Big] \\
        &= \beta_t\, \mbb{E}_{X_0, Z_t} \Big[ w_{1, t}(X_t) w_{i, t}(X_t)  \mu_{i, t} + w_{1, t}(X_t) \nabla_{x} w_{i, t} (X_t)^\top \mu_{i, t} (X_t - \mu_{1, t}) \\ 
        &\qquad\qquad\qquad\qquad\qquad+ \nabla_x w_{1, t}(X_t)^\top \mu_{i, t} w_{i, t}(X_t) (X_t - \mu_{1, t})\Big] \label{eq:loss-grad-second-term-third}
\end{align}
Combining the gradients of all the terms of Eq.~\eqref{eq:loss-grad-second-term-third}, we have
\begin{align*}
    \MoveEqLeft \nabla_{\mu_{1,t}} L_t( s_{\theta_t} ) \\
    &= \mbb{E} \Big[ w_{1, t}(X_t) (X_t - \mu_{1, t}) \mu_{1, t}^\top s_{\theta_t}(X_t) + w_{1, t}(X_t) s_{\theta_t}(X_t) - w_{1, t}(X_t) (X_t - \mu_{1, t}) \sum_i w_{i, t}(X_t) \mu_{i, t}^\top s_{\theta_t}(X_t) \\
    & \quad\quad\quad+ \nabla_x w_{1, t}(X_t) + w_{1, t}(X_t) \mu_{1, t} + \nabla_{x} w_{1, t}(X_t)^\top \mu_{1, t} (X_t - \mu_{1, t}) - w_{1, t}(X_t) \sum_i w_{i, t}(X_t)  \mu_{i, t} \\
    & \quad\quad\quad- w_{1, t}(X_t) \sum_i \nabla_{x} w_{i, t} (X_t)^\top \mu_{i, t} (X_t - \mu_{1, t}) - \sum_i \nabla_x w_{1, t}(X_t)^\top \mu_{i, t} w_{i, t}(X_t) (X_t - \mu_{1, t}) \Big] \\
    &= \mbb{E} \Big[ - w_{1, t}(X_t) (X_t - \mu_{1, t}) + w_{1, t}(X_t) (X_t - \mu_{1, t}) \sum_i w_{i, t}(X_t) \mu_{i, t}^\top (X_t - \mu_{1, t}) \\
    & \quad\quad\quad + w_{1, t}(X_t) \mu_{1, t} - w_{1, t}(X_t) (X_t  - \mu_{1, t})^\top \mu_{1, t} (X_t - \mu_{1, t}) - w_{1, t}(X_t) \sum_i w_{i, t}(X_t)  \mu_{i, t} \\
    & \quad\quad\quad - w_{1, t}(X_t) \sum_i \nabla_{x} w_{i, t} (X_t)^\top \mu_{i, t} (X_t - \mu_{1, t}) \Big]\,,
\end{align*}
where the last equality uses~Lemma \ref{lemma:gradient-wrt-x-w}. Specifically, it uses 
\begin{align*}
\nabla_x w_{1, t}(X_t) + w_{1, t}(X_t) s_{\theta_t}(X_t) &= - w_{1, t}(X_t)  (X_t - \mu_{1, t}) \\
(\nabla_{x} w_{1, t}(X_t) + w_{1, t}(X_t)s_{\theta_t}(X_t) )^\top \mu_{1, t} (X_t - \mu_{1, t}) &= -w_{1, t}(X_t)( X_t - \mu_{1,t} )^\top \mu_{1, t} (X_t - \mu_{1, t})\,. \qedhere
\end{align*}
\end{proof}

\noindent We will also need the following intermediate calculation:

\begin{lemma}
\label{lemma:gradient-wrt-x-w}
    For any $i \in [K]$, the gradient of $w_{i, t}(X_t)$ with respect to $X_t$ is given by
    \begin{align*}
        \nabla_x w_{i, t}(X_t) &= -w_{i, t}(X_t) (X_t - \mu_{i, t}) - w_{i, t}(X_t) s_{\theta_t}(X_t) \\ 
        &= - w_{i, t}(X_t) (1 - w_{i, t}(X_t) ) (X_t - \mu_{i, t}) + w_{i, t}(X_t) \cdot \sum_{j\in[K]: j\neq i} w_{j, t}(X_t) (X_t - \mu_{j, t} )\,.
    \end{align*}
\end{lemma}
\begin{proof}
    By taking the gradient of $w_{i, t}(X_t)$ and simplifying it, we get the result:
    \begin{align*}
        \nabla_x w_{i, t}(X_t) &= - \frac{ \exp \Big( -\frac{ \norm{X_t - \mu_{i, t} }^2 }{ 2 } \Big) (X_t - \mu_{i, t}) }{ \sum_{j=1}^K \exp \Big( -\frac{ \norm{X_t - \mu_{j, t} }^2 }{ 2\sigma^2 } \Big) } \\ 
        &\qquad\qquad\qquad + \frac{ \exp \big( -\frac{ \norm{X_t - \mu_{i, t} }^2 }{ 2 } \big)\cdot \sum_{j=1}^K \exp \Big( -\frac{ \norm{X_t - \mu_{j, t} }^2 }{ 2 } \Big) (X_t - \mu_{j, t}) }{ \rb{ \sum_{j=1}^K \exp \Big( -\frac{ \norm{X_t - \mu_{j, t} }^2 }{ 2 } \Big) }^2 } \\
        &= -w_{i, t}(X_t) (X_t - \mu_{i, t}) + w_{i, t}(X_t) \rb{ \sum_{j=1}^K w_{j, t}(X_t) (X_t - \mu_{j, t} ) } \\
        &= - w_{i, t}(X_t) (1 - w_{i, t}(X_t) ) (X_t - \mu_{i, t}) + w_{i, t}(X_t) \rb{ \sum_{j=1, j\neq i}^K w_{j, t}(X_t) (X_t - \mu_{j, t} ) }\,.\qedhere
    \end{align*}
\end{proof}

\noindent We are now ready to establish the connection between gradient descent on the DDPM objective and the gradient EM update, for mixtures of $K$ Gaussians:

\begin{lemma}
\label{lemma:population-GD-gradient-EM}
    Suppose the centers of the mixture of $K$ Gaussians are well-separated according to Assumption~\ref{asm:mog-k-seperation}, and the parameters $\theta = \{ \mu_1, \mu_2, \ldots, \mu_K \}$ that the student network is initialized to satisfy the warm start Assumption~\ref{asm:mog-k-initialization}. Then, for noise scale $t = O(1)$, gradient descent on the DDPM objective is close to the gradient EM update:
    \begin{align*}
        \big\| \nabla_{\mu_{1,t}} L_t( s_{\theta_t} ) + \mbb{E} [ w_{1, t}(X_t) (X_t - \mu_{1, t}) ] \big\| \lesssim \frac{K^2 B^2}{ d^{ c_r^2/4000 } } = \frac{1}{\mathrm{poly}(d)}\,,
    \end{align*}
    where $c_r$ is a large constant. 
\end{lemma}

\begin{proof}
    Observe that the first term in the expression for the population gradient of the DDPM objective in Lemma \ref{lemma:GD-update-k-mog} is exactly the gradient EM update for the mixture of $K$ Gaussian in Fact \ref{fact:gradient-EM}. To prove the closeness between the GD update and the gradient EM update, we will show that the additional terms in Lemma \ref{lemma:GD-update-k-mog} are small. 
    
    Note that when the ground truth parameters $\theta^* = \{ \mu_1^*, \mu_2^*, \ldots, \mu_K^* \}$ satisfy Assumption~\ref{asm:mog-k-seperation}, $\theta_t^*$ also satisfies Assumption \ref{asm:mog-k-seperation} for $t = O(1)$. Similarly, it is straightforward to show that when the parameters $\theta$ satisfy Assumption \ref{asm:mog-k-initialization}, $\theta_t = \{\mu_{1, t}, \mu_{2, t}, \ldots, \mu_{K, t} \}$ also satisfies the assumption. 
    
    We focus on the $d \leq K$ case for this proof. A similar calculation with projection onto $O(K)$ dimensional subspace of $\mu_{i,t}^*$ will give the result for $d \geq K$ case~\cite{VEMPALA2004841, yan2017convergence}. 
    
    Using Lemma~\ref{lemma:segol-operator-norm-bounds} below, we have 
    \begin{equation*}
        \bigl\|\mbb{E}\big[ w_{1, t}(X_t)(1 - w_{1, t}(X_t)) (X_t - \mu_{1,t})(X_t - \mu_{1,t})^\top \big] \mu_{1, t} \bigr\| 
        \leq \frac{d^2 c_r^2 B}{ d^{ c_r^2/1000 } }, 
    \end{equation*}
    for any $i \in [K]$. We can simplify additional terms as 
    \begin{align*}
        \MoveEqLeft\biggl\| \sum_{i=2}^K \mbb{E} [ w_{1, t}(X_t) w_{i, t}(X_t) (X_t - \mu_{1, t}) (X_t - \mu_{1,t})^\top \mu_{i, t} ] \biggr\| \\ 
        &\leq \sum_{i=2}^K \mbb{E} [ \| w_{1, t}(X_t) w_{i, t}(X_t) (X_t - \mu_{1, t}) (X_t - \mu_{1,t})^\top \mu_{i, t} \| ]  \\ 
        &\leq \sum_{i=2}^K \sqrt{ \mbb{E} \big[ | w_{1, t}(X_t) w_{i, t} (X_t) |^2 \big] \cdot \mbb{E} \big[ \| (X_t - \mu_{1, t}) (X_t - \mu_{1,t})^\top \mu_{i, t} \|^2 \big] } \\
        &\leq \frac{K B^2}{ d^{ c_r^2/2000 } }\,,
    \end{align*}
    where in the last step we used the second part of Lemma~\ref{lemma:EM-initialization-properties}.
    This will allow us to prove that $\| \mbb{E}[w_{1, t}(X_t) (X_t - \mu_{1, t}) \sum_{i=1}^K w_{i, t}(X_t) \mu_{i, t}^\top (X_t - \mu_{1, t}) - w_{1, t}(X_t) (X_t  - \mu_{1, t})^\top \mu_{1, t} (X_t - \mu_{1, t}) ] \|$ is small. 
    
    Using the expression for $\nabla_{x} w_{i, t} (X_t)$ from Lemma \ref{lemma:gradient-wrt-x-w}, we have
    \begin{align*}
        \MoveEqLeft \sum_{i=1}^K w_{1, t}(X_t) \nabla_{x} w_{i, t} (X_t)^\top \mu_{i, t} (X_t - \mu_{1, t}) \\
        &= - \sum_{i=1}^K w_{1,t}(X_t) w_{i, t}(X_t) (1 - w_{i, t}(X_t) )  (X_t - \mu_{1, t}) (X_t - \mu_{i, t})^\top \mu_{i, t} \\ 
        &\quad\quad\quad + \sum_{i=1}^K \sum_{j=1, j\neq i}^K w_{1, t}(X_t) w_{i, t}(X_t) w_{j, t}(X_t) (X_t - \mu_{1, t} ) (X_t - \mu_{j, t} )^\top \mu_{i, t}\,.
    \end{align*}
    The first term can be simplified as follows:
    \begin{align*}
        \MoveEqLeft\biggl\| \sum_{i=1}^K \mbb{E} \Big[ w_{1,t}(X_t) w_{i, t}(X_t) (1 - w_{i, t}(X_t) )  (X_t - \mu_{1, t}) (X_t - \mu_{i, t})^\top \mu_{i, t} \Big] \biggr\| \\ 
        & \leq \sum_{i=1}^K \mbb{E} \big[ \big\|  w_{1,t}(X_t) w_{i, t}(X_t) (1 - w_{i, t}(X_t) )  (X_t - \mu_{1, t}) (X_t - \mu_{i, t})^\top \mu_{i, t} \big\| \big] \\
        &\leq \sum_{i=2}^K \sqrt{ \mbb{E} [ w_{1,t}(X_t)^2 w_{i, t}(X_t)^2 ] \cdot\mbb{E} \big[ (1 - w_{i, t}(X_t) )^2 \cdot\wnorm{ X_t - \mu_{1, t} }^2 \cdot\wnorm{ X_t - \mu_{i, t} }^2 \cdot\wnorm{ \mu_{i, t} }^2 \big] } \\
        & \lesssim \frac{K B^2}{ d^{ c_r^2/4000 } }\,,
    \end{align*}
    where the last inequality follows from $$\mbb{E}\big[ \| X_t - \mu_{1, t} \|^2 \wnorm{ X_t - \mu_{i, t} }^2 \big] \leq \sqrt{ \mbb{E}\big[\wnorm{ X_t - \mu_{1, t} }^4\big] \mbb{E}\big[\wnorm{ X_t - \mu_{i, t} }^4\big] } \lesssim B^2\,. $$
    Similarly, by simplifying the second term, we get
    \begin{align*}
        \MoveEqLeft\sum_{i=1}^K \sum_{j=1, j\neq i}^K \mbb{E} \big[ \big\| w_{1, t}(X_t) w_{i, t}(X_t) w_{j, t}(X_t) (X_t - \mu_{1, t} ) (X_t - \mu_{j, t} )^\top \mu_{i, t} \big\| \big] \\
        &\leq \sum_{i=1}^K \sum_{j=1, j\neq i}^K \sqrt{ \mbb{E} \big[ w_{i, t}^2(X_t) w_{j, t}^2(X_t) \big] \mbb{E} \big[ w_{1, t}^2(X_t) \wnorm{ (X_t - \mu_{1, t} ) (X_t - \mu_{j, t} ) \mu_{i, t} }^2 \big] } \lesssim \frac{K^2 B^2}{ d^{ c_r^2/4000 } }\,,
    \end{align*}
    where the last inequality uses Lemma \ref{lemma:EM-initialization-properties}. Simplifying the following term using Lemma \ref{lemma:EM-initialization-properties}, we have
    \begin{align*}
        \MoveEqLeft \Big\| \mbb{E}[ w_{1, t}(X_t) \mu_{1, t} - w_{1, t}(X_t) \sum_{i=1}^K w_{i, t}(X_t)  \mu_{i, t} ] \Big\| \\ 
        &\leq \sum_{i=2}^K \mbb{E} \big[ \big\| w_{1, t}(X_t) w_{i, t}(X_t) \mu_{i, t} \big\| \big] + \sum_{i=2}^K \mbb{E} \big[ \big\| w_{1, t}(X_t) w_{i, t}(X_t) \mu_{1, t} \big\| \big] \leq \frac{2 K B}{ d^{ c_r^2/200 } }\,.
    \end{align*}
    Combining all the results, we obtain the theorem statement. 
\end{proof}

\noindent The above proof made use of the following two helper lemmas which follow from prior work analyzing EM for learning mixtures of Gaussians:

\begin{lemma}
\label{lemma:EM-initialization-properties}
    There is some absolute constant $c_r > 0$ for which the following holds. For any $\theta = \{ \mu_1, \mu_2, \ldots, \mu_K \}$ such that $\| \mu_i - \mu_i^*  \| \leq \frac{c_r}{4} \sqrt{\log d}$ for all $i \in [K]$ and any $j$ such that $j \neq i$, we have
    \begin{align*}
        \mbb{E}_{X_t \sim \mc{N}( \mu^*_{i, t}, I ) }[ w_{j, t}(X_t) ] \leq \frac{1}{ d^{ c^2_r / 100 } }\,.
    \end{align*}
    Additionally, for any $j \neq k$ such that $j \in [K]$ and $k \in [K]$, we have
    \begin{align*}
        \mbb{E}_{X_t}[ w_{j, t}(X_t) w_{k, t}(X_t) ] \leq \frac{1}{ d^{ c^2_r / 200 } }\,.
    \end{align*}
\end{lemma}
\begin{proof}
    Using Proposition 4.1 from \cite{segol2021improved}, for any $\theta = \{ \mu_1, \mu_2, \ldots, \mu_K \}$ such that $\| \mu_i - \mu_i^*  \| \leq \frac{c_r}{4} \sqrt{\log d}$ for all $i \in [K]$ and $j \neq i$, we have
    \begin{align*}
        \mbb{E}_{X_t \sim \mc{N}( \mu^*_{i, t}, I ) }[ w_{j, t}(X_t) ] \leq \frac{1}{ d^{ c^2_r / 100 } }.
    \end{align*}
    Computing the expectation of the product of the weights $w_{j,t}$ and $w_{k,t}$ for any distinct $j,k$, we have
    \begin{align*}
        \mbb{E}_{X_t}[ w_{j, t}(X_t) w_{k, t}(X_t) ] &= \sum_{i=1}^K \frac{1}{K} \mbb{E}_{x \sim \mc{N}( \mu^*_i, I )}[ w_{j, t}(x) w_{k, t}(x) ] \\ 
        &\leq \frac{1}{K} \sum_{i=1}^K \sqrt{ \mbb{E}_{x \sim \mc{N}( \mu^*_i, I )}[ w_{j, t}(x)^2 ] \mbb{E}_{x \sim \mc{N}( \mu^*_i, I )}[ w_{k, t}(x)^2 ] } \\
        &\leq \frac{1}{ d^{ c_r^2/200 } }
    \end{align*}
    where the last inequality uses the fact that either $i \neq j$ or $i \neq k$ and $w_{j, t}(x)^2 \leq w_{j, t}(x) \leq 1$.
\end{proof}

\begin{lemma}[Lemma 4.3 of \cite{segol2021improved}]
\label{lemma:segol-operator-norm-bounds}
    Suppose $X$ is distributed according to a mixture of $K$ Gaussians with centers $\theta^* = \{\mu^*_1,\ldots,\mu^*_K\}$ as in Eq.~\eqref{eq:mog-k-pdf}. For any $\theta = \{ \mu_1, \mu_2, \ldots, \mu_K \}$ such that $\| \mu_i - \mu_i^*  \| \leq \frac{c_r}{4} \sqrt{\log d}$ for all $i \in [K]$, then for any distinct $i, j \in [K]$, we have
    \begin{align*}
        \norm{ \mbb{E}_X[ w_i(X, \mu) (1 - w_i(X, \mu)) (X - \mu_i)(X - \mu_i)^\top ] }_{\mathsf{op}} &\leq \frac{d^2 c_r^2}{d^{c_r^2/1000}} \\
        \norm{ \mbb{E}_X[ w_i(X, \theta) w_j(x, \theta) (X - \mu_i) (X - \mu_j)^\top ] }_{\mathsf{op}} &\leq \frac{d^2 c_r^2}{d^{c_r^2/1000}}
    \end{align*}
\end{lemma}

\subsection{Closeness between population gradient descent and empirical gradient descent}

In this section, we show that the population gradient descent on the DDPM objective is close to the empirical gradient descent for mixtures of $K$ Gaussians. 

\begin{lemma}
\label{lemma:sample-complexity-k-mog}
    For any $\epsilon$ that is $\Theta(\frac{1}{\mathrm{poly}(d)})$ and noise scale $t > t'$ where $t' \lesssim 1$, the empirical estimate of gradient descent update on the DDPM objective with the number of samples $n > n'$ concentrates well to the population gradient descent update where $n' = O(\frac{K^4 d^5 B^6}{\epsilon^2})$. More specifically, the following inequality holds with probability at least $1 - \exp(-d^{0.99})$:
    \begin{align*}
        \norm{ \nabla_{\mu_{1,t}} \Big( \frac{1}{n} \sum_{i=1}^n L_t(s_{\theta_t}( x_{i, 0}, z_{i, t} )) \Big) - \nabla_{\mu_{1,t}}  L_t(s_{\theta_t} )  } \leq \epsilon.
    \end{align*}
\end{lemma}
\begin{proof}
Recall that the population gradient is given by
\begin{equation}
    \nabla_{\mu_{1,t}} L_t(s_{\theta_t}) = \mbb{E} \Big[ \frac{1}{2} \nabla_{\mu_{1,t}} \norm{ s_{\theta_t}(X_t) }^2 + \frac{\nabla_{\mu_{1,t}} s_{\theta_t}(X_t)^\top Z_t}{ \beta_t }  \Big]\,,
\end{equation}
where
\begin{align} 
     \mbb{E} \Big[ \frac{1}{2} \nabla_{\mu_{1,t}} \norm{ s_{\theta_t}(X_t) }^2 \Big] &= \mbb{E} \bigg[ \Big( w_{1, t}(X_t) (X_t - \mu_{1, t}) \mu_{1, t}^\top   + w_{1, t}(X_t) \cdot \Id \\ 
     & \hspace{1cm} - w_{1, t}(X_t) \sum_{i=1}^K w_{i, t}(X_t) (X_t - \mu_{1, t}) \mu_{i, t}^\top \Big) \cdot  \sum_{i=1}^K \big(w_{i, t}( X_t ) \mu_{i, t} - X_t\big) \bigg]\,,
\end{align}
and 
\begin{align}
     \mbb{E} \Big[ \nabla_{\mu_{1,t}} s_{\theta_t}(X_t)^\top Z_t \Big] &= \mbb{E} \Big[ \Big( w_{1, t}(X_t)  (X_t - \mu_{1, t}) \mu_{1, t}^\top Z_t \\
     &\quad\quad\qquad + w_{1, t}(X_t) Z_t - w_{1, t}(X_t) \sum_{i=1}^K w_{i, t}(X_t) (X_t - \mu_{1, t}) \mu_{i, t}^\top Z_t \Big) \Big]\,. \label{eq:grad-k-mog-loss-recall}
\end{align}
We will prove that the sample estimate of each coordinate in Eq.~\eqref{eq:grad-k-mog-loss-recall} concentrates well around the expectation. We will prove the concentration of the first coordinate and a similar analysis holds for other coordinates. For the rest of the proof, we use $\Tilde{x}_t$ to denote the first coordinate of $X_t$ and $\Tilde{\mu}_{i, t}$ to indicate the first coordinate $\mu_{i, t}$. 
For any random variable $Y \in \mbb{R}$, we use $\| Y \|_{\psi_1}$ to denote the sub-exponential norm of $Y$ and $\| Y \|_{\psi_2}$ to denote the sub-gaussian norm of $Y$ (See lemma \ref{def:subgaussian} for details). Using properties of a sub-Gaussian random variable from Lemma \ref{def:subgaussian}, we get
\begin{align}
    & \Big{\|} \sum_{j=1}^K w_{1, t}(X_t) w_{j, t} (X_t) (\Tilde{x}_t - \Tilde{\mu}_{1, t}) \mu_{1, t}^\top \mu_{j, t} \Big{\|}_{\psi_2} \\ \lesssim  & \; \sum_{j=1}^K \Big\| w_{1, t}(X_t) w_{j, t} (X_t) (\Tilde{x}_t - \Tilde{\mu}_{1, t}) \mu_{1, t}^\top \mu_{j, t} \Big\|_{\psi_2} \tag*{(Using sum of sub-Gaussian random variables property in Lemma \ref{def:subgaussian})} \\ 
    \lesssim  & \; \sum_{j=1}^K \Big\| w_{1, t}(X_t) w_{j, t} (X_t) \mu_{1, t}^\top \mu_{j, t} z \Big\|_{\psi_2} + \Big\| w_{1, t}(X_t) w_{j, t} (X_t) \mu_{1, t}^\top \mu_{j, t} (\tau - \Tilde{\mu}_{1, t} ) \Big\|_{\psi_2} \\
    \lesssim & \; KB^2 + K B^3 \lesssim K B^3, \label{eq:first-term-grad-sgnorm}
\end{align}
where the third inequality follows by writing $\Tilde{x}_t = z + \tau$ where $z \sim \mc{N}(0, 1)$ and $\tau$ is a random variable that takes $\Tilde{\mu}_{i, t}^*$ for every $i \in [K]$ with probability $\frac{1}{K}$. The fourth inequality follows from the sub-Gaussian property of a bounded random variable and the product of a sub-Gaussian random variable with bounded random variable property in Lemma \ref{def:subgaussian}.
Using the sum of sub-Gaussian random variable property in Lemma \ref{def:subgaussian}, we have
\begin{equation}
\label{eq:second-term-grad-sgnorm}
    \Big\| \sum_{i=1}^K w_{1, t}(X_t) w_{i, t}(X) \Tilde{\mu}_{i, t} \Big\|_{\psi_2} \lesssim \sum_{i=1}^K \sgnorm{ w_{1, t}(X_t) w_{i, t}(X) \Tilde{\mu}_{i, t} } \lesssim K B.
\end{equation}
Using properties of the sub-Gaussian random variable from Lemma \ref{def:subgaussian} in a similar way of Eq.~\eqref{eq:first-term-grad-sgnorm}, we have
\begin{align}
    & \Big\| \sum_{i=1}^K \sum_{j=1}^K w_{1, t}(X_t) w_{i, t}(X_t) w_{j, t}(X_t) \mu_{i, t}^\top \mu_{j, t} ( \Tilde{x}_t - \Tilde{\mu}_{1, t}) \Big\|_{\psi_2} \\
    \leq & \; \sum_{i=1}^K \sum_{j=1}^K \Big\| w_{1, t}(X_t) w_{i, t}(X_t) w_{j, t}(X_t) \mu_{i, t}^\top \mu_{j, t} ( \Tilde{x}_t - \Tilde{\mu}_{1, t}) \Big\|_{\psi_2} \\
    \leq & \; \sum_{i=1}^K \sum_{j=1}^K \Big\| w_{1, t}(X_t) w_{i, t}(X_t) w_{j, t}(X_t) \mu_{i, t}^\top \mu_{j, t} z \Big\|_{\psi_2} + \Big\| w_{1, t}(X_t) w_{i, t}(X_t) w_{j, t}(X_t) \mu_{i, t}^\top \mu_{j, t} (\tau - \Tilde{\mu}_{i,t}) \Big\|_{\psi_2} \\
    \leq & \; K^2 B^2 + K^2 B^3 
    \lesssim \; K^2 B^3 \label{eq:third-term-grad-sgnorm}
\end{align}
We know that $\sgnorm{ w_{1, t}(X_t) \mu_{1, t}^\top X_t } \leq \sgnorm{ \sum_{i=1}^d \mu_{1, t}(i) X_t(i) } \lesssim d B^2$ and $\sgnorm{ \Tilde{x}_t - \Tilde{\mu}_{1,t} } \lesssim B$. Using the fact that the product of two sub-Gaussian random variables is a sub-exponential random variable, we have
\begin{align}
    \label{eq:fourth-term-grad-sgnorm}
    \wnorm{w_{1, t}(X_t) \mu_{1, t}^\top X_t (\Tilde{x}_t - \Tilde{\mu}_{1,t}) }_{\psi_1} &\leq \wnorm{ \Tilde{x}_t - \Tilde{\mu}_{1,t} }_{\psi_2} \wnorm{ w_{1, t}(X_t) \mu_{1, t}^\top X_t }_{\psi_2} \lesssim dB^3
\end{align}
The sub-gaussian norm of $w_{1, t}(X_t) \Tilde{x}_t$ term in the gradient is given by
\begin{align}
\label{eq:fifth-term-grad-sgnorm}
    \sgnorm{ w_{1, t}(X_t) \Tilde{x}_t } \leq \sgnorm{ X_t } \lesssim \sgnorm{ Z } + \sgnorm{ \tau } \lesssim B 
\end{align}
Using the property that the product of two sub-Gaussian random variables is a sub-exponential random variable, we obtain
\begin{align}
    & \Big\|  w_{1, t}(X_t) (\Tilde{x}_t - \Tilde{\mu}_{1,t} ) \Big( \sum_{i=1}^K w_{i, t}(X_t)  \mu_{i, t}^\top X_t  \Big) \Big\|_{\psi_1} \\ 
    & \lesssim   \sgnorm{ w_{1, t}(X_t) (\Tilde{x}_t - \Tilde{\mu}_{1,t} ) } \Big\| \Big( \sum_{i=1}^K w_{i, t}(X_t)  \mu_{i, t}^\top X_t  \Big) \Big\|_{\psi_2} \\
    & \lesssim K d B^3 \label{eq:sixth-term-grad-sgnorm}
\end{align}
For any random variable $Y$, we know that $\wnorm{X}_{\psi_1} \leq \sgnorm{X}$. Therefore, combining Eq. \eqref{eq:first-term-grad-sgnorm}, \eqref{eq:second-term-grad-sgnorm}, \eqref{eq:third-term-grad-sgnorm}, \eqref{eq:fourth-term-grad-sgnorm}, \eqref{eq:fifth-term-grad-sgnorm} and \eqref{eq:sixth-term-grad-sgnorm}, we have 
\begin{align}
    \wnorm{ [\nabla_{\mu_{1, t}} s_{\theta_t}(X_t)^\top s_{\theta_t}(X_t) ]_1 - \mbb{E}[\nabla_{\mu_{1, t}} s_{\theta_t}(X_t)^\top s_{\theta_t}(X_t) ]_1 }_{\psi_1} &\lesssim \wnorm{ [\nabla_{\mu_{1, t}} s_{\theta_t}(X_t)^\top s_{\theta_t}(X_t) ]_1 }_{\psi_1} \\ 
    &\lesssim K^2 d B^3 \label{eq:grad-combined-first-term-senorm}
\end{align}
    Now, we shift our focus on obtaining the sub-exponential norm of $\nabla_{\mu_{1,t}} s_{\theta_t}(X_t)^\top Z_t$. Using $\wnorm{ w_{1, t}(X_t)  (\Tilde{x}_t - \Tilde{\mu}_{1, t}) }_{\psi_2} \lesssim B$ and $\wnorm{ \mu_{1, t}^\top Z_t }_{\psi_2} \lesssim d B$, we obtain
    \begin{align}
    \label{eq:seventh-term-grad-sgnorm}
        \wnorm{ w_{1, t}(X_t) (\Tilde{x}_t - \Tilde{\mu}_{1, t}) \mu_{1, t}^\top Z_t }_{\psi_1} \leq \wnorm{ w_{1, t}(X_t) (\Tilde{x}_t - \Tilde{\mu}_{1, t}) }_{\psi_2} \wnorm{ \mu_{1, t}^\top Z_t }_{\psi_2} \lesssim d B^2
    \end{align}
    Using Lemma \ref{def:subgaussian}, we have $\sgnorm{ w_{1, t}(X_t) z_t } \leq \sgnorm{ z_t } \lesssim 1$. For the last term, we have
    \begin{align}
        \Big\| w_{1, t}(X_t) (\Tilde{x}_t - \Tilde{\mu}_{1, t}) \sum_{i=1}^K w_{i, t}(X_t) \mu_{i, t}^\top Z_t \Big\|_{\psi_1} &\leq \sgnorm{ w_{1, t}(X_t) (\Tilde{x}_t - \Tilde{\mu}_{1, t}) } \Big\| \sum_{i=1}^K w_{i, t}(X_t) \mu_{i, t}^\top Z_t \Big\|_{\psi_2} \\
        &\lesssim K d B^2 \label{eq:ninth-term-grad-sgnorm}
    \end{align}
    Combining Eq.~\eqref{eq:seventh-term-grad-sgnorm}, \eqref{eq:ninth-term-grad-sgnorm}, we have
    \begin{align}
    \label{eq:grad-combined-second-term-senorm}
        \Big\| \frac{ [\nabla_{\mu_{1,t}} s_{\theta_t}(X_t)^\top Z_t]_1 }{ \beta_t } - \frac{ \mbb{E}  [\nabla_{\mu_{1,t}} s_{\theta_t}(X_t)^\top Z_t]_1 }{ \beta_t } \Big\|_{\psi_1} \lesssim \Big\| \frac{ [\nabla_{\mu_{1,t}} s_{\theta_t}(X_t)^\top Z_t]_1 }{ \beta_t } \Big\|_{\psi_1} \lesssim \frac{K d B^2}{ \beta_t },
    \end{align}
    where $[\nabla_{\mu_{1,t}} s_{\theta_t}(X_t)^\top Z_t]_1$ denotes the first coordinate of $\nabla_{\mu_{1,t}} s_{\theta_t}(X_t)^\top Z_t$. Combining Eq.~\eqref{eq:grad-combined-first-term-senorm} and Eq.~\eqref{eq:grad-combined-second-term-senorm}, we have
    \begin{align*}
        \Big\| [\nabla_{\mu_{1,t}} L_t(s_{\theta_t}( X_t ))]_1 -  [\nabla_{\mu_{1,t}} L_t(s_{\theta_t} )]_1  \Big\|_{\psi_1} \lesssim \frac{K^2 d B^3}{ \beta_t }
    \end{align*}
    For each i.i.d. sample $x_{i, t}$, the term $[\nabla_{\mu_{1,t}} L_t(s_{\theta_t}( x_{i, t} ))]_1 -  [\nabla_{\mu_{1,t}} L_t(s_{\theta_t} )]_1 $ is also independent and identically distributed. Therefore, using Lemma \ref{lemma:bernstein-inequality}, for any $\epsilon$ that is $\Theta(\frac{1}{\text{poly}(d)})$, we have
    \begin{align*}
        \Pr \Big[ \Big| \frac{1}{n} \sum_{i=1}^n [\nabla_{\mu_{1,t}} L_t(s_{\theta_t}( x_{i, t} ))]_1 -  [\nabla_{\mu_{1,t}} L_t(s_{\theta_t} )]_1  \Big| \geq \epsilon \Big] \leq 2 \exp \Big( -\frac{n \epsilon^2 \beta_t^2 }{ K^4 d^2 B^6 } \Big).
    \end{align*}
    A similar analysis will give the concentration for each coordinate. Using the union bound and rescaling $\epsilon$ as $\frac{\epsilon}{d}$, with probability at least $1 - 2 d \exp \Big( -\frac{n \epsilon^2 \beta_t^2 }{ K^4 d^4 B^6 } \Big)$, we have 
    \begin{align*}
        \norm{ \nabla_{\mu_{1,t}} \Big( \frac{1}{n} \sum_{i=1}^n L_t(s_{\theta_t}( x_{i, t} )) \Big) - \nabla_{\mu_{1,t}}  L_t(s_{\theta_t})  } \leq \epsilon
    \end{align*}
    Note that for any $t = \Omega(1)$, $\beta_t \geq c$ for some constant $c$. Therefore, choosing $n$ provided in the Lemma \ref{lemma:sample-complexity-k-mog} statement, we obtain the result.
\end{proof}

\subsection{Proof of Theorem \ref{thm:k-mog-appendix}}

\begin{proof}[Proof of Theorem \ref{thm:k-mog-appendix}]
    For any training iteration $h$, assume that parameters $\theta_t^{(h)}$ are such that $\norm{ \mu_{i, t}^{(h)} - \mu_{i, t}^{*} } \leq \frac{c_r}{4} \sqrt{\log d} $ we can write the update on the DDPM objective as follows: 
    \begin{align*}
        \wnorm{ \mu_{1,t}^{(h+1)} - \mu_{1, t}^* } = & \; \Big\| \mu_{1,t}^{(h)} - \eta \nabla \Big( \frac{1}{n} \sum_{i=1}^n L_t( s_{\theta_t^{(h)}}(x_{i, 0}, z_{i, t}) ) \Big) - \mu_{1, t}^* \Big\| \\
        \leq & \; \big\| \mu_{1, t}^{(h)} + \eta \,\mbb{E} [ w_{1, t}(X_t) (X_t - \mu_{1, t}^{(h)}) ] - \mu_{1, t}^* \big\| \\
        & \quad\quad+ \eta \Big\| \rb{ - \nabla_{\mu_{1,t}}  L_t( s_{\theta_t} ) } - \mbb{E} [ w_{1, t}(X_t) (X_t - \mu_{1, t}^{(h)}) ] \Big\| \\ 
        & \quad\quad+  \eta \Big\| \rb{ \nabla_{\mu_{1,t}}  L_t( s_{\theta_t} ) }  - \nabla_{\mu_{1,t}} \Big( \frac{1}{n} \sum_{i=1}^n L_t( s_{\theta_t^{(h)}}(x_{i, 0}, z_{i, t}) ) \Big) \Big\|\,.
    \end{align*}
    Using Lemma \ref{lemma:population-GD-gradient-EM}, Lemma \ref{lemma:sample-complexity-k-mog} and Theorem 3.2 from \cite{segol2021improved}, for any $\eta \in (0, K)$, we have
    \begin{align*}
        \wnorm{ \mu_{1,t}^{(h+1)} - \mu_{1, t}^* } \leq & \; \rb{ 1 - \frac{3 \eta}{8K} } \wnorm{ \mu_{1, t}^{(h)} - \mu_{1, t}^* } + \frac{\eta K^2 B^2}{ d^{ \frac{c_r^2}{4000} } } + \eta \epsilon.
    \end{align*}
    Choosing $\eta = \frac{2 K}{3}$, $c_r$ to be sufficiently large constant and $\epsilon$ to be $\Theta(\frac{1}{\text{poly}(d)})$, we have
    \begin{align*}
        \wnorm{ \mu_{1,t}^{(h+1)} - \mu_{1, t}^* } \leq \frac{3}{4} \wnorm{ \mu_{1,t}^{(h)} - \mu_{1, t}^* } + \epsilon
    \end{align*}
    By assumption \ref{asm:mog-k-initialization}, $\wnorm{ \mu_{1,t}^{(0)} - \mu_{1, t}^* } \leq O(\sqrt{\log d})$ and therefore, choosing $H$ to be $\Omega( \log (\frac{ \log d }{ \epsilon } ) )$, we obtain the result. 
\end{proof}

\section{Additional proofs}
\label{sec:additional}

\subsection{Proof of Lemma \ref{lemma:update}}
\label{subsec:update-proof}

\begin{proof}[Proof of Lemma \ref{lemma:update}]
By calculating the negative gradient of the DDPM objective in Eq.~\eqref{eq:diffusion-loss}, we obtain
\begin{equation}
\label{eq:pure-GD-update-2-mog}
\begin{aligned}
    -\nabla_{\mu_t} L_t(s_{\mu_t}) %
    &= - \mbb{E}_{X_0, Z_t} [ ( \tanh(\mu_t^\top X_t) I + \tanh'( \mu_t^\top X_t ) X_t \mu_t^\top ) (  s_{\mu_t}(X_t) + \frac{Z_{t} }{ \beta_t } ) ] \\
    &= - \mbb{E}[ ( \tanh(\mu_t^\top X_t) I + \tanh'( \mu_t^\top X_t ) X_t \mu_t^\top ) (  \tanh( \mu_t^\top X_t ) \mu_t - X_t + \frac{Z_{t}}{ \beta_t } ) ] \\ 
    &= \mbb{E}[ -  \tanh^2(\mu_t^\top X_t) \mu_t - \tanh( \mu_t^\top X_t ) \tanh'( \mu_t^\top X_t ) X_t \norm{ \mu_t }^2 + \tanh (\mu_t^\top X_t) X_t \\
    & \quad + \tanh'( \mu_t^\top X_t ) \mu_t^\top X_t X_t - \tanh(\mu_t^\top X_t) \frac{Z_t}{ \beta_t } - \tanh'( \mu_t^\top X_t ) X_t \mu_t^\top \frac{Z_t}{ \beta_t } ] \\
\end{aligned}
\end{equation}
By simplifying the gradient terms involving $Z_t$ by the Stein's identity as in Lemma \ref{lemma:gradient-simplifying-zterms} and plugging it back in the gradient, we obtain
\begin{align*}
    -\nabla_{\mu_t} L_t(s_{\mu_t}) &= \mbb{E} \Big[ \rb{ \tanh (\mu_t^\top X_t) - \tanh( \mu_t^\top X_t ) \tanh'( \mu_t^\top X_t ) \norm{\mu_t}^2 + \tanh'( \mu_t^\top X_t ) \mu_t^\top X_t } X_t \Big]  \\
    & \quad - \mu_t - \mbb{E} \sbr{  \tanh''( \mu_t^\top X_t ) \norm{ \mu_t }^2 X_t  } - \mbb{E}\sbr{\tanh'( \mu_t^\top X_t ) \mu_t } \\
    & = \mbb{E} \Big[ \rb{ \tanh (\mu_t^\top X_t)  - 0.5 \tanh''( \mu_t^\top X_t ) \norm{ \mu_t }^2 + \tanh'( \mu_t^\top X_t ) \mu_t^\top X_t } X_t \Big] \\ 
    & \quad - \mu_t  - \mbb{E}\sbr{\tanh'( \mu_t^\top X_t ) \mu_t } 
\end{align*}
Observe that $\rb{ \tanh (\mu^\top x) - \frac{1}{2}  \tanh''( \mu^\top x ) \norm{ \mu }^2 + \tanh'( \mu^\top x ) \mu^\top x } x$ and $\tanh'( \mu^\top x )$ are even functions and $X_t$ is a symmetric distribution, therefore, for any even function $f$, we can write $\mbb{E}_{X_t}[ f( X_t ) ] = \frac{1}{2} \mbb{E}_{X_t \sim \mc{N}(\mu_t^*, \Id)}[ f( X_t ) ] + \frac{1}{2} \mbb{E}_{X_t \sim \mc{N}(-\mu_t^*, I)}[ f( X_t ) ] = \mbb{E}_{X_t \sim \mc{N}(\mu_t^*, \Id)}[ f( X_t ) ]$. Applying this property of the even function on the gradient update, we obtain the result. 
\end{proof}

\begin{lemma}
\label{lemma:gradient-simplifying-zterms}
     When random variable $X_t = \alpha_t X_0 + \beta_t Z_t$ where $Z_t \sim \mc{N}(0, I), \alpha_t = \exp(-t)$ and $\beta_t = \sqrt{1 - \exp(-2t)}$, then for any $t>0$, the following two equations hold. 
     \begin{align*}
        \mbb{E}_{X_0, Z_t} & \Big[  \tanh( \mu_t^\top X_t ) \frac{Z_t}{ \beta_t } + \tanh^2 ( \mu_t^\top X_t ) \mu_t \Big] = \mu_t \\
        \mbb{E}_{X_0, Z_t} & \Big[ \tanh'( \mu_t^\top X_t ) \frac{ \mu_t^\top Z_t }{ \beta_t } X_t \Big] = \mbb{E}_{X_0, Z_t} \sbr{ \tanh''( \mu_t^\top X_t ) \norm{ \mu_t }^2 X_t  + \tanh'( \mu_t^\top X_t ) \mu_t }
     \end{align*}
\end{lemma}
\begin{proof}
Applying Stein's lemma on the first term, we get the first equation of the statement in the Lemma. 
\begin{align*}
    \mbb{E}_{X_0, Z_t} \sbr{ \tanh( \mu_t^\top X_t ) \frac{Z_t}{ \beta_t } } &= \mbb{E}_{X_0, Z_t} \sbr{ \tanh( \mu_t^\top ( \alpha_t X_0 + \beta_t Z_t ) ) \frac{Z_t}{ \beta_t } } = \mbb{E}_{X_0, Z_t} \sbr{ \tanh' ( \mu_t^\top X_t ) \mu_t } \\ 
    &= \mbb{E}_{X_0, Z_t} \sbr{ \rb{ 1 - \tanh^2 ( \mu_t^\top X_t ) } \mu_t }
\end{align*}
For the second term, we have
\begin{align*}
    \mbb{E} & \Big[ \tanh'( \mu_t^\top X_t ) \frac{ \mu_t^\top Z_t }{ \beta_t } X_t \Big]  = \mbb{E} \Big[ \tanh'( \mu_t^\top X_t ) \frac{ \mu_t^\top Z_t }{ \beta_t } \alpha_t X_0 \Big] + \mbb{E}\sbr{\tanh'( \mu_t^\top X_t ) \mu_t^\top Z_t Z_t } \\
    &= \sum_{i=1}^d \mbb{E} \Big[ \alpha_t X_0 \tanh'( \mu_t^\top X_t ) \frac{ \mu_t(i) Z_t(i) }{ \beta_t } \Big] + \mbb{E}\sbr{\tanh'( \mu_t^\top X_t ) \mu_t } + \mbb{E}\sbr{ \tanh''( \mu_t^\top X_t ) \mu_t^\top Z_t \beta_t \mu_t } \\
    &= \sum_{i=1}^d \mbb{E} \Big[ \alpha_t X_0 \tanh''( \mu_t^\top X_t ) \mu_t(i) \mu_t(i) \Big] + \mbb{E}\sbr{\tanh'( \mu_t^\top X_t ) \mu_t } + \mbb{E}\sbr{ \tanh''( \mu_t^\top X_t ) \mu_t^\top Z_t \beta_t \mu_t }
\end{align*}
where the second equality follows from the Stein's lemma on the $\mbb{E}[\tanh'( \mu_t^\top X_t ) \mu_t^\top Z_t Z_t ]$ and the last equality follows from the Stein's lemma on $\mbb{E} [ \alpha_t X_0 \tanh''( \mu_t^\top X_t ) \mu_t(i) Z_t(i) ]$. Applying Stein's inequality on the $\mbb{E}\sbr{ \tanh''( \mu_t^\top X_t ) \mu_t^\top Z_t \beta_t \mu_t }$, we obtain
\begin{align*}
    &= \mbb{E} \sbr{ \alpha_t X_0 \tanh''( \mu_t^\top X_t ) \norm{ \mu_t }^2  } + \mbb{E}\sbr{\tanh'( \mu_t^\top X_t ) \mu_t } + \sum_{i=1}^d \beta_t \mu_t  \mbb{E}\sbr{ \tanh'''( \mu_t^\top X_t ) \mu_t(i) \beta_t \mu_t(i) } \\
    &= \mbb{E} \sbr{ X_t \tanh''( \mu_t^\top X_t ) \norm{ \mu_t }^2  } - \mbb{E} \sbr{ \beta_t Z_t \tanh''( \mu_t^\top X_t ) \norm{ \mu_t }^2  } + \mbb{E}\sbr{\tanh'( \mu_t^\top X_t ) \mu_t } \\ & \quad + \beta_t^2 \norm{ \mu_t }^2 \mu_t  \mbb{E}\sbr{ \tanh'''( \mu_t^\top X_t )  } \\
    &= \mbb{E} \sbr{ X_t \tanh''( \mu_t^\top X_t ) \norm{ \mu_t }^2  } + \mbb{E}\sbr{\tanh'( \mu_t^\top X_t ) \mu_t }.
\end{align*}
\end{proof}

\subsection{Proof of Lemma \ref{lemma:G-contraction}}
\label{section:proof-G-contraction-lemma}

\begin{proof}[Proof of Lemma \ref{lemma:G-contraction}]
    Recall that the gradient update for any $\mu_t^*$ is given by
    \begin{align}
    \label{eq:reduced-mean-value-thm}
        - \nabla_{\mu^*_t} L_t(s_{\mu^*_t}) 
        &= G(\mu_t^*, \mu_t^*) + \eta \mbb{E}_{x \sim \mc{N}(\mu_t^*, \Id)}[ \tanh (\mu^{*\top}_t x) x ] - \eta \mu_t^*
    \end{align}
    We know that $\mbb{E}_{x \sim \mc{N}(\mu_t^*, \Id)}[ \tanh (\mu^{*\top}_t x) x ] = \mu_t^*$ (Eq.(2.1) of \cite{daskalakis2017ten}) and $\nabla_{\mu^*_t} L_t(s_{\mu^*_t})  = 0$ because $\mu^*_t$ is a stationary point of the regression objective of diffusion model. This implies that $G(\mu^*_t, \mu^*_t) = 0$ for any $\mu^*_t$.
    
    Note that this proof only talks about 1D case therefore, for the purpose of this proof, we use $a$ to denote $\mu$ and $b$ to denote $\mu^*$. In 1D, using Mean value theorem, we have
    \begin{align}
    \label{eq:mean-value-thm}
        \frac{ G( a, b) - G( a, a ) }{ b - a } &= \frac{d G(a, \xi)}{ d \xi } \text{ for some $\xi \in [a, b]$ (if $a < b$)} 
    \end{align} 
    Using the fact that $G(a, a) = 0$ in Eq.~\eqref{eq:mean-value-thm}, we have
    \begin{align*}
        \abs{ G(a, b) } &= \abs{\frac{d G(a, \xi)}{ d \xi }} \abs{ b - a }
    \end{align*}
    Observe that it suffices to prove $\abs{\frac{d G(a, \xi)}{ d \xi }} \leq 0.01$ to obtain the lemma. By computing the gradient of $G$, we obtain
    \begin{align*}
        \frac{ d G(a, \xi) }{ d \xi } =  \eta \mbb{E}_{x \sim \mc{N}(\xi, 1)} \Big[ 2 \tanh'(a x) a x + \tanh''(a x) \rb{ \frac{-3 a^2}{2} + a^2 x^2 } - \frac{1}{2} a^3 x \tanh'''( a x ) \Big]
    \end{align*}
    For the first term, we have
    \begin{align*}
        \mbb{E}_{x \sim \mc{N}(\xi, I)}[ \tanh'(a x) a x ] &= \frac{1}{\sqrt{2 \pi}} \int_{-\infty}^\infty \tanh'(a x) a x e^{-\frac{ (x - \xi)^2 }{2}} dx \\
        &= \frac{1}{\sqrt{2 \pi}} \int_0^\infty \tanh'(a x) a x \rb{ e^{-\frac{ (x - \xi)^2 }{2}} - e^{-\frac{ (x + \xi)^2 }{2}} } dx \\
        &\leq \frac{1}{\sqrt{2 \pi}} \int_0^\infty e^{-a x} a x e^{-\frac{ (x - \xi)^2 }{2}} dx \\
        &\leq \frac{a e^{\frac{a^2 - 2 a \xi}{2}} }{\sqrt{2 \pi}}  \int_0^\infty x e^{-\frac{(x - \xi + a)^2 }{2}} dx \\
        &\leq a e^{\frac{a^2 - 2 a \xi}{2}} ( \sqrt{\frac{2}{\pi} }e^{ - \frac{ (\xi - a)^2 }{2}} + (\xi - a) \text{erf} \rb{ \frac{ \xi - a }{ \sqrt{2} } } ) \\
        &\leq a e^{-\frac{\xi^2}{2}} + a \abs{\xi - a} e^{\frac{-2a(\xi - a) - a^2 }{2} } 
    \end{align*}
    Using Lemma 1 of \cite{daskalakis2017ten}, we know that $\mbb{E}_{x \sim \mc{N}(\xi, I)}[ \tanh'(a x) a x ] > 0$. Therefore, we have
    \begin{align}
    \label{eq:tanhp-bound}
        \abs{ \mbb{E}_{x \sim \mc{N}(\xi, I)}[ \tanh'(a x) a x ] } \leq a e^{-\frac{\xi^2}{2}} + a\abs{\xi - a} e^{\frac{-2a(\xi - a) - a^2 }{2} }
    \end{align}

For the second term, we have
\begin{align*}
    &\mbb{E}_{x \sim \mc{N}(\xi, 1)}[ \tanh''(a x) ( -\frac{3 a^2}{2} + a^2 x^2 ) ]  \\ 
    &= \frac{1}{\sqrt{2\pi}} \int_{0}^\infty a^2 \tanh''( a x ) ( -\frac{3}{2} +  x^2 ) \rb{ \exp( - \frac{ (x - \xi)^2 }{2} ) - \exp( - \frac{ (x + \xi)^2 }{2} ) } dx \\
    &\leq \frac{1}{\sqrt{2\pi}} \int_{0}^{\sqrt{\frac{3}{2}}}  a^2 e^{-2 a x} ( \frac{3}{2} - x^2 ) \exp( - \frac{ (x - \xi)^2 }{2} )  dx   \\
    &\leq \frac{3}{\sqrt{2\pi}} a^2 \exp( -\frac{a^2}{16} )
\end{align*}
Assuming $a \geq \sqrt{6}$, then when $\xi \geq a \geq \sqrt{6}$, we have $\exp( - \frac{ (x - \xi)^2 }{2} ) \leq \exp( - \frac{a^2}{4} )$ and when $\xi \leq a$, using $\xi \geq \frac{3 a}{4}$, we have $\exp( - \frac{ (x - \xi)^2 }{2} ) \leq \exp( - \frac{a^2}{16} )$. For the lower bound, we have
\begin{align*}
    &\mbb{E}_{x \sim \mc{N}(\xi, 1)}[ \tanh''(a x) ( -\frac{3 a^2}{2} + a^2 x^2 ) ] \\ 
    &= \frac{1}{\sqrt{2 \pi}} \int_{0}^\infty \tanh''( a x ) ( -\frac{3 a^2}{2} + a^2 x^2 ) \rb{ \exp( - \frac{ (x - \xi)^2 }{2} ) - \exp( - \frac{ (x + \xi)^2 }{2} ) } dx \\
    &\geq \frac{1}{\sqrt{2 \pi}}  \int_{ \sqrt{ \frac{3}{2} } }^\infty \tanh''( a x ) ( -\frac{3 a^2}{2} + a^2 x^2 ) \rb{ \exp( - \frac{ (x - \xi)^2 }{2} ) - \exp( - \frac{ (x + \xi)^2 }{2} ) } dx \\
    &\geq \frac{1}{\sqrt{2 \pi}}  \int_{ \sqrt{ \frac{3}{2} } }^\infty \tanh''( a x ) a^2 x^2 \rb{ \exp( - \frac{ (x - \xi)^2 }{2} ) - \exp( - \frac{ (x + \xi)^2 }{2} ) } dx \\
    &\geq - \frac{8 a^2}{\sqrt{2 \pi}}  \int_{ \sqrt{ \frac{3}{2} } }^\infty e^{-2 a x }  x^2 \rb{ \exp( - \frac{ (x - \xi)^2 }{2} ) - \exp( - \frac{ (x + \xi)^2 }{2} ) } dx \\
    &\geq - \frac{ 8 a^2 e^{- \sqrt{6} a } }{ \sqrt{2\pi} } \int_{ \sqrt{ \frac{3}{2} } }^\infty x^2 \exp( - \frac{ (x - \xi)^2 }{2} ) dx \geq - 8 a^2 e^{- \sqrt{6} a }
\end{align*} Using upper bound and lower bound, we have
\begin{align*}
    \abs{ \mbb{E}_{x \sim \mc{N}(\xi, 1)}[ \tanh''(a x) a^2 ( -\frac{3}{2} + x^2 ) ] } \leq  8 a^2 e^{- \sqrt{6} a }
\end{align*}
For the third term, we have 
\begin{align*}
    &\Big| \mbb{E}_{x \sim \mc{N}(\xi, 1)} [ \frac{a^3 x}{2} \tanh'''( a x ) ] \Big| \\ 
    &= \bigg| \frac{1}{32 \sqrt{2 \pi} } \int_0^\infty a^3 x \sigma(2 a x)(1 - \sigma(2 a x)) \rb{ 1 - 6 \sigma(2a x)(1 -  \sigma(2 a x) ) } \bigg( \exp \bigg( -\frac{ (x - \xi)^2 }{2} \bigg) \\
    & \quad\quad - \exp \bigg( -\frac{ (x + \xi)^2 }{2} \bigg) \bigg) dx \bigg| \\
    &\leq \bigg| \frac{3a^3}{16 \sqrt{2 \pi} } \int_0^\infty x \sigma^2(2 a x)(1 - \sigma(2 a x))^2 
    \bigg( \exp \bigg( -\frac{ (x - \xi)^2 }{2} \bigg) - \exp \bigg( -\frac{ (x + \xi)^2 }{2} \bigg) \bigg) dx \bigg| \\
    &\leq \frac{3a^3}{16 \sqrt{2 \pi} } \int_0^\infty x e^{-a x} \exp \bigg( -\frac{ (x - \xi)^2 }{2} \bigg) dx \\
    &\leq \frac{a^3}{10} e^{-\frac{\xi^2}{2}} + \frac{a^3}{10} \abs{\xi - a} e^{\frac{-2a(\xi - a) - a^2 }{2}}\,.
\end{align*}
We can lower bound the third term as follows:
\begin{align*}
    & \mbb{E}_{x \sim \mc{N}(\xi, 1)} [ \frac{a^3 x}{2} \tanh'''( a x ) ] \\ 
    &\geq  \frac{1}{2 \sqrt{2 \pi} } \int_0^c a^3 x \tanh'''(ax)  \bigg( \exp \bigg( -\frac{ (x + \xi)^2 }{2} \bigg) - \exp \bigg( -\frac{ (x - \xi)^2 }{2} \bigg) \bigg) dx \\
    &\geq  \frac{a^3}{2 \sqrt{2 \pi} } \int_0^c x \exp \bigg( -\frac{ (x - \xi)^2 }{2} \bigg) \rb{\exp \rb{ -2 \xi x } - 1 } dx  \\
    &\geq  -\frac{a^3 \xi}{ \sqrt{2 \pi} } \int_0^c x^2  \exp \bigg( -\frac{ (x - \xi)^2 }{2} \bigg) dx  \geq -\frac{\xi \exp( - \frac{\xi^2}{4}) }{ \sqrt{2 \pi} }
\end{align*}
Using all the bounds, we have
\begin{align*}
    \abs{\frac{ d G(a, \xi) }{ d \xi }} &\leq \frac{a^3}{10} e^{-\frac{\xi^2}{2}} + \frac{a^3}{10} \abs{\xi - a} e^{\frac{-2a(\xi - a) - a^2 }{2}} + 8 a^2 e^{- \sqrt{6} a } + a e^{-\frac{\xi^2}{2}} + a \abs{ \xi - a} e^{\frac{-2a(\xi - a) - a^2 }{2}}
\end{align*}
When $\xi \geq a$ and $a \geq c$ for some sufficiently large constant $c$ (for example, $c=25$), then, we have 
\begin{align*}
    \abs{\frac{ d G(a, \xi) }{ d \xi }} &\leq \frac{a^3}{10} e^{-\frac{a^2}{2}} + \frac{a^3}{10} \abs{\xi - a} e^{\frac{ - a^2 }{2}} + 8 a^2 e^{- \sqrt{6} a } + a e^{-\frac{a^2}{2}} + a \abs{\xi - a} e^{\frac{ - a^2 }{2}} \leq 0.01
\end{align*}
When $\frac{3 a}{4} \leq \xi \leq a$ and  $a > c$ for sufficiently large constant $c$ (for example, $c=25$), we have 
\begin{align*}
    \abs{\frac{ d G(a, \xi) }{ d \xi }} &\leq \frac{a^3}{10} e^{-\frac{9 a^2}{32}} + \frac{a^4}{40} e^{\frac{- a^2 }{4}} + 8 a^2 e^{- \sqrt{6} a } + a e^{-\frac{a^2}{2}} + \frac{a^2}{4} e^{\frac{- a^2 }{4}} \leq 0.01
\end{align*}
Pluggint the bound on $|\frac{ d G(a, \xi) }{ d \xi }|$ in Eq.~\eqref{eq:reduced-mean-value-thm}, we obtain the final result. 
\end{proof}

\subsection{Proof of Lemma \ref{lemma:mu-in-interval}}
\label{subsec:mu-in-interval-proof}

\begin{proof}[Proof of Lemma \ref{lemma:mu-in-interval}]

We will prove this by induction. For $h=0$, this is true because the algorithm initializes the gradient descent on the low noise regime with the output of gradient descent on the high noise regime, and the output is guaranteed to have $\dtp{\hat{\mu}_t^{(0)}}{\hat{\mu}_t^*}$ to be $\Omega(1)$ and by assumption $\wnorm{\mu_t^*} > c'$, therefore $\wnorm{\mu_t^{(0)}} \in [c, \frac{4 \dtp{ \hat{\mu}_t^{(0)} }{ \mu_t^* } }{3}]$. 

Suppose $\wnorm{\mu_t^{(h)}} \in [c, \frac{4 \dtp{ \hat{\mu}_t^{(h)} }{ \mu_t^* } }{3}]$, then we know that $\wnorm{ \mu_t^{(h+1)} - \mu_t^* } < \wnorm{ \mu_t^{(h)} - \mu_t^* }$. To prove $\wnorm{\mu_t^{(h+1)}} \in [c, \frac{4 \dtp{ \hat{\mu}_t^{(h+1)} }{ \mu_t^* } }{3}]$, first we will prove that $\dtp{ \hat{\mu}_t^{(h)} }{ \mu_t^{(r+1)} } \in [c, \frac{6 \dtp{ \hat{\mu}_t^{(h)} }{ \mu_t^* } }{5} ]$. Note that the update in the direction of $\dtp{\hat{\mu}_t}{\mu_t}$ works like 1D. Therefore, we have a contraction for it as follows.
    \begin{align*}
        & \abs{ \dtp{ \hat{\mu}_t^{(h)} }{ \mu_t^{(h+1)} } - \dtp{ \hat{\mu}_t^{(h)} }{ \mu_t^* } } < \abs{\dtp{ \hat{\mu}_t^{(h)} }{ \mu_t^{(h)} } - \dtp{ \hat{\mu}_t }{ \mu_t^* }}
    \end{align*}
    If $\wnorm{\mu_t^{(h)}} \leq \dtp{ \hat{\mu}_t^{(h)} }{ \mu_t^* }$, then using Lemma \ref{lemma:update-bound}, we know $ \dtp{\hat{\mu}_t^{(h)}}{ \mu_t^{(h+1)} } \leq \frac{6 \dtp{ \hat{\mu}_t^{(h)} }{ \mu_t^* } }{5}$ and $\dtp{\hat{\mu}_t^{(h)}}{ \mu_t^{(h+1)} } \geq \wnorm{\mu_t^{(h)}} \geq c$ because of the contraction. If $\wnorm{ \mu_t^{(h)} } \geq \dtp{ \hat{\mu}_t^{(h)} }{ \mu_t^* }$ and $\dtp{ \hat{\mu}_t^{(h)} }{ \mu_t^{(h+1)} } \geq \dtp{ \hat{\mu}_t^{(h)} }{ \mu_t^* } $, then $\dtp{ \hat{\mu}_t^{(h)} }{ \mu_t^{(h+1)} } \leq \|\mu_t^{(h)} \|$ because of the contraction. If $\wnorm{ \mu_t^{(h)} } \geq \dtp{ \hat{\mu}_t^{(h)} }{ \mu_t^* }$ and $\dtp{ \hat{\mu}_t^{(h+1)} }{ \mu_t^{(h)} } \leq \dtp{ \hat{\mu}_t^{(h)} }{ \mu_t^* }$, then using $\dtp{ \hat{\mu}_t^{(h+1)} }{ \mu_t^{(h)} } \geq \wnorm{ \mu_t^{(h)} } - \abs{U(\dtp{ \hat{\mu}_t^{(h)} }{ \mu_t^{(h)} }, \dtp{ \hat{\mu}_t^{(h)} }{ \mu_t^* })} \geq \frac{4 \dtp{ \hat{\mu}_t^{(h)} }{ \mu_t^* } }{5} \geq \frac{4 \dtp{ \hat{\mu}_t^{(0)} }{ \mu_t^* } }{5} \geq c$ from Lemma \ref{lemma:angle-decrease}, we get the result that $\dtp{ \hat{\mu}_t^{(h)} }{ \mu_t^{(h+1)} } \in [c, \frac{6 \dtp{ \hat{\mu}_t^{(h)} }{ \mu_t^* } }{5} ]$. Now, using Lemma \ref{lemma:angle-decrease}, we get
    \begin{align*}
        \dtp{ \hat{\mu}_t^{(h)} }{ \mu_t^{(h+1)} } \in  [c, \frac{6 \dtp{ \hat{\mu}_t^{(h)} }{ \mu_t^* } }{5} ] & \implies \wnorm{ \mu_t^{(h+1)} } \in \big[ \frac{c}{\cos \alpha_h }, \frac{ 6 \wnorm{ \mu_t^* } \cos \beta_h }{5 \cos \alpha_h } \big] \\ 
        &  \implies \wnorm{ \mu_t^{(h+1)} } \in \big[ c, \frac{ 4 \wnorm{ \mu_t^* } \cos \beta_{h+1} }{3} \big] \\
        & \implies \wnorm{ \mu_t^{(h+1)} } \in \big[ c, \frac{ 4 \dtp{ \hat{\mu}_t^{(h+1)} }{ \mu_t^* }  }{3} \big]\qedhere
    \end{align*}
\end{proof}

    \begin{lemma}
    \label{lemma:angle-decrease}
    Suppose the angle between $\mu^{(r)}$ and $\mu^*$ is $\beta_r$ and $\alpha_r$ is the angle between $\mu^{(r)}$ and $\mu^{(r+1)}$ and assume the contraction is true at time $r$. Assume that $\beta_0 \in (0, \frac{\pi}{2})$. Then:
    \begin{equation}
        \alpha_r \in (0, \pi/2) \;\; \forall r \qquad \text{and} \qquad
        \cos \beta_r  \leq \cos \beta_{r+1}
    \end{equation}
    which implies that 
    \begin{align*}
        \cos \beta_r \leq \cos \beta_{r+1} \;\; \forall r \implies \dtp{\hat{ \mu }^{(r)}}{ \mu^* } \geq \dtp{\hat{ \mu }^{(0)}}{ \mu^* }
    \end{align*}
\end{lemma}

\begin{proof}
    First, we will prove that if $\beta_r \in (0, \frac{\pi}{2})$ and $\wnorm{\mu^{(r)}} \in [c, \frac{4 \dtp{ \hat{\mu}_t^{(r)} }{ \mu_t^* } }{3}]$, then $\alpha_r \in (0, \beta_r)$ for any $r$. We denote $\alpha_r > 0$ if $\mu^{(r)}$ moves towards $\mu^{(r)\perp}$ and hence towards $\mu^*$. The following simple observation of $\dtp{ \hat{\mu}^{(r)\perp} }{\mu^{(r+1)}} \geq 0$ proves that $\alpha_r > 0$.
    \begin{align*}
        \MoveEqLeft \dtp{ \hat{\mu}^{(r)\perp} }{\mu^{(r+1)}} \\
        &= \mbb{E}_{x \sim \mc{N}(\mu^*, 1)} \Big[  \eta  \big( \tanh (\mu^{(r)\top} x) - \frac{1}{2} \tanh''( \mu^{(r)\top} x ) \wnorm{ \mu^{(r)} }^2 + \tanh'( \mu^{(r)\top} x ) \mu^{(r)\top} x \big)\cdot \dtp{ \hat{\mu}^{(r)\perp} }{ x } \Big] \\
        &= \mbb{E}_{x \sim \mc{N}(0, 1)} \Big[  \eta  \Big( \tanh (\mu^{(r)\top} (x + \mu^*)) - \frac{1}{2} \tanh''( \mu^{(r)\top} (x + \mu^*) ) \wnorm{ \mu^{(r)} }^2 \\ 
        & \quad\quad\quad + \tanh'( \mu^{(r)\top} (x + \mu^*) ) \mu^{(r)\top} (x + \mu^*) \Big) \cdot \dtp{ \hat{\mu}^{(r)\perp} }{ (x + \mu^*) } \Big] \\
        &= \mbb{E}_{\alpha_1, \alpha_2 \sim  \mc{N}( \dtp{ \hat{\mu}^{(r)} }{\mu^*} , 1)}  \Big[  \eta  \Big( \tanh ( \wnorm{ \mu^{(r)} } \alpha_1) - \frac{1}{2} \tanh''(\hspace{0.005em} \wnorm{ \mu^{(r)} } \alpha_1  ) \wnorm{ \mu^{(r)} }^2 \\
        &\quad\quad\quad +  \tanh'(\hspace{0.005em} \wnorm{ \mu^{(r)} } \alpha_1 ) \wnorm{ \mu^{(r)} } \alpha_1 \Big) (\alpha_2 + \dtp{ \hat{\mu}^{(r)\perp} }{ \mu^* } ) \Big] \\
        &= \mbb{E}_{\alpha_1, \alpha_2 \sim  \mc{N}( \dtp{ \hat{\mu}^{(r)} }{\mu^*} , 1)}  \Big[  \eta  \Big( \tanh ( \wnorm{ \mu^{(r)} } \alpha_1) - \frac{1}{2} \tanh''(\hspace{0.005em} \wnorm{ \mu^{(r)} } \alpha_1  ) \wnorm{ \mu^{(r)} }^2 \\
        &\qquad\qquad\qquad\qquad\qquad\qquad+ \tanh'(\hspace{0.005em} \wnorm{ \mu^{(r)} } \alpha_1 ) \wnorm{ \mu^{(r)} } \alpha_1 \Big) \cdot \dtp{ \hat{\mu}^{(r)\perp} }{ \mu^* } \Big] > 0\,,
    \end{align*}
    where in the last step we used the fact that $\dtp{ \hat{\mu}^{(r)} }{\mu^*} > 0$ and $\dtp{ \hat{\mu}^{(r)\perp} }{ \mu^* } > 0$.
    
    Now, we will prove that $\cot \alpha_r > \cot \beta_r$ which will prove that $\alpha_r \in (0, \beta_r)$. Note that
    \begin{align*}
        \cot \alpha_r &= \frac{ \dtp{ \hat{\mu}^{(r)} }{\mu^{(r+1)}} }{ \dtp{ \hat{\mu}^{(r)\perp} }{\mu^{(r+1)}} } \hspace{1cm} \text{where} \\
        \dtp{ \hat{\mu}^{(r)} }{\mu^{(r+1)}} = & \; (1-\eta) \wnorm{ \mu^{(r)} } + \eta \mbb{E}_{ \alpha_1 \sim \mc{N}(\hat{\mu }^{(r)\top} \mu^*, 1) }[ \tanh(\hspace{0.005em} \wnorm{ \mu^{(r)} } \alpha_1 ) \alpha_1 ] \\
        &+ \eta \mbb{E}_{\alpha_1 \sim \mc{N}(\hat{\mu }^{(r)\top} \mu^*, 1)}[ -\frac{1}{2} \tanh''(\wnorm{\mu^{(r)}} \alpha_1) \wnorm{\mu^{(r)}}^2 \alpha_1 + \tanh'(\hspace{0.005em} \wnorm{\mu^{(r)}} \alpha_1 ) \wnorm{\mu^{(r)}} \alpha_1^2 \\ 
        &- \tanh'(\hspace{0.005em} \wnorm{\mu^{(r)}} \alpha_1 ) \wnorm{ \mu^{(r)} } ] \\
        \dtp{ \hat{\mu}^{(r)\perp} }{\mu^{(r+1)}} =  & \; \eta \dtp{ \hat{\mu}^{(r)\perp} }{\mu^*} \mbb{E}_{\alpha_1 \sim \mc{N}( \hat{\mu}^{(r)\top} \mu^* , 1)} [ \tanh(\hspace{0.005em} \wnorm{\mu^{(r)}} \alpha_1 ) - \frac{1}{2} \tanh''( \; \wnorm{ \mu^{(r)} }  \alpha_1 ) \wnorm{ \mu^{(r)} }^2 \\ 
        & + \tanh'( \; \wnorm{ \mu^{(r)} } \alpha_1  ) \wnorm{ \mu^{(r)} } \alpha_1 ] \\
        \text{ and } \;\; & \cot \beta_r = \frac{ \dtp{ \hat{\mu}^{(r)} }{ \mu^* } }{ \dtp{ \hat{\mu}^{(r)\perp} }{ \mu^* } } 
    \end{align*}
    Observe the fact that to prove $\frac{a+c'}{b+c} - \frac{a}{b} > 0$, it is sufficient to prove $c' > \frac{ac}{b}$ for $b, c > 0$. Using this observation, to prove $\cot \alpha_r > \cot \beta_r$, it is sufficient to prove 
    \begin{align*}
        &\Big(1 - \eta - \eta \mbb{E} [ \tanh'(\wnorm{\mu^{(r)}} x)] \Big) \wnorm{ \mu^{(r)} } + \eta \mbb{E}_{x} \Big[ -\frac{1}{2} \tanh''(\hspace{0.005em} \wnorm{ \mu^{(r)} } x ) \wnorm{\mu^{(r)}}^2 (x - \dtp{ \hat{\mu}^{(r)} }{ \mu^* }) \\ 
        &+ \tanh'(\hspace{0.005em} \wnorm{\mu^{(r)}} x )(x^2 - \dtp{\hat{\mu}^{(r)} }{ \mu^*} x) + \tanh(\hspace{0.005em} \wnorm{\mu^{(r)}} x )(x - \dtp{\hat{\mu}^{(r)}}{ \mu^* } ) \Big] > 0, 
    \end{align*}
    where the expectation is wrt $\mc{N}( \dtp{\mu^{(r)}}{ \mu^* }, 1 )$. Lemma \ref{lemma:positive-bound} shows that this is indeed true. 
\end{proof}

\begin{lemma}
    \label{lemma:positive-bound} For any $\eta = \frac{1}{20}$, assuming $a \in [30, \frac{4b}{3}]$, we have
    \begin{align*}
        \MoveEqLeft(1 - \eta - \eta \mbb{E}_{x \sim \mc{N}(b, 1)}[ \tanh'(ax) ] ) a \\
        &+ \eta\, \mbb{E}_{x \sim \mc{N}(b, 1)} \Big[ -\frac{1}{2} \tanh''( a x ) a^2 (x - b) \tanh'(ax)(x^2 - bx) + \tanh(ax) (x - b) \Big] > 0\,.
    \end{align*}
\end{lemma}
\begin{proof}
    First, we will find the upper bound on $\mbb{E}[\tanh''(ax)(x - b)]$. 
    \begin{align*}
        \mbb{E}[ \tanh''(ax) (x - b) ] &= \int_{-\infty}^\infty \tanh''(ax) (x - b) \exp \Big(-\frac{(x - b)^2}{2} \Big) dx \\
        &\leq \int_0^b \tanh''(ax)(x - b) \exp \Big(-\frac{(x - b)^2}{2} \Big) dx \\
        &\leq \int_0^b \tanh''(ax) x \exp \Big(-\frac{(x - b)^2}{2} \Big) dx \\
        &\leq \int_0^b \exp(-ax) x \exp \Big(-\frac{(x - b)^2}{2} \Big) dx \\ 
        &\leq \exp\Big(\frac{a^2 - 2ab}{2}\Big) \int_0^b x \exp \Big(-\frac{(x - b)^2 + 2a(x-b) + a^2 }{2} \Big) dx \\
        &\leq \exp(\frac{a^2 - 2ab}{2}) \int_0^\infty x \Big[ \exp \Big(-\frac{(x - b + a)^2 }{2} \Big) + \exp \Big(-\frac{(x + b - a)^2 }{2} \Big) \Big] dx \\
        &\leq \exp(-b^2/2) + \abs{a - b} \cdot \exp\Big(\frac{a^2 - 2ab}{2}\Big)\,.
    \end{align*}
    Now, for the second term, we have
    \begin{align*}
        \MoveEqLeft \mbb{E}_{x \sim \mc{N}(b, 1)}[\tanh'(ax)(x^2 - bx)] \\
        &= \int_{-\infty}^\infty \tanh'(ax) x (x - b) \exp \Big( -\frac{(x-b)^2}{2} \Big) dx \\
        &\geq - b \int_0^b x e^{-ax} \exp \Big( -\frac{(x-b)^2}{2} \Big) dx \\
        &\geq - b \exp\Big( \frac{ a^2 - 2ab }{2}\Big) \int_0^\infty x \Big[ \exp \Big(-\frac{(x - b + a)^2 }{2} \Big) + \exp \Big(-\frac{(x + b - a)^2 }{2} \Big) \Big] dx \\
        &\geq -b \exp(-b^2/2) - b \abs{a-b}\cdot \exp\Big( \frac{ a^2 - 2ab }{2}\Big)
    \end{align*}
    We can rewrite the last term as $\mbb{E}_{x \sim \mc{N}(0, 1)}[ \tanh(a(x + b))x ]$. Using the fact that $\tanh(a(x+b)) > \tanh(a(-x + b))$, we get that $\mbb{E}_{x \sim \mc{N}(0, 1)}[ \tanh(a(x + b))x ] > 0$. Finally, using the upper bound on $\mbb{E}[\tanh'(ax)]$, we get the following lower bound. 
    \begin{align*}
        \MoveEqLeft(1 - \eta - \eta \,\mbb{E}_{x \sim \mc{N}(b, 1)}[ \tanh'(ax) ] ) \,a + \eta \mbb{E}_{x \sim \mc{N}(b, 1)} \Big[ -\frac{1}{2} \tanh''( a x ) a^2 (x - b) + \tanh'(ax)(x^2 - bx) \Big] \\
        &\geq \frac{a}{20}(19 - 4 e^{\frac{a^2 - 2 a b}{2}}) + \frac{1}{20} \Big( -\frac{a^2}{2} \Big[ \exp(-b^2/2) + \abs{a - b} \exp\Big(\frac{a^2 - 2ab}{2}\Big) \Big] \\
        &\qquad\qquad\qquad\qquad\qquad\qquad\quad- b \exp(-b^2/2) - b \abs{a-b} \exp( \frac{ a^2 - 2ab }{2}) \Big) \geq 1\,.\qedhere
    \end{align*} 
\end{proof}

\begin{lemma}
\label{lemma:update-bound}
    For any $a, b > 0$ and $a \in [30, \frac{4b}{3}]$, the following holds. Define
    \begin{equation*}
        U(a, b) \triangleq \eta \mbb{E}_{x \sim \mc{N}(b, 1)} \Big[ \Big( \tanh (a x) - \frac{1}{2} \tanh''( a x ) a^2 + \tanh'( a x ) a x \Big) x \Big]  - \eta \mbb{E}_{x \sim \mc{N}(b, 1)} \sbr{\tanh'( a x ) a } - \eta a\,.
    \end{equation*}
    When the learning rate $\eta=\frac{1}{20}$, is given by, we have
    \begin{align*}
        \abs{U(a, b)} \leq \frac{a+b}{10}
    \end{align*}
\end{lemma}

\begin{proof} We upper bound each term in $U(a, b)$ and they apply triangle inequality to get the result. We start with $|\mbb{E}_{x \sim \mc{N}(b, 1) } \sbr{ \tanh''(a x) a^2 x }|$:
    \begin{align*}
    -\mbb{E}_{x \sim \mc{N}(b, 1) } \sbr{ \tanh''(a x) a^2 x } &= \frac{a^2}{8 \sqrt{2 \pi} } \int_0^\infty  x \sigma(2 a x) (1 - \sigma(2 a x))(2 \sigma(2 a x) - 1) \rb{ e^{-\frac{(x - b)^2}{2}} + e^{-\frac{(x + b)^2}{2}} } dx  \\ 
    &\leq \frac{a^2}{4 \sqrt{2 \pi} } \int_0^\infty  x e^{-2a x} e^{-\frac{(x - b)^2}{2}}  dx \\
    &\leq \frac{a^2}{4 \sqrt{2 \pi} } \int_{0}^\infty e^{-a x} x e^{-\frac{(x - b)^2}{2}} dx \\
    &\leq \frac{a^2}{2} e^{-\frac{b^2}{2}} + \frac{a^2}{2}  \abs{b - a} e^{-\frac{-2a(b - a) - a^2 }{2}} \\
    \end{align*}
    \begin{align*}
        \mbb{E}_{x \sim \mc{N}(b, 1)} [ \tanh'(a x) a x^2 ] & = \frac{1}{\sqrt{2\pi}} \int_0^\infty \tanh'(a x) a x^2 \rb{ e^{- \frac{(x - b)^2}{2} } + e^{- \frac{(x + b)^2}{2} } } dx \\
        & \leq a \int_0^\infty e^{-a x} x^2 e^{-\frac{(x - b)^2}{2}} dx \\
        & \leq a e^{\frac{ a^2 - 2 a b }{2}} \int_0^\infty x^2 e^{-\frac{(x - b + a)^2}{2}} dx \\
        &\leq 2a (a - b)^2 e^{\frac{ a^2 - 2 a b }{2}}
    \end{align*}
    \begin{align*}
        -\mbb{E}_{x \sim \mc{N}(b, 1)}[ a \tanh'(a x) ] & = -\frac{a}{\sqrt{2\pi}} \int_0^\infty \tanh'(a x)  \rb{ e^{- \frac{(x - b)^2}{2} } + e^{- \frac{(x + b)^2}{2} } } dx \\
        & \geq -a \int_0^\infty e^{-a x} e^{-\frac{(x - b)^2}{2}} dx \\
        & \geq - a e^{\frac{ a^2 - 2 a b }{2}} \int_0^\infty e^{-\frac{(x - b + a)^2}{2}} dx \\
        & \geq - 4a e^{\frac{ a^2 - 2 a b }{2}}\,.
    \end{align*}
    Now, using the fact that $\tanh'(x)$ and $-\tanh''(x)x$ are always positive, we have the following upper bound. 
    \begin{align*}
        \abs{U(a, b)} &\leq \eta\, \Big| \mbb{E}_{x \sim \mc{N}(b, 1)}\Big[ \big( \tanh (a x) - \frac{1}{2} \tanh''( a x ) a^2 + \tanh'( a x ) a x \big)\cdot  x \Big] \Big| \\
        &\qquad\qquad + \eta \abs{ a } + \eta\, \big| - \mbb{E}_{x \sim \mc{N}(b, I)} \sbr{\tanh'( a x ) a } \big| \\
        & \leq \eta \Big(2 b + a + \frac{a^2}{2} e^{-\frac{b^2}{2}} + \frac{a^2}{2}  \abs{b - a} e^{\frac{-2a(b - a) - a^2 }{2}} + 2a (b - a)^2 e^{\frac{ a^2 - 2 a b }{2}} + 2a e^{\frac{ a^2 - 2 a b }{2}} \Big)
    \end{align*}
    If $b \geq a$ and $a \geq 30$, then we have
    \begin{align*}
        \abs{U(a, b)}
        & \leq \eta \rb{2 b + a + 0.1 }
    \end{align*}
    If $b \leq a \leq \frac{4b}{3}$ and $a \geq 30$, then 
    \begin{align*}
        \abs{U(a, b)}
        & \leq \eta \rb{2 b + a + 0.1 }
    \end{align*}
    Using $\eta = 1/20$ and for any $a > 30$, we have
    \begin{align*}
        \abs{ U(a, b) } \leq \frac{a + b}{10}. 
    \end{align*}
\end{proof}

\subsection{Additional proofs for mixtures of two Gaussians}
\label{subsec:additional-2-mog}
\begin{lemma}[]
\label{lemma:small-expectation} Suppose $a, b > 0$ satisfy $a \in [30, \frac{4 b }{3}]$, then the following inequality holds: 
\begin{align*}
    | \mbb{E}_{x \sim \mc{N}( b , 1)} [ - 0.5 \tanh''( a  x ) a^2 + \tanh'( a x ) a x ] | \leq 0.01
\end{align*}
\end{lemma}
\begin{proof}
    We first show that $\mbb{E}_{x \sim \mc{N}( b , 1)} [ - 0.5 \tanh''( a  x ) a^2 ] > 0$ for any $a, b > 0$. 
    \begin{align*}
        \mbb{E}_{x \sim \mc{N}( b , 1)} [ - 0.5 \tanh''( a  x ) a^2 ] &= -0.5 a^2 \int_{-\infty}^\infty \tanh''(ax) \exp( -0.5( x - b )^2 ) dx \\
        = -0.5 & a^2 \int_{0}^\infty \tanh''(ax) ( \exp( -0.5( x - b )^2 ) - \exp( -0.5( x + b )^2 ) ) dx > 0 
    \end{align*}
    where the last inequality follows from $\exp( -0.5( x - b )^2 ) > \exp( -0.5( x + b )^2 )$ and $\tanh''(ax) < 0$ for $x > 0$. We can upper bound $\mbb{E}_{x \sim \mc{N}( b , 1)} [ - 0.5 \tanh''( a  x ) a^2 ]$ as follows: 
    \begin{align*}
        \mbb{E}_{x \sim \mc{N}( b , 1)} [ - \frac{1}{2} \tanh''( a  x ) a^2 ]  & \leq -\frac{1}{2} a^2 \int_0^\infty \tanh''(ax)  \exp( - \frac{1}{2} ( x - b )^2 ) dx \\
        & \leq a^2 \int_0^\infty \exp(-ax)  \exp( - \frac{1}{2} ( x - b )^2 ) dx \\
        & \leq a^2 \exp( \frac{1}{2}(a^2 - 2 a b) ) \int_0^\infty  \exp( -\frac{1}{2} ( x - b + a )^2 ) dx \\
        & \leq a^2 \exp( \frac{1}{2}(a^2 - 2 a b) )
    \end{align*}
    When $a \leq b$, by writing $a^2 - 2 a b = -2a(b - a) - a^2 \leq - a^2$, we have $\mbb{E} [ - \frac{1}{2} \tanh''( a  x ) a^2 ] \leq 0.005$ for $a \geq 30$. When $a \in [b, \frac{4b}{3}]$,  $a^2 - 2 a b = \leq - \frac{2b^2}{9}$, we have $| \mbb{E} [ - \frac{1}{2} \tanh''(a x) a^2 ] | \leq 0.005$. Similar to the $\mbb{E}_{x \sim \mc{N}( b , 1)} [ - \frac{1}{2} \tanh''( a  x ) a^2 ]$, we prove $\mbb{E}_{x \sim \mc{N}( b , 1)} [ \tanh'( a  x ) ax ] > 0$ and $\mbb{E}_{x \sim \mc{N}( b , 1)} [ \tanh'( a  x ) ax ] < 0.005$. Combining bounds for $| \mbb{E} [ \tanh'( a  x ) ax ] | $ and $| \mbb{E} [ - \frac{1}{2} \tanh''( a  x ) a^2 ] |$ using triangle inequality, we obtain the result. 
\end{proof}

\end{document}